\documentclass[reqno,10pt]{amsart}

\usepackage[pdftex]{pstricks}
\usepackage{amssymb}
\usepackage{stackrel}
\usepackage{dsfont}
\usepackage{graphicx,subfigure}

\def\Tr{\mbox{Tr}}
\newtheorem{theorem}{Theorem}[section]

\newtheorem{proposition}[theorem]{Proposition}

\theoremstyle{definition}
\newtheorem{definition}[theorem]{Definition}
\newtheorem{example}[theorem]{Example}
\newtheorem{xca}[theorem]{Exercise}

\theoremstyle{remark}
\newtheorem{remark}[theorem]{Remark}

\numberwithin{equation}{section}

\begin{document}

\title{Some Aspects of Operator Algebras in Quantum Physics}


\author{Andr\'es F. Reyes-Lega}
\address{Departamento de F\'{\i}sica, Universidad de los Andes,\\
Bogot\'a, Colombia}
\email{anreyes@uniandes.edu.co}
\thanks{I would like to thank all students who enthusiastically took part in the lectures on which these notes are based. Many of them contributed with useful criticisms and suggestions. I also want to  warmly thank  A.P. Balachandran, Andr\'es Vargas and Nicol\'as Escobar, for their very careful reading of the manuscript. Their feedback helped to improve the presentation of this document in a substantial way. Financial support from the Faculty of Science and the Vice Rectorate for Research of Universidad de los Andes, through project No. P13.700022.005, is gratefully acknowledged.}

\subjclass[2010]{81Qxx}

\keywords{Quantum Phase Transitions, Entanglement, Operator Algebras, Spin Chains}

\date{\today}

\begin{abstract}
Motivated by the sharp contrast between classical and quantum physics as probability theories, in these lecture notes I introduce the basic notions of operator algebras that are relevant for the algebraic approach to quantum physics. Aspects of the representation theory of C*-algebras will be motivated and illustrated in physical terms.  Particular emphasis will be given to explicit examples from the theory of quantum phase transitions, where concepts coming from strands as diverse as quantum information theory, algebraic quantum physics and statistical mechanics agreeably converge, providing a more complete picture of the physical phenomena involved.
\end{abstract}

\maketitle

\section{Introduction}
This notes represent the written version of lectures I gave in  mini-courses at  Universidade de Bras\'{\i}lia  (April 3-6, 2013), Universidad Central de Venezuela (May 23-27, 2016)  and at the Villa de Leyva Summer School ``Geometric, Topological and Algebraic Methods for Quantum Field Theory'' (July 15-27, 2013). They were mainly intended as an introduction to some aspects of operator algebras, emphasizing  the prominent role they play in quantum physics. As the audience consisted of students of both physics and mathematics at different stages of their studies, my choice was to focus on the most basic structures and examples, in the hope that a good grasp of these would  motivate them to go deeper into the subject.
Now, something that to a physicist may appear as completely familiar (as, say, an experimental set-up with polarizers, or the distinction between a classical and a quantum field) to a mathematician may not. The same could be said of the proof that the spectrum of any element in a  $C^*$-algebra is never empty: It is a standard result in analysis, but may look quite awkward to many physics students.
Therefore, the emphasis of these lecture notes will be on explaining \emph{why} certain mathematical structures may be useful for the study of quantum theory. This will be illustrated by means of several examples that include a discussion of bipartite entanglement, algebraic and geometric aspects of quantum phase transitions in spin chains, quasi-free states in fermionic systems and applications to quantum field theory.

Our starting point will be based on the sharp distinction between \emph{classical probability} and the probability theory inherent to quantum mechanics. This will provide a physical motivation to the various mathematical concepts we will be dealing with. After introducing the basic concepts about $C^*$-algebras, we will show how an algebraic approach to entanglement can lead to a resolution of certain discrepancies appearing when we deal with systems of identical particles.
Then we will focus on the study of certain specific models (quantum Ising and XY models) that turn out to be ideal in order to  illustrate how  entanglement, geometry and the theory of  CAR algebras\footnote{A special type of $C^*$-algebras used to model fermionic systems} are interrelated.

\subsection{Quantum correlations: Bell-type inequalities}
 Our first objective will be to understand what are the main structural differences  between \emph{classical} and \emph{quantum} physics, when regarded as probability theories. This will, by the way, provide a physical motivation for some of the mathematical notions we will consider in the next sections. Let us recall that
classical physics is usually modeled on a configuration (or phase) space, with a dynamics governed by, {\it e.g.}, Hamilton's principle. On the other hand, quantum mechanics is modeled on a Hilbert space. So the first issue we want to explore is: Why do we have to use Quantum Mechanics to describe the microscopic world?
Indeed, simple experiments with light polarizers make it clear that there is no way we can describe certain
phenomena using classical physics or, better said, classical \emph{probability}. Let us then explore some elementary polarization phenomena, following the presentation in~\cite{Maassen2010}.

 First let us recall that light is just made of electromagnetic waves, their behavior being governed
by Maxwell's equations:
\begin{eqnarray}
\label{eq:Intro-Maxwell}
                 \nabla \cdot E = \rho/\epsilon_0, &  & \nabla \cdot B = 0,\\
                 \nabla \times E = -\frac{\partial B}{\partial t}, &  & \nabla \times  B = \epsilon_0 \mu_0\frac{\partial E}{\partial t} +\mu_0 j,\nonumber
 \end{eqnarray}
where $E$ (resp. $B$) stands for the electric (resp. magnetic) field, $\rho$ for the  charge density and $j$ for the current density.
\begin{xca}
\label{ex:intro1} Show that Maxwell's equations (\ref{eq:Intro-Maxwell}) in vacuum ($\rho=0,\,  j=0$) lead to plane-wave solutions for the electromagnetic field propagating at a speed $c=1/\sqrt{\epsilon_0 \mu_0}$ and such that $E$ and $B$ are always perpendicular to each other and to the propagation direction.
What is the relation between the  \emph{intensity} of the wave (defined as $I=\| E \|^2$) and the energy content of the fields?
\end{xca}
In simple terms, a \emph{polarizer} is a filter that only allows the transmission of light waves which have a specific
polarization angle.  If we let unpolarized light go through an ideal polarizer, the intensity $I_1$ of the transmitted wave  will be found to be half the intensity $I_0$ of the incident wave: $I_1=1/2 \,I_0$ (Malus' law). After passing through the polarizer, the light is said to be linearly polarized. Let us now suppose  that we have a beam of linearly polarized light and we let it go through a second polarizer, such that its polarization axis has been rotated by an angle $\varphi$ with respect to the axis of the first polarizer. Then,  experiment tells us  that the intensity of the transmitted light will be $I_2 = \cos ^2 \varphi\, I_1$.
This is all fine if we are working  with classical electromagnetic waves, which are described by Maxwell's equations. In this case we just need to consider the projection of the field onto the direction singled out by the polarization axis.  But we know that light is actually made of photons and, if their number is small, we are led to regard the $\cos ^2 \varphi$ term as a kind of ``expectation value''.

In a very influential paper~\cite{Einstein1935}, Einstein, Podolski and Rosen presented a criticism of quantum theory in what is now known as
the \emph{EPR paradox}. This led to the development of alternative, so-called ``hidden variable" theories that aimed at
 explaining physical phenomena using classical probability models. It was only until Bell proved his famous inequalities,
 and Aspect's experiments proved the former were violated,  that the controversy could be resolved, showing that quantum
 theory provided the correct description of the phenomena.
An experimental set-up, of the type studied by Aspect,
 consists of a source (Ca atom) located in the middle that emits simultaneously a pair of photons. One of them
 goes to the right, the other to the left. There are two detectors, one at each extreme. There is also a
 polarizer in-between each detector and the source, as depicted in Fig.~\ref{fig:Intro-Ca-atom}.
\begin{figure}
\includegraphics[scale=0.675]{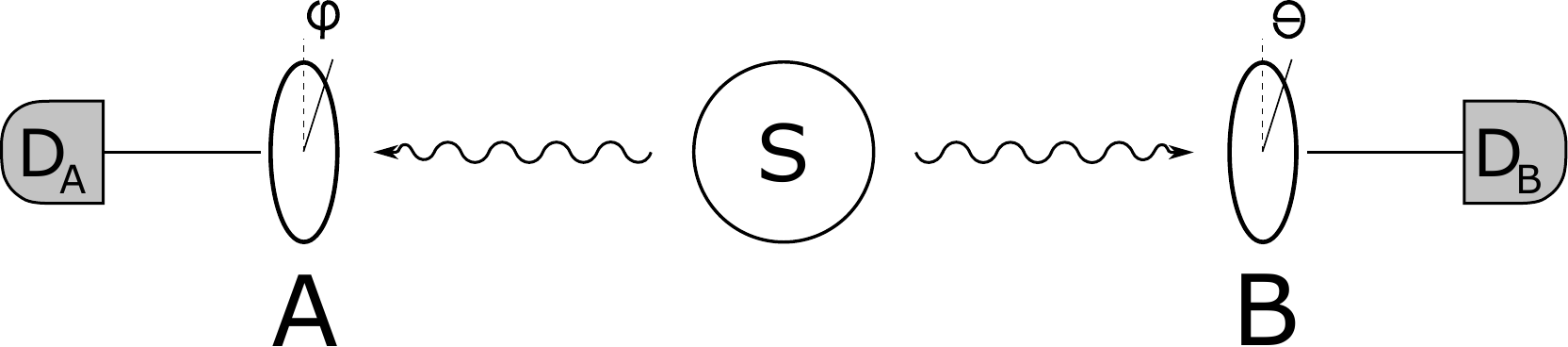}
\caption{A source S (Calcium atom) emits a pair of photons, each of which may be detected by detectors located in the extremes ($D_A$, $D_B$). There is a polarizer between the Ca atom and each of the detectors.}
\label{fig:Intro-Ca-atom}
\end{figure}
Let us now consider the following two propositions:
\bigskip
\\
$A=$``Left photon passes through (is detected by $D_A$) when polarizer's angle is
$\varphi$''.\medskip
\\
$B=$``Right photon passes through (is detected  by $D_B$) when polarizer's angle is $\theta$''.
\medskip
\\
Then, what we learn from experiment is that the \emph{joint probability} for both photons passing through the polarizers, thus being detected, is
\[
p(AB) = \frac{1}{2} \sin^2(\varphi -\theta).
\]
Before exploring why classical probability is in conflict with this result, let us recall how the principles of quantum mechanics allow us to predict it.

 Let $\mathcal H$ denote a Hilbert space which, for simplicity,  will be considered to be finite dimensional.
 In quantum mechanics, probability distributions are obtained from state vectors $|\Psi\rangle \in \mathcal H$  or, more generally, by density matrices $\rho$, that is, self-adjoint positive operators of trace one.
 On the other hand, observables are  described by self-adjoint operators. Let $A$ be such an observable. Let $\sigma(A)= \lbrace a_1, \ldots, a_N\rbrace$ denote its spectrum.
We can then consider its spectral decomposition
\begin{equation}
\label{eq:Intro-specA}
A= \sum_j a_j E_j,\;\;\; \sum_j E_j = \mathds 1,\;\;\;E_iE_j= \delta_{ij}E_j.
\end{equation}
If the state of the system is given by a density matrix $\rho$, the \emph{expectation value} of $A$ is defined as
\begin{equation}
\label{eq:Intro<A>}
\langle A\rangle_\rho:= \mbox{Tr}(\rho A).
\end{equation}
In the case of a \emph{pure} state, the density matrix is a rank-one projector, of the form $\rho=|\psi\rangle\langle\psi|$, $|\psi\rangle\in \mathcal H$, and so (\ref{eq:Intro<A>}) reduces to
 $\langle A\rangle_\psi=\langle \psi| A| \psi\rangle$.  In this way we obtain a probability distribution over $\sigma(A)$, with
\[
p(a_i)= \langle E_i\rangle_\rho.
\]
In fact, it follows  from (\ref{eq:Intro-specA}) and (\ref{eq:Intro<A>}) that $0\leq p(a_i) \leq 1$ and $\sum_i p(a_i)=1$.

Let us use this to give a mathematical description of the experiment described above. The polarization state of a photon can be described using a 2-dimensional Hilbert space.  Let $\lbrace |x\rangle, |y\rangle\rbrace$ denote an orthonormal basis, that can be used to describe, say, horizontal and vertical polarization states. Since we are considering a system consisting of two photons, the Hilbert space of the system can be taken to be  $\mathcal H= \mathds C^2 \otimes \mathds C^2$ (as the photons are supposed to be far from each other at the time of detection, the symmetrization postulate can be ignored).
The 2-photon state, as produced by the emission from an excited state of the Ca atom, can be described by the following state vector:
\begin{equation}
\label{eq:Intro-2-photon-state}
|\psi\rangle=\frac{1}{\sqrt 2}\left(|x\rangle \otimes |y\rangle -|y\rangle\otimes |x\rangle\right).
\end{equation}
For a polarizer with polarization axis pointing at an angle $\alpha$ we can use
 $P(\alpha)$, defined as the  \emph{projector} onto $\cos\alpha |x\rangle +\sin \alpha |y\rangle$.
Since we have two polarizers, we introduce the following projectors ({\it cf.} Fig.~\ref{fig:Intro-Ca-atom}):
\[
P_A(\varphi)= P(\varphi)\otimes \mathds 1,\;\;\;P_B(\theta)= \mathds 1 \otimes P(\theta).
\]
It is important to remark that these two operators correspond to compatible observables, in the sense that the measurement
of one of them does not affect the result of the other, \emph{i.e.}, they are commuting observables:
\[
[P_A(\varphi),P_B(\theta)]=0.
\]
\begin{xca}
Show that
\[
\langle\psi|P_A(\varphi)|\psi\rangle =   \langle\psi|P_B(\theta)|\psi\rangle= \frac{1}{2}.
\]
Also show that, for the \emph{joint} measurement of the two polarization states, one obtains
\begin{eqnarray}
\label{eq:Intro-PA-PB}
\langle\psi|P_A(\varphi)P_B(\theta)|\psi\rangle & =& \langle\psi|(\mathds 1-P_A(\varphi))(\mathds
1-P_B(\theta))|\psi\rangle\nonumber\\
&=&  \frac{1}{2} \sin^2(\varphi-\theta).
\end{eqnarray}
\end{xca}
The previous exercise shows that the predictions of quantum mechanics for this type of experiment are in accordance with what is actually measured in the laboratory.  In spite of its extreme simplicity, there is an intriguing feature of this result: By appropriately arranging the polarization angles $\theta$ and $\varphi$, we can obtain a \emph{total anticorrelation} for the joint measurements. As we will see, this is due to the fact that the state  (\ref{eq:Intro-2-photon-state}) is an \emph{entangled state}. That these kind of correlations cannot be obtained from a (local, realistic) classical theory is at the core of the original  EPR controversy. For an interesting discussion of these issues in the context of an actual experimental situation, we recommend~\cite{Guzman2015}.

Our immediate aim will  therefore be to understand where exactly classical probability fails at describing the results of
such experiments. For this purpose we will assume an approach to (classical) probability based on propositions and degrees
of plausibility, following Jaynes~\cite{Jaynes2003}.

In this setting, the objects to be considered are the following:
\begin{itemize}
\item[$\cdot$] A set of propositions:
$\lbrace A,B,C,\ldots\rbrace$, each one of which can take on (only) two values: true or false.
\item[$\cdot$]
Logical operations, that can be performed on the set of propositions:
\begin{itemize}
\item Conjunction, or logical product (AND): $AB$. It is true if and only if  both $A$ and $B$ true.
 \item Disjunction, or logical sum (OR): $A+B$. It is true if  at least one of them true.
 \item Negation (NOT): $\bar A$, with opposite truth value as $A$.
\end{itemize}
\end{itemize}
These logical operations are assumed to satisfy the defining rules of a \emph{Boolean algebra}:
\begin{itemize}
\item[$\cdot$] Idempotency: $AA= A\;\;,\;\; A+A=A$
 \item[$\cdot$] Commutativity: $AB = BA \;\;,\;\; A+B = B+A$
 \item[$\cdot$] Associativity:
 \begin{itemize}
\item[(A1)] $A(BC)= (AB)C \equiv ABC$ \item[(A2)] $A+ (B+C)= (A+B)+ C\equiv A+B+C$
 \end{itemize}

 \item[$\cdot$] Distributivity:
 \begin{itemize}
\item[(D1)] $A(B+C)= AB + AC$ \item[(D2)] $(A+ B)(A+C)= A + BC$
 \end{itemize}
\item[$\cdot$] Duality (De Morgan's laws): $\overline{AB}= \bar A + \bar B, \;\;\;\;\overline{A+B}= \bar A\,  \bar B$
\end{itemize}

\begin{xca}
Show that the proposition ``$A\Rightarrow B$ '' is equivalent to ``$A= A\,B$''.
\end{xca}

\begin{xca}
Show that (D2) follows from the other rules.
\end{xca}
Now, given a proposition $A$, a probability (or ``degree of plausibility'', \emph{cf.}~\cite{Jaynes2003}) $p(A)$ can be assigned to it, under the following basic assumptions:
\begin{itemize}
\item[I.]  Let $p(A)\in \mathds R$ denote the probability, or  degree of plausibility, of a given proposition $A$. Then we assume that:
\begin{itemize}
\item $p(f(A_1,A_2,\ldots,A_n))\in \mathds R$, for  any logical function $f$ of the propositions $A_1,A_2,\ldots,A_n$
\item Using the notation $p(A|B)$ for the conditional probability  that $A$ is true, given $B$ is true, we assume
that $p(A|C)>p(B|C)$ whenever $A|C$ more plausible than $B|C$.
\end{itemize}
\item[II.] Given a proposition $C$, let us suppose that we improve our state of knowledge, obtaining a
new proposition $C'$. If as a result  $A$ becomes \emph{more} plausible, \emph{i.e.}  $p(A|C')>p(A|C)$,
but $p(B|A C')=p(B|A C)$, then  $p(AB|C')\geq p(AB| C)$ should hold (this is dubbed the \emph{common sense} assumption by Jaynes ~\cite{Jaynes2003}).
\end{itemize}
From I and II above we may obtain, under very general assumptions (like consistency), the following two basic rules:
\begin{itemize}
\item Product Rule:\begin{eqnarray}\label{eq:prod-rule} P(AB|C) &=& P(A|BC)P(B|C)\\
                                                                  &=& P(B|AC) P(A|C)\nonumber \end{eqnarray}

\item Sum Rule:
 \begin{equation}\label{eq:sum-rule}
P(A|C) + P(\bar A|C) = 1.
 \end{equation}
\end{itemize}
The reader is invited to consult Jaynes for a comprehensive exposition of this point of view, including many illustrative derivations and examples.

Now, in order to return to our experiment, let us define the following ``coincidence'' function $f$, for two given
propositions $A$ and $B$:
\begin{equation}
f(A,B):= A \,B +\bar A\,\bar B.
\end{equation}
\begin{xca}
Given 4 propositions $A_1,A_2,B_1,B_2$, use the product and sum  rules to show that the following proposition is true:
\begin{equation}
f(A_1,B_1) \Rightarrow f(A_1,B_2)+ f(A_2,B_2)+ f(A_2,B_1).
\end{equation}
\end{xca}
Since  $p(A|X)\leq p(B|X)$ whenever $A\Rightarrow B$ holds (product rule) we obtain, from the previous exercise,
\medskip\\
``Bell's Inequality'':
\begin{equation}
\label{eq:Intro-Bell-CHSH}
 p(f(A_1,B_1)) \leq p(f(A_1,B_2))+ p(f(A_2,B_2)) + p(f(A_2,B_1)).
\end{equation}
Referring back to figure \ref{fig:Intro-Ca-atom}, let us consider the following propositions ($i,j=1,2$):
\medskip\\
\noindent
$A_i=$ ``Left photon passes through (is detected by $D_A$) when polarizer's angle is $\varphi_i$''.
\medskip\\
$B_j=$ ``Right photon passes through (is detected  by $D_B$) when polarizer's angle is $\theta_j$''.

The inequality (\ref{eq:Intro-Bell-CHSH}) should hold true for all choices of $\varphi_i,\theta_j$. But this then implies:
\[
\sin^2(\varphi_1-\theta_1) \leq \sin^2(\varphi_1-\theta_2) + \sin^2(\varphi_2-\theta_2)  + \sin^2(\varphi_2-\theta_1).
\]
For the choice $\varphi_1=0,\;\varphi_2= \frac{\pi}{3},\; \theta_1=\frac{\pi}{2}, \; \theta_2=\frac{\pi}{6}$, this
means:
\[
\sin^2\left(\frac{\pi}{2}\right) \leq \sin^2\left(\frac{\pi}{6}\right) + \sin^2\left(\frac{\pi}{3}- \frac{\pi}{6}\right)  + \sin^2\left(\frac{\pi}{3}-\frac{\pi}{2}\right),
\]

\[
\sin^2\left(\frac{\pi}{2}\right ) \leq \sin^2\left(\frac{\pi}{6}\right) + \sin^2\left(\frac{\pi}{6}\right)  + \sin^2\left(-\frac{\pi}{6}\right),
\]
\[
1\leq \frac{3}{4}.
\]
This contradiction means that the assumptions we have considered above (which are at the basis of classical probability) do not apply in the quantum realm. So, definitely, quantum theory leads to a very different type of probability theory. In the next section we will explore some of the more notorious differences between classical and quantum probabilities.

\subsection{Classical versus quantum probability}
Until now we have avoided any mention of propositions in terms of set theory, which is the basis of the
 Kolmogorov axiomatics. In the finite dimensional case, at least, they turn out to give equivalent structures.
 Let us then consider the main properties of a (classical) probability theory, formulated in terms of set theory and,
for simplicity,  in the finite dimensional context. Let us consider a finite set $\Omega=\lbrace x_1, x_2,\ldots,x_N\rbrace$, regarded here as the \emph{sample space}. The   \emph{event space}  $\mathcal E$ is  a certain collection of subsets of $\Omega$, that  must contain the empty set  and be closed under
complements and unions.

The set of  events forms a Boolean algebra, under the standard set theoretic operations:
            \begin{itemize}
               \item[$\cdot $]  AND: $A \cap B$,
               \item[$\cdot $] OR: $A \cup B$,
               \item[$\cdot $]  NOT: $A^c= \Omega\setminus A$.
             \end{itemize}
A probability distribution is then defined as a map $p:\mathcal{E}\rightarrow [0,1]\subset \mathds R$, such that
\begin{itemize}
 \item[(i)]  $p(\emptyset)= 0,\;\;$ $p(\Omega)=1$,
\item[(ii)] $p(\bigcup_{k=1}^n A_k)=\sum_{k=1}^n p(A_k),$ for $A_1,A_2,\ldots,A_n\in \mathcal{E}$ pairwise disjoint.
\end{itemize}
\begin{remark}
Notice that the space of probability distributions is a convex set, for if $p_1$ and $p_2$ are two probability
distributions, then $\lambda p_1+ (1-\lambda)p_2$
is again a probability distribution, provided $0\leq \lambda \leq1$. But not only is the space of probability distributions a convex set, it is a special kind of convex set: a \emph{simplex}. To see this, all we have to do is to consider the following ``extremal'' distributions, $p^{(1)}, p^{(2)},\ldots, p^{(N)}$, defined by
   \[
      p^{(i)}(x_j):= \delta_{ij}.
   \]
It is clear that any $p$ can be written as a \emph{unique} convex combination of these extremal distributions, as we have
\[
p= \sum_i \lambda_i p^{(i)},
\]
with $\lambda_j=p(x_j)$.
\end{remark}
A random variable (or ``observable'') is a function $f:\mathcal E \rightarrow \mathds R$. We call the space of all such functions $\mbox{Obs}(\Omega)$.
Notice that $\mbox{Obs}(\Omega)$ forms a commutative algebra.

Summarizing, in the finite-dimensional case we have the following structure:
\begin{itemize}
\item[$\cdot$]  The event space ($\mathcal E\subseteq \mathcal P(\Omega)$) forms a Boolean algebra.
\item[$\cdot$]  The
space of probability distributions is \emph{convex} and furthermore  has the structure of a simplex: Every probability
distribution $p$ can be uniquely
 written as convex combination of   ``extremal" distributions.
\item[$\cdot$] The space of  observables (random variables) has the structure of a commutative algebra.
\end{itemize}
\begin{remark}
In cases where the sample space is not a finite set (like in classical mechanics) we need a suitable generalization  of the above definitions.
This is afforded by measure theory, in the following way: Sample and event spaces are now replaced by  a measurable space  $(\Omega, \mathcal E)$,
 where  $\mathcal E$ is a  suitable $\sigma$-algebra of $\Omega$.
Probability distributions are then defined as \emph{normalized, positive measures}. In particular, we can now handle countable additivity:
\[
p\left(\bigcup_{k=1}^\infty A_k\right)=\sum_{k=1}^\infty p(A_k), \;\;\;\mbox{for}\;\; A_j\;\;\mbox{pairwise disjoint}.
\]
Finally, the space of observables is still a commutative algebra, for the product and sum  of two measurable functions
is a measurable function. A similar remark applies to the convex structure of the space of probability distributions.
\end{remark}
How do these structures show up in classical physics?  In classical mechanics, for example,  the dynamics of a system can be described in
terms of canonical variables ``position'' ($q_i$) and ``momentum'' ($p_j$), that give rise to the \emph{phase space} of the system. It
follows from Hamilton's variational principle that if  $H(q,p)$ is the Hamiltonian of the system, the canonical variables will evolve
along solutions to Hamilton's equations:
\begin{eqnarray}
\label{eq:Intro-Hamilton-eqns}
\frac{dp_i}{dt} &=& -\frac{\partial H}{\partial q_i}, \nonumber\\
\frac{dq_i}{dt} &=& \;\;\,\frac{\partial H}{\partial p_i},\;\;\; i=1,\ldots, n.
\end{eqnarray}
Usually the phase space is the cotangent bundle $T^*Q$ of some configuration space $Q$. The canonical variables $(q,p)$ are then
local coordinates on $T^*Q$ and the \emph{observables} are smooth functions on phase space: $C^\infty(T^*Q)$.
Since time evolution is given by Hamilton's equations or, at the level of observables by
 $\frac{df}{dt} = \lbrace H,f \rbrace$, with $\lbrace \cdot,\cdot\rbrace$ denoting the Poisson bracket, one
 would think that there is no much space for a probabilistic model here since, given a set of initial conditions, by solving the
  equations of motion we are able, in principle, to predict the position and momenta of all the particles.  But, as we learn from
   classical statistical physics, for  $n >> 1$ ($n\sim 10^{23}$) we certainly need a statistical approach! Therefore we are forced
   to introduce probability distributions, expressed in terms of probability densities $\rho(q,p)$, so that the average value of an
    observable $f\in C^\infty( T^*Q)$ is given by
 \[
\langle f \rangle = \int  f(q,p) \rho(q,p) d\mu,
\]
where  $d\mu=d^n qd^np$ is the Liouville measure on phase space. For instance, for the \emph{canonical ensemble} (used
to describe a subsystem embedded in a thermal bath at temperature $T$), we have $\rho(q,p) \propto e^{-H(q,p)/k_B T}$.
Now we can return to the case of a few point particles and interpret the state of the system at a given time, usually
defined as just a point $(q_0,p_0)$ in phase space, as an extremal probability distribution (a ``pure state''), for
which the probability density is just a Dirac delta distribution:
 $\rho(q,p)= \delta(q-q_0, p-p_0)$. In any case, we see that classical physics can be regarded as a special type of probability
 theory, where the three properties of classical probability highlighted above still hold, \emph{i.e.},
 every classical system of point particles can be understood as a probability theory, in the
Kolmogorov sense.

What about Quantum Mechanics? As discussed above, probability distributions are obtained from state vectors
$|\Psi\rangle$ in a  Hilbert space $\mathcal H$,  or more generally from density matrices, $\rho$, whereas observables
are described by self-adjoint operators. Consider the spectral decomposition of a self-adjoint operator $A$, as in
(\ref{eq:Intro-specA}). Assuming a non-degenerate, discrete spectrum, we notice that the projectors $E_i$, having as
spectrum the set $\lbrace 0, 1 \rbrace$, can be regarded as ``indicators of events''. Now, there is a correspondence
between projections $E:\mathcal H \rightarrow \mathcal H$ and subspaces $V\subset \mathcal H$. The partial order  that
we naturally obtain by inclusion then gives rise to the structure of an ``orthocomplemented lattice of proposition'',
where the corresponding operations are defined as
    \begin{eqnarray*}
    \mbox{OR:} \;\; V_1\vee V_2 &\longleftrightarrow& \mbox{span}(V_1,V_2)\\
    \mbox{AND:}\;\; V_1\wedge V_2 &\longleftrightarrow& V_1 \cap V_2\\
    \mbox{NOT:} \;\;\;\;\;\;\;V' \;\;\; &\longleftrightarrow&  V^\perp.
    \end{eqnarray*}
An important feature of this system is that it \emph{does not} give rise to a  Boolean algebra structure. The reason for this is that the orthogonal
complement is not the only possibility for a complement in this lattice. This in turn implies the breakdown of the distributive law, which is part of the definition of a Boolean algebra.
As in the classical case, the state space (here the space of density matrices) is a convex space. But, in contrast to
the classical case, this space is not a simplex. A consequence of this is that the representation of a  density matrix as
 a convex sum of  pure states is (highly) non-unique. Finally, in quantum theory the space of
observables forms a non-commutative algebra, in contrast to the classical case, where the algebra is commutative.

It is now clear that, in principle, both approaches (classical and quantum) can be considered in order to describe physical
phenomena in probabilistic terms.  In fact, EPR-like arguments are in favor of a ``local-realistic'' point of view according to
which even quantum phenomena should be explained in terms of classical probability.
That this is not the case is proved by Bell's  inequalities and their violation, experimentally verified by Aspect in the 80's.
So quantum theory can be regarded as a kind of ``non-commutative'' probability theory. One of the points of these lectures is
that, when formulated in the language of operator algebras, both quantum and classical physics can be described in a unified way.
We will also take advantage of this formulation to discuss  a similar phenomenon, occurring in topology and geometry: In
non-commutative geometry~\cite{Connes1995} the generalization of topological/geometric notions to the non-commutative setting
has been mainly achieved by first expressing them in algebraic terms and then realizing that dropping the commutativity assumption
 allows for  vast  generalizations.



\section{Aspects of operator algebras  in quantum physics}
\label{sec:2}


\subsection{Observable algebras and states}
\label{sec:algebras-and-states}
Let $\mathcal H$ be a separable Hilbert space. Recall that, given a linear operator $T:\mathcal H\rightarrow \mathcal H$, $T$ is
said to be bounded if there is some $C>0$ such that
\[
\|T(x)\| \leq C \| x\|
\]
for all $x$ in $\mathcal H$, where the norm is the one induced by the inner product: $\|x\|^2= (x,x)$. If $T$ is a
bounded operator, we define its norm as follows:
\[
\| T \|:=\; \stackrel[x \neq 0]{\mbox{sup}}{ } \frac{\|T(x)\|}{\|x\|}.
\]
One then checks that for $T$ and $S$ bounded the inequalities
\begin{eqnarray*}
\| T + S\| \leq \| T \| +\|S\|,\\
\| T \, S \| \leq \| T \| \, \| S \|,
\end{eqnarray*}
are satisfied. From the completeness of $\mathcal H$ it follows that the space
\[
\mathcal B (\mathcal H)= \lbrace T: \mathcal H \rightarrow \mathcal H\,|\,T\; \mbox{is linear and bounded}\rbrace
\]
is a complete normed space. It is also an algebra and has an involution ``$*$'' given by the adjoint: $T^*:=T^\dagger$. We thus may call  $\mathcal B (\mathcal H)$ the ``$*$-algebra of bounded operators on $\mathcal H$''
\begin{xca}
Show that for $T\in \mathcal B (\mathcal H)$ we have:
\[
\|T^* T\| = \|T\|^2.
\]
\end{xca}
We now abstract these notions and make them independent of any underlying Hilbert space. As we will see below, a lot will be gained from this, since we will then be able to study \emph{representations} of the (abstract) operator algebras, and the equivalence/inequivalence of these representations will have a deep physical meaning.
\begin{definition}A \emph{Banach space} is a normed vector space $(V,\|\cdot\|)$ which is complete\footnote{ ``all Cauchy sequences converge''} with respect to the (metric induced by) $\|\cdot\|$.
\end{definition}
\begin{definition}
A \emph{Banach algebra}  is a Banach space
  $(\mathcal A,\|\cdot\|)$ which is also an algebra,  with the property that
$\|ab\|\leq \|a\|\|b\|,\;\;\forall a,b \in \mathcal A$ (which in turn implies that multiplication is a continuous
operation).
\end{definition}
\begin{definition}
An \emph{involution} on a (complex) algebra $\mathcal A$ is a map $*:\mathcal A \rightarrow \mathcal A$ such that, for any
$a,b\in \mathcal A$ and $\mu,\nu \in \mathds C$:
\begin{itemize}
\item[i.] $(\lambda a + \mu b)^* = \bar \lambda a^* + \bar \mu b^*$
\item[ii.] $(ab)^*= b^* a^*$
\item[iii.] $(a^*)^*= a$
\end{itemize}
Notice that this is, basically, an abstraction of the adjoint operation on $\mathcal B (\mathcal H)$.
\end{definition}
\begin{definition} A \emph{$C^*$-algebra} is a Banach $*$-algebra $(\mathcal A,\|\cdot\|, *)$ with the
 fundamental property
\begin{equation}
\label{eq:basics-cstar-ppty}
\|a^* a\| = \|a\|^2.
\end{equation}
\end{definition}
\begin{example}\label{example:Hausdorff}
Let $M$ be a compact, Hausdorff topological space and set $\mathcal A= C(M)$, the space of continuous complex functions
on $M$ and, for $f\in C(M)$, set $f^*(x)=\overline{f(x)}$, and $\|f\|=\stackrel[x\in M]{\mbox{sup}}{ }|f(x)|$. Then
$(C(M), \|\cdot\|,*)$ is a $C^*$-algebra.
\end{example}
\begin{example} Let $\mathcal H$ be a Hilbert space. Then, as we expect, $(\mathcal B(\mathcal H), \|\cdot\|, *)$ is a $C^*$-algebra, where the norm is the operator norm and the involution is given by the adjoint operation.
\end{example}
Later we will see that the list of examples of $C^*$-algebras is basically exhausted by the previous two examples, a remarkable fact. A motivation  to work with algebras of bounded operators comes from physics since, as can be shown, the \emph{canonical commutation relations} (CCR)
\[
[\hat q ,\hat p] =i\hbar \mathds 1
\]
cannot be implemented by means of bounded operators $\hat q$ and $\hat p$. Then, although these commutation relations have a very clear meaning from the physical point of view, mathematically they correspond to \emph{unbounded} operators and hence issues like self-adjointness, domains, etc. come into play, which make them more difficult to deal with. One can, nevertheless, replace the CCR by their exponentiated (or Weyl) form, as follows. First we define operators $U(a)$ and $V(b)$,
for $a,b\in \mathds R$, acting on wave functions as follows:
\begin{eqnarray}
\label{eq:U(a)V(b)def}
\left(U(a) \psi\right)(x) & := & \psi (x -\hbar a),  \\
\left(V(b) \psi\right)(x) & := & e^{-ib x}\psi (x).\nonumber
\end{eqnarray}
By Stone's theorem, it follows that $U(a)= e^{-ia\hat p}$ and $V(b)=e^{-ib \hat q}$.
\begin{xca}
\label{ex:Weyl-CCR} Show that the operators $U(a)$  and $V(b)$ satisfy the following commutation relations:
\begin{eqnarray}
\label{eq:Weyl-CCR}
U(a_1)U(a_2) &=& U(a_1+a_2),\nonumber\\
V(b_1)V(b_2) &=& V(b_1+b_2),\\
U(a)V(b)&=& e^{i\hbar ab} V(b) U(a).\nonumber
\end{eqnarray}
\end{xca}
As we will see later on, these operators give rise to a $C^*$-algebra. In fact, we will see that to any symplectic vector space we can canonically associate a $C^*$-algebra (its Weyl $C^*$-algebra). The unitary representations of these algebras, in the case
of infinite dimensional symplectic vector spaces, play a prominent role in the study of quantum field theory on curved spacetimes.
\begin{example}
Consider a (complex, involutive, with unit $\mathds 1$) algebra $\mathcal A$ generated by elements
$a_1,a_2,\ldots, a_n, \mathds 1$, subject to the following \emph{canonical anticommutation relations (CAR)}:
\begin{equation}
\label{eq:basics-CAR1}
a_i a_j^* + a_j^* a_i = \delta_{ij} \mathds 1,\;\;\;a_i a_j + a_j a_i =0.
\end{equation}
If we want to define a $C^*$ norm on this algebra, it has to be such that the $C^*$-property (\ref{eq:basics-cstar-ppty})
is satisfied.  But from (\ref{eq:basics-CAR1}) we obtain
$(a_i^*a_i)^2 = a_i^*a_i$, which along with  (\ref{eq:basics-cstar-ppty}) implies $\|a_i^*a_i\|^2= \|a_i^* a_i\|$. The only option we have, then, is to define $\|a_i\|=1$ for every $i\in \lbrace 1,\ldots,n\rbrace$.
\end{example}
\begin{xca}
\label{ex:CAR-algebra-finite-dim} Generalize the previous example to the case of a finite dimensional Hilbert space
$\mathcal H\cong \mathds C^n$, in the following way. Let $\langle\, \cdot\, | \,\cdot\,\rangle$ denote the inner product. For
any pair of vectors, $u,v\in \mathcal H$, consider  generators $a(u), a(v)$ satisfying the following relations (CAR):
\[
\lbrace a(u), a(v)^*\rbrace= \langle u|v\rangle \mathds 1,\;\;\; \lbrace a(u), a(v)\rbrace =0.
\]
where $\lbrace A,B \rbrace \equiv AB +BA$.
How must  $\|a(u)\|$  be defined in this case?
\end{xca}

The (mathematical) notion of \emph{state} for a $C^*$-algebra is also very close to the notion of quantum state. To appreciate this, consider a physical system modeled by the algebra $\mathcal B(\mathcal H)$ of bounded operators on a Hilbert space $\mathcal H$. As mentioned in the previous section, in physics we distinguish between two kinds of states:
\begin{itemize}
\item[$\cdot$] Pure states: These are described by normalized vectors $|\psi\rangle\in \mathcal H$  (actually, by rays on Hilbert space). If $A$ is an observable, then its expectation value is defined as
    \[ \langle A\rangle_\psi := (\psi, A\psi)\equiv \langle \psi |A |\psi\rangle.\]
\item[$\cdot$] Mixed states: They are described by operators $\rho:\mathcal H\rightarrow \mathcal H $ such that $\rho^\dagger=\rho>0$, $\Tr \rho =1$. The expectation value is given, in this case, by
    \[
    \langle A\rangle_\rho := \Tr_{\mathcal H}(\rho A).
    \]
\end{itemize}
    \begin{xca}
    Show that the condition for a mixed state (described by a density matrix $\rho$) to be pure is  $\rho^2=\rho$. Thus, all states can be described by density matrices. The space of all density matrices is naturally a convex space. The extremal elements then turn out to be the pure ones. Check this assertion in the simple case $\mathcal H=\mathds C^2$.
    \end{xca}
We therefore see that a state can be regarded as a mapping $A\mapsto \langle A\rangle$ from the observable algebra to the complex numbers. It must be possible, then, to express the properties $\rho>0$ and $\Tr \rho =1$ only in terms of this mapping. This leads us to the general definition of state for a $C^*$-algebra.
\begin{definition}
\label{def:basics-state}
Let $\mathcal A$ denote a (unital) $C^*$-algebra, with unit $\mathds 1$. A positive linear functional
\[
\omega: \mathcal A\rightarrow \mathds C
\]
such that $\omega (\mathds 1)= 1$ is called a \emph{state}. We will denote the set of all states on $\mathcal A$ with $\mathcal S_{\mathcal A}$.
\end{definition}
\begin{remark}
In this context, positivity of $\omega$ means that $\omega(a^*a) \ge 0$ for all $a\in \mathcal A$.
\end{remark}
\begin{remark}
If $\mathcal A$ is not unital, we can replace the normalization condition by the condition $\|\omega\|=1$. For unital
algebras, both definitions coincide.
\end{remark}
Having in mind the example of $\mathcal B (\mathcal H)$, it then makes sense to have a notion of spectrum for an abstract $C^*$-algebra. This leads, in turn, to the spectral theorem for $C^*$-algebras~\cite{Werner2000}.
\begin{definition}
Let $a$ be an element of a unital  $C^*$-algebra $\mathcal A$. Then its \emph{spectrum} is defined as the set
\[
\sigma(a):= \lbrace \lambda\in \mathds C\,|\, (\lambda \mathds 1 - a )\;\;\mbox{is not invertible}\rbrace.
\]
\end{definition}
\begin{definition}
Let $a$ be an element of a unital  $C^*$-algebra $\mathcal A$. Then its \emph{spectral radius} is defined as
\[
\rho(a):= \mbox{sup}\lbrace |\lambda|\,|\, \lambda \in \sigma(a)\rbrace.
\]
\end{definition}
\begin{xca} Show that the space $M_n(\mathbb C)$ of $n\times n$ matrices with complex entries is a $C^*$-algebra
with respect to the norm defined as
\[
\| A\|_{\tiny\mbox{max}} := |\lambda_{\tiny\mbox{max}}|^{1/2},
\]
where $\lambda_{\tiny\mbox{max}}$ is the eigenvalue of $A^\dagger A$ with the largest absolute value. Show also that this norm
coincides with the operator norm.
\end{xca}
\begin{remark}
The previous exercise illustrates the remarkable fact that, for a $C^*$-algebra ($\mathcal A,\|\cdot\|, *$), the norm
of any element normal $a\in \mathcal A$ coincides with its spectral radius:
\begin{equation}
\label{eq:spectral-radius-formula}
\|a\| = \sqrt{\rho(a^* a)}.
\end{equation}
This is an important fact, as it links topological properties of the algebra with algebraic ones. In particular, it
implies that the norm in a $C^*$-algebra is \emph{unique}.
\end{remark}

The following two properties will be very useful:
\begin{xca}\label{ex:Cauchy-Schwarz}
Let $\phi: \mathcal A\rightarrow \mathds C$ be a positive linear functional on a $C^*$-algebra. Prove that
\begin{itemize}
\item[($i$)] $\phi (a^* b )= \overline{\phi(b^* a)}$.
\item[($ii$)] $|\phi(a^* b)|^2 \leq \phi (a^* a) \phi (b^* b)$
(``Cauchy-Schwarz inequality'')
\end{itemize}
Hint: Consider $\phi( (\lambda a+b)^* (\lambda a+b))$ as a quadratic form on $\lambda$. As an alternative, you may
try deriving these relations just by thinking of $\phi(a^*b)$ as  ``$\langle a|b\rangle$''.
\end{xca}

Another consequence of positivity is continuity:
\begin{proposition} Let $\phi$ be a positive linear functional on a unital $C^*$-algebra $\mathcal A$. Then $\phi$ is a continuous
linear functional, and $\|\phi\|=\phi(\mathds 1)$.
\end{proposition}
\begin{proof}[Proof (\emph{cf.}~\cite{Gracia-Bond'ia2001})]
If $a$ is positive ($a=a^*$ and $\sigma(a)\subseteq \mathbb R_+$) then we have, from the spectral radius formula, that
$\sigma(a)\subseteq [0,\|a\|]$. This means that $\|a\|-\lambda \ge 0$ for all $\lambda\in \sigma(a)$. This can be
restated as follows. Define
\begin{eqnarray*}
f: \sigma(a) & \longrightarrow & \mathbb C \\
\lambda &\longmapsto & f(\lambda):= \|a\| -\lambda.
\end{eqnarray*}
Then $f\in C(\sigma(a))$ and $f\ge 0$. But then it follows (from the continuous functional calculus) that
$\sigma(f(a))=f(\sigma(a))$, so that $f(a)\ge 0$  or, in other words, that  $\|a\|\mathds 1 -a$ is a positive operator. Since $\phi$ is positive and
linear, this implies $\phi(a)\le \phi(\mathds 1) \|a\|$. Now, for an arbitrary $\alpha\in \mathcal A$, apply the
Cauchy-Schwarz inequality from exercise \ref{ex:Cauchy-Schwarz} with $a=\mathds 1$ and $b=\alpha$, and use the $C^*$
property of the norm to obtain $\phi(\alpha)\le \phi(\mathds 1) \|\alpha\|$. It follows that $\phi$ is continuous, with
$\|\phi\| \leq \phi(\mathds 1)$. But we also have
\[
\phi(\mathds 1)= |\phi(\mathds 1)|\leq \|\phi\| \|\mathds 1\| =\|\phi\|,
\]
so that $\|\phi\|=\phi(\mathds 1)$.
\end{proof}
Thus, for unital $C^*$-algebras we conclude that  the set of states over $\mathcal A$ is a convex subset of the continuous dual
$\mathcal A^*$ of $\mathcal A$.

We gather from all this that it makes sense to formulate quantum physics in terms of $C^*$-algebras. The general philosophy will
 be to describe a given physical system in terms of its ``algebra of observables'', which will be here taken to be a $C^*$-algebra
  $(\mathcal A, \|\cdot\|,*)$. Then we will consider the following two ``dual'' notions:
\begin{itemize}
\item[$\cdot$] A \emph{quantum state} will be defined as a state (in the sense of definition \ref{def:basics-state}) on the algebra
 $\mathcal A$.
\item[$\cdot$] An \emph{observable} will be an element of $\mathcal A$.
\end{itemize}
Furthermore, the ``pairing'' between an observable $a$ and a state $\omega$ will be given the physical interpretation of an \emph{expectation value}, and we will write
\[
\langle a\rangle_\omega := \omega (a).
\]
 Regarding the distinction between pure and mixed states, we can now call a state \emph{pure} if it cannot be written
 as a convex
 combination of other states; otherwise we will call it a \emph{mixed} state.

\subsection{The Gelfand-Naimark theorem and the GNS construction}
Before discussing physical applications of the notions introduced in the previous section, we will take the opportunity
to briefly review the characterization of $C^*$-algebras due to Gelfand, Naimark and Segal. The characterization of
commutative $C^*$-algebras (the Gelfand-Naimark theorem) is of fundamental importance, in particular because it leads
to the notion of a \emph{``noncommutative topological space''}~\cite{Connes1995}. Also, the so-called GNS-construction
discussed below is relevant not only because of its role in the characterization problem for noncommutative
$C^*$-algebras, but also because of its striking  consequences for the study of  quantum systems (equilibrium states,
symmetry breaking, inequivalent vacua, among others)~\cite{Haag1996,Bratteli1997}.

\begin{definition}
A \emph{character} of a $C^*$-algebra $\mathcal A$ is a $*$-homomorphism $\mu:\mathcal A\rightarrow \mathbb C$. Let us
denote with $\mathcal M_{\mathcal A}$ the set of all characters of $\mathcal A$.
 \end{definition}
Now, it is easy to see that if $\mu$ is a character, then $\mu(a)\in \sigma(a)$. This, together with the spectral
radius formula (\ref{eq:spectral-radius-formula}), implies that $\|\mu\|\leq
1$, so that $\mathcal M_{\mathcal A}\subset \mathcal A^*_1$, where
\[
\mathcal A^*_1:= \lbrace \phi\in \mathcal A^* \,|\, \|\phi\| \leq 1\rbrace
\]
is the unit ball in the dual space. Since (by  the Banach-Alaoglu theorem~\cite{Rudin1991}) this ball is compact in the
weak-$*$ topology (\emph{i.e.} the one induced by the family of seminorms on $\mathcal A^*$:
$p_\alpha(\phi):=|\phi(\alpha)|$, $\alpha\in \mathcal A$) then  $\mathcal M_{\mathcal A}$ becomes a \emph{compact
topological space} in the subspace topology inherited from $\mathcal A^*$.
 This is quite interesting in view of example
\ref{example:Hausdorff} as this is providing the  converse statement. In fact, for  a commutative,
unital $C^*$-algebra $\mathcal A$, we define the \emph{Gelfand transform} $\mathcal G$ as the map
\begin{eqnarray*}
\mathcal G: \mathcal A &\longrightarrow & C(\mathcal M_{\mathcal A})\\
a & \longmapsto & \hat a,
\end{eqnarray*}
where $\hat a(\mu):=\mu(a)$.
\begin{theorem}[Gelfand-Naimark~\cite{Gelfand1943,Connes1995,Gracia-Bond'ia2001}] \label{thm:Gelfand-Naimark}The Gelfand
 transform is an isometric $*$-isomorphism $\mathcal A\cong C(\mathcal M_{\mathcal A}).$
\end{theorem}
Rather than reviewing the proof of this theorem, we will consider some basic examples that should provide a good
intuition  about the correspondence between  commutative $C^*$-algebras and (locally compact, Hausdorff) topological
spaces.
\begin{example} Consider the unital, commutative $C^*$-algebra $\mathcal A$ generated by a unitary element  $u$
(\emph{i.e.} $u^* u= u u^* = \mathds 1$), with norm fixed by the condition~\cite{Paschke2001} that
\[
\|\mathds 1 +e^{-i\alpha} u\|= 2,
\]
for all $\alpha\in [0,2\pi)$. Theorem \ref{thm:Gelfand-Naimark} states that this algebra must be the  function algebra
of a topological space which, as a set, is precisely the set $\mathcal M_{\mathcal A}$ of all characters of $\mathcal
A$. Since the algebra is generated by $u$, a character $\mu$ is fixed by its value on $u$. But $|\mu(u)|^2=\mu(u^* u)=1$, and
so every character is of the form $\mu(u)=e^{i\theta}$, \emph{i.e.} $\mathcal M_{\mathcal A}\subseteq S^1$. In order to
see that actually $\mathcal M_{\mathcal A}= S^1$, we proceed as in~\cite{Paschke2001} by noticing that
  from example \ref{example:Hausdorff} and from the uniqueness of the norm, it follows that
$\|a\|=\mbox{sup}_{\mu\in \mathcal M_{\mathcal A}}|\mu(a)|$, so that we must have
\[
\stackrel[\mu\in \mathcal M_{\mathcal A}]{\mbox{sup}}{ }|1+e^{-i\alpha}\mu(u)|=2.
\]
But this can only happen if for every $\alpha\in [0,2\pi)$ there is a character $\mu$ with $\mu(u)=e^{i\alpha}$.
\end{example}
\begin{example} The 2-sphere $S^2$ can also be easily characterized in terms of its algebra of continuous  functions.
 In this case, the algebraic structure is well-known, as $C(S^2)$ is generated by the spherical harmonics $Y_{lm}$,
 their product being given as in the usual Clebsch-Gordan decomposition. More details on how  $S^2$ can be obtained
 as the character space of such an algebra (this involves an appropriate definition of the norm, as in the previous example)
  can be found in~\cite{Paschke2001}.
\end{example}
\begin{example} Take $\mathcal A= C(S^2)$ as above (with the sup norm) and consider the subalgebra $\mathcal A_+$ consisting
of all even functions on $\mathcal A$, that is, functions $f$ such that $f(-x)=f(x)$. Then   $\mathcal A_+\cong
C(\mathbb R P^2)$. This fact has been exploited in studies of the spin-statistics connection in quantum mechanics~\cite{Paschke2001,Papadopoulos2004,Papadopoulos2010,Reyes-Lega2011}.
\end{example}
\begin{remark}
The isomorphism provided by theorem \ref{thm:Gelfand-Naimark} is actually an equivalence between the categories of
commutative $C^*$-algebras and locally  compact, Hausdorff topological spaces. This means, among other things,  that
all information regarding the topology of $M$ is encoded in the algebra $C(M)$: An open set in $M$ can be equivalently
described in terms of a ideal in $C(M)$, closed sets are described by quotient algebras, metrizability of $M$ amounts to
 separability of $C(M)$, and so on~\cite{Wegge-Olsen1994}. This is the reason why the Gelfand-Naimark theorem is considered one of the main
  sources for the development of noncommutative geometry~\cite{Connes1995,Gracia-Bond'ia2001}.
\end{remark}
Now we turn to the GNS (Gelfand-Naimark-Segal) construction. Let $\mathcal A$ be a $C^*$-algebra, and $\omega$ a state thereon.
 The GNS construction furnishes a representation of $\mathcal A$ on some ($\omega$-dependent) Hilbert space. The basic idea is to
 use multiplication in $\mathcal A$ in order to obtain a linear action of $\mathcal A$ on a vector space. So we begin by regarding
 $\mathcal A$ as a vector space $\hat{\mathcal A}$ (\emph{i.e.} we just ``forget'' multiplication). Even though  the underlying
 spaces are equal, it will be convenient to distinguish the \emph{algebra}, with elements $a\in \mathcal A$, from the underlying
  \emph{vector space}, with elements $|a\rangle \in \hat{\mathcal A}$. Recalling the Cauchy-Schwarz inequality
  (exercise \ref{ex:Cauchy-Schwarz}), we realize that it makes sense to introduce the following sesquilinear form on $\hat {\mathcal A}$:
\begin{equation}
\label{eq:w(a*b)}
\langle a | b \rangle_\omega:= \omega (a^* b).
\end{equation}
This is ``almost'' an inner product, as it may happen that $\langle a | a \rangle_\omega=0$ for some $a\neq 0$. But ignoring that
fact for a moment, we notice that the product on $\mathcal A$ can be used to make $\mathcal A$ act as a linear operator on
 $\hat {\mathcal A}$:
\begin{eqnarray*}
\mathcal A \times \hat{ \mathcal A} & \longrightarrow & \hat{ \mathcal A} \\
a, |b\rangle & \longmapsto & a\cdot | b \rangle := | a b\rangle.
\end{eqnarray*}
Now, in order to have a Hilbert space representation, we have to ``fix'' the problem with (\ref{eq:w(a*b)}). This is done as
follows. \begin{xca}
Define $\mathcal N_\omega:=\lbrace a\in \mathcal A\,|\, \omega(a^* a)=0\rbrace$, and show that it is a closed left-ideal of
$\mathcal A$.
\end{xca}
We therefore obtain a Hilbert space $\mathcal H_\omega$ given by the completion of $\hat {\mathcal A}/\mathcal
N_\omega$ with respect to the inner product $\langle [a] | [b]\rangle_\omega:= \omega (a^*b)$, where we use
$|[a]\rangle$ to denote the equivalence class of $a$ in  (the completion of) $\hat {\mathcal A}/\mathcal N_\omega$.
With this we have obtained a $*$-representation of
 $\mathcal A$ by bounded operators acting on $\mathcal H_\omega$:
\begin{eqnarray*}
\pi_\omega: \mathcal A & \longrightarrow & \mathcal B(\mathcal H_\omega)\\
a & \longmapsto & \pi_\omega (a),
\end{eqnarray*}
where $\pi_\omega (a) |[b]\rangle:= |[ab]\rangle$.
\begin{xca}
Obtain the following inequality:
\[
\sqrt{\omega(a^* a)} \leq \|\pi_\omega (a)\| \leq \|a\|.
\]
\end{xca}
Although the representation $\pi_\omega$ is a $*$-homomorphism, it is not isometric. But it can be shown that for every
$a\in \mathcal A$ there is a state $\omega$ such that $\omega(a^*a)=\|a\|^2$. Then, if we consider the direct sum
\[
\mathcal H = \bigoplus_{\omega\in \mathcal S_{\mathcal A}} \mathcal H_\omega:=\big\lbrace
 (\xi_\omega)_{\omega\in \mathcal S_{\mathcal A}}\,|\,\xi_\omega\in \mathcal H_\omega,\;\;\sum_{\omega\in \mathcal S_{\mathcal A}   }
 \|\xi_\omega\|^2_\omega<\infty\big\rbrace,
\]
where in each sequence $(\xi_\omega)_{\omega\in \mathcal S_{\mathcal A}}$ only  countable many elements are different
from zero. The representations $\pi_\omega$ then give rise to a representation $\pi$ on $\mathcal H$ for which
$\|\pi(a)\|=\|a\|$. The details of this construction can be found \emph{e.g.} in~\cite{Werner2000}. The importance of
this result is that \emph{any}  $C^*$-algebra  is isometrically isomorphic to the $C^*$-algebra of bounded operators in
some Hilbert space.
 \subsection{Composite systems, entanglement}
In this section, our aim will be to study certain quantum correlations that arise as a result of \emph{entanglement}.
This is  a vast subject, and here we will only consider elementary examples, corresponding to bipartite systems.
Although very simple from the mathematical point of view, these examples already contain the essence that will allow us
to distinguish between classical and genuinely quantum correlations. We will also discuss a recent application  of the
algebraic formalism to the study of entanglement for systems of identical
particles~\cite{Balachandran2013,Balachandran2013a}, a topic of current interest for quantum information, condensed
matter, atomic physics, and quantum optics.

We thus start by considering a composite system that is the result of coupling two subsystems 1 and 2.  We will first
consider quantum systems described directly in terms of Hilbert spaces. But, as will become clear, the algebraic
approach based only on observable algebras and states will allow us to reformulate these ideas in more generality. This
will be important for the applications to systems of identical particles. Let $\mathcal H_1$ and $\mathcal H_2$ denote
the corresponding (finite dimensional, for simplicity) Hilbert spaces, and assume that the subsystems are described in
terms of observable algebras $\mathcal A_i \subseteq \mathcal B(\mathcal H_i)$, $i=1,2$. Subsystems 1 and 2 are then
coupled to form a composite system that will be described by $\mathcal A= \mathcal A_1 \otimes\mathcal A_2$.

Let now $\omega: \mathcal A\rightarrow \mathds C$ be a state of the composite system. Consider ``partial'',  or
``local'' measurements performed on each subsystem. This leads us to consider the restriction of $\omega$ to $\mathcal
A_i$, $i=1,2$:
\[
\omega_1:=\omega\big|_{\mathcal A_1}, \;\;\;\;\omega_2:=\omega\big|_{\mathcal A_2}.
\]
\begin{definition}
In the present context of a bipartite composite system, if $\omega$ is a pure state on $\mathcal A$, we say it is
a \emph{separable} state if
\[
\omega (a\otimes b)=\omega_1(a) \omega_2(b),
\]
for all $a\in \mathcal A_1$, $b\in \mathcal A_2$. A (pure) state that is not separable in the above sense is called
an \emph{entangled} state.
\end{definition}
\begin{example}
Let $\mathcal H_1=\mathds C^n$, with basis $\lbrace |e_1\rangle,\ldots,|e_n\rangle\rbrace$ and
$\mathcal H_2=\mathds C^m$, with basis $\lbrace |u_1\rangle,\ldots,|u_m\rangle\rbrace$. Let
$\mathcal A_1=M_n(\mathds C)$ and $\mathcal A_2=M_m(\mathds C)$. The composite system is then
described by the matrix algebra $\mathcal A=\mathcal A_1\otimes \mathcal A_2$. Pick now a vector state
$|\psi\rangle \in \mathcal H_1\otimes \mathcal H_2$ and set
\begin{eqnarray*}
\omega_\psi: \mathcal A&\longrightarrow & \mathds C \\
  \alpha  & \longmapsto & \langle \psi|\alpha|\psi\rangle.
\end{eqnarray*}
Then, the following three conditions are equivalent:
\begin{itemize}
\item[($i$)] The vector state $|\psi\rangle$ is of the form $|\psi\rangle = |\varphi\rangle \otimes |\xi\rangle$.
\item[($ii$)] $S(\rho_{\psi,i})=0$, ($i=1,2$) where $S(\rho_{\psi,i})$ stands for the \emph{von Neumann entropy} of
the reduced density matrix  $\rho_{\psi,i}$ (see definitions below). \item[($iii$)] The state $\omega_\psi$ is
separable.
\end{itemize}
To see where these equivalences come from, let us expand $|\psi\rangle$ as
\[
|\psi\rangle = \sum_{i=1}^n \sum_{j=1}^m A_{ij}|e_i\rangle \otimes |u_j\rangle.
\]
Then $A\in M_{n,m}(\mathds C)$, so we may perform a singular value decomposition and write it in the form $A=U D
V^\dagger$, with $D$ a diagonal matrix in $M_{n,m}(\mathds C)$,  $U\in M_n(\mathds C)$ and $V\in M_m(\mathds C)$.  The
elements of $D$ are $D_{kl}= \delta_{kl} \sqrt{\lambda_k}$, where the $\lambda_k$, called the ``Schmidt coefficients",
are the eigenvalues of $A^\dagger A$. If we define vectors
\[
|k\rangle_1:= \sum_{i=1}^{n}U_{ik} |e_i\rangle, \;\;\;
|k\rangle_2:= \sum_{j=1}^{m}\bar V_{jk} |u_j\rangle,
\]
then we obtain
\[
|\psi\rangle = \sum_k \sqrt\lambda_k |k\rangle_1\otimes |k\rangle_2.
\]
From this expression, it is clear that ($i$) holds precisely when  there is only one non-vanishing Schmidt coefficient.
Furthermore,  from $\langle \psi |\psi\rangle=1$ we obtain $\sum_k \lambda_k =1$, with $0\leq \lambda_k \leq 1.$ We can,
therefore, regard $\lbrace \lambda_k\rbrace_k$ as a probability distribution, and compute its (Shannon) entropy,
defined as
\begin{equation}
\label{eq:H(lambda)}
H(\lbrace \lambda_k\rbrace_k):=-\sum_k \lambda_k \log \lambda_k.
\end{equation}
Notice that this function vanishes precisely when $|\psi\rangle$ is a separable state. We can relate it to the  von
Neumann entropy of the restricted states $\omega_{\psi,i}$ as follows.
\begin{xca}
\label{ex:partial-trace} Consider the restriction $\omega_{\psi,i}$ of $\omega_\psi$ to $\mathcal A_i \subseteq
\mathcal A$ ($i=1,2$). Find density matrices $\rho_{\psi,i}$ (``reduced density matrices'') acting on $\mathcal H_i$
and such that $\omega_\psi(a\otimes\mathds 1_m)= \Tr_{\mathcal H_1}(\rho_{\psi,1} a)$ and $\omega_\psi(\mathds 1_n
\otimes b)= \Tr_{\mathcal H_2}(\rho_{\psi,2} b)$.

The \emph{von Neumann entropy} of a density matrix $\rho$ is defined as
\[
S(\rho)=- \Tr (\rho \log \rho).
\]
Show that $S(\rho_{\psi,1})$ and $S(\rho_{\psi,2})$ coincide and are exactly equal to $H(\lbrace \lambda_k\rbrace_k)$
from (\ref{eq:H(lambda)}).
\end{xca}
\end{example}
We therefore consider the quantity  $S(\rho_{\psi,i})$ ($i=1,2$) as a \emph{measure} of the ``amount of entanglement''
of the state $\omega_\psi$. In order to really appreciate the meaning of this assertion, we have to understand the
(very surprising, and interesting)  features of quantum correlations. This can be considered as one of the starting
points of quantum information theory. In quantum information theory, there are orderly ways to pose and study these type
of problems where, for instance, in the (pure) bipartite case the von Neumann entropy of the reduced density matrix can
be obtained as (basically) the unique entanglement measure of the state $|\psi\rangle$, when entanglement is regarded
as a physical resource and defined in operational terms~\cite{Nielsen2010,Keyl2002}. The general problem, that of
multipartite entanglement, is much more complicated and involves many mathematically as well as physically interesting
problems. However, our interest here will be restricted to the (pure) bipartite case. This will be enough to illustrate
one of our points, which is that the algebraic framework leads to a unified description of physical systems, no matter
whether they are classical or quantum. Also, the ``detachment'' from an \emph{a priori} given Hilbert space in the
quantum case allows for a more clear understanding of phenomena. As it turns out, the noncommutativity of the
observable algebra can be seen as the main source of (truly quantum) correlations. This fits nicely with the relation
between commutative and noncommutative spaces, as hinted above.

The next example illustrates the previous remarks.

\begin{example}
Let us consider the composition of two classical systems. By this we mean that the ``coupling'' of two systems
described by (commutative) observable algebras of the  form $\mathcal A_1=C(X)$ and $\mathcal A_2= C(Y)$ is effected by
the tensor product $\mathcal A_1\otimes \mathcal A_2$. Tensor products of $C^*$-algebras have to be treated carefully,
but  assuming we have defined the appropriate tensor product, let us assume that $C(X)\otimes C(Y)\simeq C(X\times Y)$
(a detailed discussion about tensor products can be found in~\cite{Wegge-Olsen1994}). Hence, the composite system will
be assumed to be described by the observable algebra $\mathcal A= C(X\times Y)$. Consider now a state $\omega$ on
$\mathcal A$. It is not difficult to see that it must be given by a probability distribution $p(x,y)$  in such a way
that for $h\in\mathcal A$,
\[
\omega(h)= \int_{X\times Y}h(x,y) p(x,y) dx dy.
\]
Restriction to $\mathcal A_1$ or $\mathcal A_2$ in this case leads to marginal  distributions $p_1(x)=\int_Y p(x,y) dy$
and $p_2(y)=\int_X p(x,y) dx$, which in turn define states $\omega_i$ on $\mathcal A_i$. From the  above definition we
have that $\omega$ is separable if, and only if,
\[
\omega(f\otimes g)= \omega_1(f)\omega_2(g),\;\;\mbox{for all}\;\;f\in \mathcal A_1, g\in \mathcal A_2.
\]
But from the properties of the Gelfand transform it follows that any pure state on $C(X\times Y)$ is of the form
$\omega(h)=h(x_0,y_0)$, \emph{i.e.}, they are given by point measures: $p(x,y)=\delta_{(x_0,y_0)}(x,y)$. It follows
immediately that $\omega(f\otimes g)= \omega_1(f)\omega_2(g)$. This illustrates the fact that entanglement entails
purely quantum correlations,  as all pure states of a classical composite (bipartite) system are separable.
\end{example}
 \begin{xca}
 The previous example was formulated in rather loosely terms. Provide the necessary details to turn it into a rigorous proof
 that works for locally compact spaces.
 \end{xca}
\begin{example}\label{example:psi_lambda}
On $\mathcal H=\mathbb C^2\otimes \mathbb C^2$ consider the following family of states:
\[
|\psi_\lambda \rangle:= \sqrt{\lambda}\, |+-\rangle -\sqrt{1-\lambda}\,|-+\rangle,
\]
and the following types of observables, acting on $\mathbb C^2$:
\[
P(a):=\frac{1}{2}\left(\mathds 1_2 +\vec a\cdot \vec \sigma\right),\;\;\;
E(a):=\vec a\cdot \vec \sigma,
\]
with $\vec a$ a \emph{unit} vector in $\mathbb R^3$ and $\vec \sigma=(\sigma_1,\sigma_2,\sigma_3)$ (Pauli matrices).
Straightforward computations give
\[
\langle\psi_\lambda|E(a)\otimes E(b)|\psi_\lambda\rangle= -a_3 b_3-2\sqrt{\lambda(1-\lambda)}(a_1b_1+a_2b_2)
\]
and
\begin{eqnarray*}
\lefteqn{\langle\psi_\lambda|P(a)\otimes P(b)|\psi_\lambda\rangle={ }}\\
&& { } =\frac{1}{4}\left( 1+(2\lambda-1)(a_3-b_3)-a_3b_3-2\sqrt{\lambda(1-\lambda)}(a_1b_1+a_2b_2)) \right).{
}
\end{eqnarray*}
Whereas for $\lambda=0,1$ we obtain separable states, for $0<\lambda <1$ we have entangled states.  In particular, for
$\lambda =1/2$ we obtain a maximally entangled state (Bell state), for which
\[
\langle\psi_{1/2}|P(a)\otimes P(b)|\psi_{1/2}\rangle=\frac{1}{4}\left(1- \vec a\cdot \vec b \right).
\]
In this case, the marginal distributions give
\begin{eqnarray*}
P_A(a,+)\;=&\frac{1}{2}&=\;P_A(a,-),\\
P_B(b,+)\;=&\frac{1}{2}&=\;P_B(b,-),
\end{eqnarray*}
whereas for the joint probabilities we obtain ($r,r'= \pm 1$):
\[
P_{r r'}=\frac{1}{4}\left(1- r r' \vec a\cdot \vec b\right).
\]
In particular, for the choice $\vec a \cdot \vec b=1$, we obtain \emph{total anticorrelation}.
\end{example}
What is so special about the previous example? Consider the Clauser-Horne-Shimony-Holt (CHSH)
inequality~\cite{Nielsen2010,Keyl2002}:
\[
\langle   A (B+B')+ A' (B-B' )\rangle \leq 2,
\]
where $\langle X \rangle$ denotes the expectation value of an observable $X$, and where for the involved  observables
it is assumed that $-1\leq A,A',B,B'\leq 1$. The inequality can be obtained if we assume some underlying (classical)
probability space $(\Omega, d\mu)$, for which the correlations above take the form
\[
\langle  A B\rangle_\mu = \int_\Omega A(x)B(x)d\mu (x),
\]
or, stated in terms of algebras of observables:
\begin{xca}
Let $\mathcal A$ be a \emph{commutative} $C^*$-algebra, $\omega$ a pure state thereon, and
$a,a',b, b'$ elements of $\mathcal A$, all of them with norm less or equal to one. Show that
\[
\omega (a(b+b') + a' (b-b'))\leq 2.
\]
\end{xca}
But, in contrast to this,  in  example \ref{example:psi_lambda}  (with $\lambda=1/2$) we obtain
\begin{eqnarray*}
\lefteqn{ \langle \psi_{1/2}|E(a)\otimes (E(b)+E(b')) + E(a')(E(b)-E(b')) |\psi_{1/2}\rangle= { }}\\
&&  { } \hspace{6cm} =  \vec a \cdot \vec b +\vec a \cdot \vec b'+\vec a'\cdot\vec b-\vec a' \cdot \vec b'. { }
\end{eqnarray*}
Now, it is easy to find an arrangement for the vectors $a,a',b,b'$ such that the CHSH inequality is violated. The
importance of these type of inequalities, of which Bell's inequality was the first one, cannot be overemphasized, as
they have allowed for definitive  experimental tests of quantum mechanics.
 The CHSH inequality is also relevant in the context of hidden-variable models, and related no-go theorems~\cite{Holevo2011}.

In the next section we will consider  the problem of entanglement for systems of identical particles, for which the
tools developed in the last sections will prove very useful.

\subsection{Identical particles and entanglement}
When combined, two of the most intriguing features of quantum theory -the intrinsic indistinguishability of  identical
particles and quantum entanglement- lead to formidable conceptual issues that have been addressed for years, and for
which no generally  accepted framework exists. The main source of problems comes from the fact that  entanglement and
related concepts have been studied mainly in cases where a subsystem decomposition coincides with a tensor product
decomposition of the Hilbert space. But precisely in the case of identical particles, the Hilbert space is the
antisymetric/symmetric subspace obtained from the action of the permutation group on a tensor product space. So
precisely in this case decomposition into subsystems does not correspond to tensor product decomposition. This makes
the use of partial trace a doubtful operation, leading to a clash with standard notions of entanglement.  In~\cite{Balachandran2013,Balachandran2013a,Balachandran2013b}, we have shown how, using the representation theory of
operator algebras, it is possible to obtain a generalized notion of entanglement. The universality of the approach,
which can be applied to particles obeying any kind of statistics including bosons, fermions, parafermions and particles
obeying braid statistics, provides a unifying view of entanglement, with many potential applications.

The key idea is that, since  subsystems can always be described in terms of subalgebras,  the use  of partial trace
will be superseded by the more general notion of restriction of a quantum state $\omega$ to a subalgebra. Then one can
make use of the GNS construction in order to find a Hilbert space representation $\pi_\omega$, as above. 
\begin{xca}
\label{ex:entropy-irreducibility} Prove that the GNS representation $\pi_\omega$ is irreducible precisely when the
state $\omega$ is pure.
\end{xca}
 From this exercise, it follows that  the condition for irreducibility of the representation $\pi_\omega$ is
 precisely that the von Neumman entropy of the state $\omega$ vanishes and so, given a subsystem described in
 terms of a subalgebra $\mathcal A_0\subset \mathcal A$, one may consider the restriction of the state $\omega$
 to $\mathcal A_0$. This is a generalization of the notion of partial trace. Hence, a generalized notion of
  entanglement emerges,  based on the von Neumann entropy of the restricted state. The properties of this
  entanglement measure are very closely tied to the GNS-representation of the subalgebra and the restricted state.
   In other words, one finds that entanglement depends \emph{both} on the state and the subsystem  of the full system.
This formalism, then, lends itself to the study of  problems where partial trace loses its meaning,
as is the case with systems of identical particles. Let us consider a few illustrative examples.
\begin{example}
\label{example:2x2matrix}
Consider the algebra $\mathcal A=M_2(\mathds C)$ of $2\times 2$ matrices with complex entries, and let
\[
\omega_\lambda (a)= \lambda a_{11} +(1-\lambda ) a_{22}
\]
be a state on $\mathcal A$. Of course, this can
only happen if $0\leq \lambda \leq1$. Let us then see how the GNS representation is constructed, and how we can
associate an entropy to the  state $\omega_\lambda$. For this purpose, let us consider the matrix units
$e_{ij}=|i\rangle\langle j|$ $(i,j=1,2)$. They generate the algebra, and provide a basis for the underlying vector
space $\hat {\mathcal A}$. They also fulfill the relations
\begin{equation}
\label{eq:e_ij}
e_{ij}e_{kl}=\delta_{jk} e_{il}.
\end{equation}
Let us now check whether there will be null vectors. Writing a general element $a\in \mathcal A$ as
\[
a= \sum_{ij} a_{ij} e_{ij},
\]
we compute
\begin{eqnarray}
\label{eq:null-states-2x2}
\omega_\lambda (a^* a) & = & \sum_{i,j,k}\omega_\lambda(\bar a_{ki} a_{kj}e_{ij})\\
 & = &  \lambda (|a_{11}|^2+ |a_{21}|^2) +  (1-\lambda) (|a_{12}|^2+ |a_{22}|^2).\nonumber
\end{eqnarray}
From this expression we recognize that, if $\lambda\in (0,1)$, then $\mbox{dim}(\mathcal H_{\omega_\lambda})= 4$,
whereas for $\lambda=0$ or $\lambda=1$ we obtain  $\mbox{dim}(\mathcal H_{\omega_\lambda})= 2$. In fact, if $\lambda$
is different from zero or one, we must have both $(|a_{11}|^2+ |a_{21}|^2) =0$ and $(|a_{12}|^2+ |a_{22}|^2)=0$. It
follows that $\mathcal N_{\omega_\lambda}=\lbrace 0 \rbrace$. From
\[
\pi_{\omega_\lambda} (a)  |[b]\rangle = |[ab]\rangle
\]
 and (\ref{eq:e_ij}) we see that the representation is reducible, with 2 invariant subspaces generated,
 respectively, by $\lbrace |[e_{11}]\rangle,|[e_{21}]\rangle\rbrace$ and $\lbrace |[e_{12}]\rangle,|[e_{22}]\rangle\rbrace$.
 From exercise \ref{ex:entropy-irreducibility}  we conclude that, in this case, the state has to be a mixed state.
  The cases $\lambda=0$ and $\lambda =1$ lead to a non-trivial null space, and from
(\ref{eq:null-states-2x2}) we conclude that the dimension of these irreducibles is, in both cases, equal to 2.
 It is left to the reader to obtain a pure state on the algebra $M_2(\mathds C)\otimes M_2(\mathds C)$ such that
$\omega_\lambda$ is the restriction to $M_2(\mathds C)$.
\end{example}
The next question we need to address is that of how to compute the entropy of the reduced state.  Since our algebras
will in general be unital, there is a simple way to acomplish this, namely, given a (unital) $C^*$-algebra $\mathcal A$
and a state $\omega$ on it, it is always possible to find a \emph{density matrix} $\rho_\omega$ acting on the GNS space
$\mathcal H_\omega$, and such that
\[
\omega(a)= \Tr_{\mathcal H_\omega}( \rho_\omega a)
\]
for all $a$ in $\mathcal A$. From the decomposition of the GNS space into irreducibles,
\[
\mathcal H_\omega = \bigoplus_j \mathcal H^{(j)}_\omega,
\]
we obtain \emph{projectors} $P^{(j)}$, with  $\sum_j P^{(j)} = \mathds 1_{\mathcal H_\omega}$.  From the definition of
the inner product in $\mathcal H_\omega$ we have $\omega (a) = \langle [\mathds 1_{\mathcal A}] | \pi_\omega (a) |
[\mathds 1_{\mathcal A}]\rangle$. But then, using an orthonormal basis $\lbrace |n\rangle \rbrace_n$ on $\mathcal
H_\omega$, we can write:
\begin{eqnarray}
\omega (a) & =& \langle [\mathds 1_{\mathcal A}] | \pi_\omega (a) |  [\mathds 1_{\mathcal A}]\rangle\nonumber\\
&=&
\langle [\mathds 1_{\mathcal A}] | \sum_k P^{(k)}\pi_\omega (a) |  [\mathds 1_{\mathcal A}]\rangle\nonumber \\
&=&
\langle [\mathds 1_{\mathcal A}] | \sum_k P^{(k)}\pi_\omega (a)P^{(k)} |  [\mathds 1_{\mathcal A}]\rangle\nonumber\\
&=&
\langle [\mathds 1_{\mathcal A}] | \sum_k P^{(k)}\pi_\omega (a)  \sum_n |n\rangle \langle n| P^{(k)} |  [\mathds 1_{\mathcal A}]\rangle\nonumber\\
&=& \Tr_{\mathcal H_\omega} (\rho_\omega \pi_\omega(a)),\nonumber
\end{eqnarray}
with $\rho_\omega$ given by
\[
\rho_\omega= \sum_k P^{(k)} | [\mathds 1_{\mathcal A}] \rangle \langle [\mathds 1_{\mathcal A}]| P^{(k)}.
\]
 Consider now a composite, bipartite system for which the Hilbert space is of the form $\mathcal H=\mathcal H_A\otimes \mathcal H_B$. If $\omega_\psi$ is a state on the algebra of bounded  operators on  $\mathcal H_A\otimes \mathcal H_B$ that is obtained from a vector state $|\psi\rangle \in \mathcal H$, then its restriction to the subalgebra $\mathcal A_1$ generated by elements of the form
$K\otimes \mathds 1_B$ gives a state
\[
\omega_1(K)\equiv \omega\mid_{\mathcal A_1}(K\otimes \mathds 1_B)= \Tr_{\mathcal H_A}(\rho_A K),
\]
where $\rho_A$ is the \emph{reduced} density matrix, defined through partial trace (\emph{cf.} exercise~\ref{ex:partial-trace}),
\[
\rho_A =\Tr_{\mathcal H_B} |\psi\rangle\langle \psi|.
\]
We therefore see that, for these kind of systems,  partial trace equals restriction (to some subalgebra).
\begin{xca}
Compute de von Neumann entropy  of the density matrix $\rho_{\omega_\lambda}$ corresponding to the GNS  space of
example \ref{example:2x2matrix} and provide a physical interpretation.
\end{xca}
 The real usefulness  of the algebraic approach becomes apparent only when we consider situations where the Hilbert space does
  not have a simple tensor product structure. This is what happens, \emph{e.g.}, with systems of identical particles. Because of the
  symmetrization postulate, the Hilbert space contains only the symmetric (for bosons) or the antisymmetric (for fermions)
  subspaces of the many-particle Hilbert space. As a consequence, one finds that states that from a physical point of view
  should not \emph{a priori} be considered to be entangled, will have reduced density matrices with non-vanishing entanglement
  entropy~\cite{Tichy2011}. In approaches like those of~\cite{Benatti2012} or~\cite{Balachandran2013}, the description of
  subsystems is given by specifying suitable subalgebras. Then, the restriction of a given state to the subalgebra provides
  a physically sensible generalization of the notion of partial trace.
 Applying the GNS construction to the restricted state, it is possible
 to  study the entropy emerging from  restriction and use it as a generalized measure of entanglement.

Let us briefly discuss how this entanglement measure can be computed in concrete cases. For this purpose, consider a
Hilbert space $\mathcal H^{(1)}=\mathds C^d$, assumed to correspond to the space of 1-particle states of a fermionic
system. The full Hilbert space is then the antisymmetric Fock space $\mathcal F$ obtained from $\mathcal H^{(1)}$. It
decomposes as a direct sum of spaces of fixed number of particles. The $k$-particle Hilbert space $\mathcal{H}^{(k)}$
is just the antisymmetrized $k$-fold tensor product of $\mathcal{H}^{(1)}$. This is a kind of ``toy model'' for a
fermionic quantum field theory, but many important features of a quantum field theory can be seen to appear already at
this level. One of these features is the connection to the representation theory of Clifford algebras, as explained in
full detail in~\cite{Gracia-Bond'ia2001}. As is well-known, the Clifford algebra of $\mathcal H^{(1)}$ acts naturally
on the exterior algebra $\bigwedge^\bullet(\mathcal H^{(1)})$ which, in turn, is related to the Fock space construction
in the following way.

Let $\lbrace e_n \rbrace_n$ denote an orthonormal basis for $\mathcal H^{(1)}$, and denote with $a_n^{(\dagger)}$ the
corresponding annihilation (creation) operators (\emph{cf.}~exercise \ref{ex:CAR-algebra-finite-dim}). Then, there is a vector
space isomorphism between $\mathcal F$ and $\Lambda^\bullet(\mathcal H^{(1)})$, furnished by the correspondence
\begin{equation}
\label{eq:Nelson-correspondence}
e_{i_1}\wedge e_{i_2}\wedge \cdots \wedge e_{i_k} \; \longleftrightarrow \; a^\dagger_{i_1} a^\dagger_{i_2} \cdots
a^\dagger_{i_k} |0\rangle.
\end{equation}
This correspondence is behind the famous quote by E. Nelson: \emph{``first quantization is a mystery, but second quantization
is a functor''}. In fact, given a self-adjoint operator $A$ on $\mathcal{H}^{(1)}$ (that is, a 1-particle observable),
we obtain (by functoriality) an operator $d\Gamma(A)$  acting on Fock space, whose restriction to the $k$-particle
sector $\mathcal{H}^{(k)}$ is given by
\begin{eqnarray}
\label{eq:dGamma^k}
\lefteqn{
d\Gamma^{(k)}(A):= (A\otimes\mathds{1}_d \otimes\cdots\otimes\mathds{1}_d)+{ }}\\
&& { }
 + (\mathds{1}_d\otimes A\otimes\cdots\otimes\mathds{1}_d) +\cdots+ (\mathds{1}_d\otimes\cdots\otimes\mathds{1}_d\otimes
 A).
  { }\nonumber
 \end{eqnarray}
Taking the correspondence (\ref{eq:Nelson-correspondence}) into account, we obtain the following expression for
$d\Gamma(A)$ in terms of creation/annihilation operators,
\begin{equation}
\label{eq:dGamma}
d  \Gamma(A) =\sum_{i,j} A_{ij} a_i^\dagger a_j,
\end{equation}
where $A_{ij}= \langle e_i|A|e_j\rangle$. In the physics literature, the operator $d\Gamma (A)$ is referred to as the
\emph{second quantization} of $A$.

In the present quantum-mechanical context, where the number of particles is kept fixed, we want to focus our attention
on the operator $d\Gamma^{(k)}(A)$. One of the properties of this operator is that it preserves the symmetries of
$\mathcal{H}^{(k)}$. Furthermore, the map $A \longrightarrow d\Gamma^{(k)}(A)$ allows us to study subalgebras of
1-particle observables.

The simplest example we can consider in order to illustrate entanglement issues for systems of identical particles is
that of just two fermions. Let us then consider, as done in~\cite{Eckert2002} and in~\cite{Balachandran2013}, a
2-fermion system, where each fermion can be in  a linear superposition  of 4 basic states which, for the sake of physical
interpretation, will be divided into internal and external degrees of freedom. So we describe 1-particle states in
terms of a set of (fermionic) creation/annihilation operators $a^{(\dagger)}_\lambda, b^{(\dagger)}_\lambda$, where $a$
stands for \emph{``left''}, $b$ for \emph{``right''} (the external degrees of freedom) and $\lambda=1,2$ for spin up and
down (the internal degrees of freedom). Hence, in this case we have $\mathcal H^{(1)}=\mathds C^4$ and,   for the
2-fermion space, $\mathcal H^{(2)}=\bigwedge^2 \mathds C^4$. An orthonormal basis for $\mathcal H^{(2)}$ is given
 by the  vectors
\[
 a^\dagger_1 a^\dagger_2|0\rangle,\;\; b^\dagger_1
b^\dagger_2|0\rangle,\;\; a^\dagger_1 b^\dagger_2|0\rangle\;\;\mbox{and}\;\; a^\dagger_2 b^\dagger_1|0\rangle.
\]
 The two-particle algebra $\mathcal A$ of observables is thus isomorphic to the  matrix algebra $M_{6}(\mathds
C)$.

For $|\psi_\theta\rangle=(\cos\theta a_1^\dagger b_2^\dagger + \sin\theta a_2^\dagger b_1^\dagger) |\Omega\rangle$, the
corresponding state $\omega_\theta$ is given by $\omega_{\theta} (\alpha)= \langle \psi_{\theta}
|\alpha|\psi_{\theta}\rangle$ for $\alpha\in \mathcal A$. We choose the subalgebra $\mathcal A_0$ to be the one generated  by
  $\mathds 1_{\mathcal A}$, $n_{12}= a_1^\dagger a_1 a_2^\dagger a_2$, $N_a=a_1^\dagger a_1 + a_2^\dagger a_2$ and
$T_{i=1,2,3}=(1/2)~a_\lambda^\dagger (\sigma_i)^{\lambda\lambda'} a_{\lambda'}$.
Physically, it corresponds to the subalgebra of  one-particle observables corresponding to measurements at the left location.
\begin{xca}
Consider the restriction of
$\omega_\theta$ to $\mathcal A_0$ and study the GNS representation corresponding to this choice.
For $\theta = 0,\pi/2$ you should obtain vanishing entropy, in contrast to the result $S=\log_2 2$ obtained via partial trace
  for states with Slater rank one~\cite{Balachandran2013,Balachandran2013a,Ghirardi2004,Tichy2011}.
\end{xca}

\section{Spin Chains}

\subsection{The transverse Ising chain}
The Hamiltonian for the (quantum) Ising chain in a transverse field is given by
\begin{equation}
\label{eq:H-Ising-transv}
H = -  \sum_{i=1}^{N-1} \sigma_i^x \sigma_{i+1}^x  - \lambda  \sum_{i=1}^N \sigma_i^z.
\end{equation}
An important aspect when solving the model and studying its solution is the type of boundary conditions considered. We will be interested in both open as well as periodic boundary conditions. The first step in the solution of this model is the so-called Wigner-Jordan transformation, that allows us to express all spin operators in terms of \emph{fermionic} creation/annhilation operators:

As can be easily checked, the operators defined by
\begin{eqnarray}
\label{eq:Wigner-Jordan-def}
 a_i &  = & \sigma^z \otimes \cdots\otimes \sigma^z\otimes \sigma^+, \nonumber\\
 a_i^\dagger & = & \sigma^z \otimes \cdots\otimes \sigma^z\otimes \sigma^-,
\end{eqnarray}
where $\sigma^\pm=\frac{1}{2}(\sigma^x\pm i\sigma^y)$,
are fermionic operators.
\begin{xca}\label{ex:1.1}
Check that the operators defined above indeed satisfy the  CAR algebra: $\lbrace a_i, a_j^\dagger\rbrace =\delta_{ij}$,
 $\lbrace a_i, a_j\rbrace = 0 = \lbrace a_i^\dagger, a_j^\dagger\rbrace$.
\end{xca}
 The inverse transformation is given by
\begin{eqnarray}
\label{eq:Wigner-Jordan-inverse}
\sigma_i^z  & = & 1 -2 a_i^\dagger a_i,\nonumber\\
\sigma_i^x  & = & \left(\prod_{m<i} (1- 2 a_m^\dagger a_m)\right) (a_i^\dagger + a_i),\\
 \sigma_i^y  & = & i\left(\prod_{m<i} (1- 2 a_m^\dagger a_m)\right) (a_i^\dagger - a_i).\nonumber
\end{eqnarray}
\begin{xca}\label{ex:1.2} Check that (\ref{eq:Wigner-Jordan-inverse}) is the inverse transformation to (\ref{eq:Wigner-Jordan-def}).
\end{xca}
With this we obtain, for  the interaction terms of the Hamiltonian,
\[
\sigma_i^x \sigma_{i+1}^x = (a_i^\dagger -a_i)(a_{i+1}^\dagger + a_{i+1}),
\]
and so $H$ the takes the form
\begin{equation}
\label{eq:H-Ising2}
H= -  \sum_{i=1}^{N-1}(a_i^\dagger -a_i)(a_{i+1}^\dagger + a_{i+1}) - \lambda \sum_{i=1}^N (1-2 a_i^\dagger a_i).
\end{equation}
The term containing the external field $\lambda$ is diagonal in this basis. The  constant term  $-\lambda N$ coming from the last sum  is usually disregarded, because its only effect is to shift the energy spectrum. We will nevertheless keep all terms, in order to be able to compare with numerical solutions in the spin basis, for small values of $N$.

Expanding all terms in (\ref{eq:H-Ising2}) and moving all creation operators to the left, we obtain (\emph{cf.}~Eq. (10.14) in~\cite{Sachdev2011}):
\[
H = - \sum_{i=1}^{N-1} (a_i^\dagger a_{i+1}  + a_{i+1}^\dagger a_i +a_i^\dagger a_{i+1}^\dagger - a_i a_{i+1} )
+ 2 \lambda \left( \sum_{i=1}^N a_i^\dagger a_i\right) - \lambda N.
\]
The point of using the Wigner-Jordan transformation for this model is that the Hamiltonian becomes ``almost'' diagonal. By this we mean that it is a sum of local, \emph{quadratic} expressions in the creation and annihilation operators. As we will see, such models can be exactly solved. We now write $H$ in a suggestive matrix notation. Arranging all creation and annihilation operators in rows and columns, we can write $H$ as a kind of quadratic form. For example, we have, for $N=2$,
\[
(a_1^\dagger, a_2^\dagger) \left(
\begin{array}{cc}
 2 \lambda & -1\\
 -1    &  2\lambda
\end{array}
\right) \left(
\begin{array}{c}
a_1\\
a_2
\end{array}
\right) = -a_1^\dagger a_2 -a_2^\dagger a_1 + 2 \lambda a_1^\dagger a_1 + 2 \lambda a_2^\dagger a_2.
\]
The expression corresponding to the same term for $N=4$ is then
\begin{eqnarray*}
\lefteqn{
(a_1^\dagger, a_2^\dagger,a_3^\dagger, a_4^\dagger) \left(
\begin{array}{cccc}
 2 \lambda &    -1   &   0    &    0 \\
 -1    &   2\lambda  &   -1    &    0 \\
 0    &    -1   &  2\lambda  &    -1    \\
 0    &    0   &   -1    &  2\lambda \\
\end{array}
\right) \left(
\begin{array}{c}
a_1\\
a_2\\
a_3\\
a_4
\end{array}
\right) =  { } }    \\
 & & { } =-a_1^\dagger a_2 +a_2^\dagger a_1 -
 a_2^\dagger a_3 -a_3^\dagger a_2 -
 a_3^\dagger a_4 -a_4^\dagger a_3 +
  2 \lambda \sum_{i=1}^4 a_i^\dagger a_i . { }
\end{eqnarray*}
For arbitrary $N$, we may define the following $N\times N$ matrices:
\begin{equation}
\label{eq:AB-open} \hspace{-0.3cm}A=
 \begin{pmatrix}
  2\lambda    &  -1  &  0  &  0 &\cdots \\
  -1    &   2\lambda   &  -1  &  0 &  \\
  0    &  -1  &   2\lambda  &  -1 &  \\
  0  &   0   &  -1  &   2\lambda  & \cdots  \\
    \vdots            &         &      &     \vdots   & \ddots
 \end{pmatrix},\;
 B=
 \begin{pmatrix}
  0    &  -1  &  0  &  0 &\cdots \\
  1    &  0  &  -1  &  0 &  \\
  0    &  1  &  0  &  -1 &  \\
  0  &   0   &  1  &  0 & \cdots  \\
    \vdots            &         &      &     \vdots   & \ddots
 \end{pmatrix}.
 \end{equation}
 The matrix elements of these two matrices can be written as follows:
\begin{equation}
A_{ij}   =   2 \lambda \delta_{i,j} - (\delta_{i+1,j}  + \delta_{i,j+1}),\;\;\;
B_{ij}   =  - (\delta_{i+1,j} - \delta_{i,j+1}). \nonumber
\end{equation}
Hence, the Hamiltonian takes the following  form:
\begin{equation}
\label{eq:H-Ising3}
 H = \sum_{i,j=1}^N \left[ a_i^\dagger A_{ij}a_j +\frac{1}{2}\left( a_i^\dagger B_{ij}a_j^\dagger -
a_i B_{ij} a_j  \right)\right] -\lambda N.
\end{equation}

\subsection{Open boundary conditions}\label{sec:OpenBoundaryConditions}
For several reasons, including the study of edge states, it is instructive to explore the explicit solution of this model for \emph{open} boundary conditions.
 We start with the Ising Hamiltonian written in the form (\ref{eq:H-Ising3}). Any model that can be written as a quadratic form
  can be expressed in this way, the ``only'' difference being the explicit form of the matrices $A$ and $B$ (notice that,
   in order for $H$ to be  Hermitian, $A$ has to be symmetric and $B$ antisymmetric). In principle, thus, the method presented below
   (following the  work of Lieb, Schultz and Mattis~\cite{Lieb1961})
    can be applied to any such model.

For periodic boundary conditions, it is usually more convenient to take into account translation invariance and hence
to introduce Fourier transformed operators. But for open boundary conditions, translation invariance is ``broken'' and
then it is a good idea to start right away with a Bogoliubov transformation, as explained below.

Recall that the operators $a_i,\, a_j^\dagger$  defined in (\ref{eq:Wigner-Jordan-def}) obey fermionic canonical anti-commutation (CAR) relations (\emph{cf.}~exercise~\ref{ex:1.1} ):
\[
\lbrace a_i, a_j^\dagger\rbrace = \delta_{ij},\;\;\;\lbrace a_i, a_j\rbrace=0=\lbrace a_i^\dagger, a_j^\dagger\rbrace.
\]
In general, a Bogoliubov transformation  is a mapping induced by a  change of  basis on the one-particle Hilbert space
(unitary transformation), its main effect being to provide  a new set of creation/annihilation operators for which the
Hamiltonian (\ref{eq:H-Ising3}) becomes diagonal. Consider, then, a new set of operators given by
\begin{equation}
\label{eq:g-h-defined}
c_k   =   \sum_{i=1}^N\left(g_{ki} a_i + h_{ki} a_i^\dagger\right),\;\;\;\;\;c_k^\dagger   =
\sum_{i=1}^N\left(\bar{g}_{ki} a_i^\dagger + \bar{h}_{ki} a_i\right),
\end{equation}
where $g$ and $h$ are $N\times N$ matrices to be chosen so that
\begin{itemize}
\item[(i)] The new operators satisfy the same CAR algebra:
\begin{equation}
\label{eq:Bogol-CAR}\lbrace c_k, c_l^\dagger\rbrace = \delta_{kl},\;\;\lbrace c_k, c_l\rbrace=0=\lbrace c_k^\dagger, c_l^\dagger\rbrace.
\end{equation}
\item[(ii)] The Hamiltonian becomes \emph{diagonal} in the new basis:
\begin{equation}
\label{eq:H-Ising-diagonal}H=\sum_k \Lambda_k c_k^\dagger c_k + \mu\end{equation} (with $\mu$ some constant).
\end{itemize}
\begin{xca}\label{ex:1.3} Show that the requirement (\ref{eq:Bogol-CAR}) leads to the following conditions:
\begin{eqnarray}
\label{eq:gh}
g g^\dagger + h h^\dagger &=& \mathds 1_N,\nonumber\\
g h^t + h g^t & =& 0.
\end{eqnarray}
\end{xca}
\begin{xca}\label{ex:1.4}
Compute the trace of $H$ in two different ways, in order to show that the constant term in (\ref{eq:H-Ising-diagonal}) is
 given by
 \[
 \mu  =\frac{1}{2}\left(  \Tr A-\sum_k \Lambda_k  \right) -\lambda N.
 \]
\end{xca}
The second condition above, Eq. (\ref{eq:H-Ising-diagonal}),  will lead to an eigenvalue problem for $g$ and $h$, the solution of which amounts -in principle- to the solution of the full problem.

Now we compute the commutator $[c_k, H]$ in two different ways, once using  (\ref{eq:H-Ising3}) and once using
(\ref{eq:H-Ising-diagonal}). This leads to the following set of
equations:
\begin{equation}
\label{eq:gA-hB}
g_{ki}\, \Lambda_k   =  \sum_{j=1}^N \left(g_{kj} A_{ji} -h_{kj} B_{ji}\right),\,\;\;
h_{ki}\, \Lambda_k   =  \sum_{j=1}^N \left(g_{kj} B_{ji} -h_{kj} A_{ji}\right).
\end{equation}
In order to solve this eigenvalue problem, it proves  convenient to introduce new matrices $\Phi$ and $\Psi$, as follows:
\begin{equation}
\Phi  :=  g + h,\;\;\;\;\;
\Psi  :=  g - h. \nonumber
\end{equation}
If we now define for each $k$ a vector $|\Phi_k\rangle$, the $i^{th}$ component of which is given by $\Phi_{ki}$, and
similarly for $\Psi$, we find that (\ref{eq:gA-hB}) can be written as follows:
\begin{equation}
(A-B) |\Psi_k\rangle  =  \Lambda_k |\Phi_k\rangle,\;\;\;\;
(A+B) |\Phi_k\rangle  =  \Lambda_k |\Psi_k\rangle,  \nonumber
\end{equation}
or, equivalently, as
\begin{eqnarray}
\label{eq:(A-B)(A+B)}
(A-B) (A+B)|\Phi_k\rangle & = & \Lambda_k^2 \; |\Phi_k\rangle\nonumber\\
(A+B) (A-B) |\Psi_k\rangle & = & \Lambda_k^2 \; |\Psi_k\rangle.
\end{eqnarray}
With this we have reduced our problem from the diagonalization of a $2^N\times 2^N$ matrix to that of diagonalizing two
$N\times N$ ones. As mentioned before, this would be a very easy task if we would have chosen periodic boundary
conditions. The reason being  that for periodic boundary conditions the matrices $(A\pm B)(A\mp B)$ are
\emph{Toeplitz}. But for open boundary conditions the matrices $A$ and $B$ are given by (\ref{eq:AB-open}), so that
\begin{equation}
\label{eq:(A-B)(A+B)matrix}
\frac{1}{4}(A-B)(A+B) = \begin{pmatrix}
  \lambda^2    &  - \lambda   &           &           &         &     \\
 -\lambda   & 1+\lambda^2 &  \hspace*{-0.5cm}-\lambda     &           &         &     \\
         &  -\lambda   &  1+\lambda^2  & \ddots    &         &     \\
         &         &  \ddots   &   \ddots  &         &   \\
         &         &           &           &  -\lambda   &  \\
         &         &           &    -\lambda   & \;\;1+ \lambda^2  & \\
 \end{pmatrix}
\end{equation}
and
\begin{equation}
\label{eq:(A+B)(A-B)matrix}
\frac{1}{4}(A+B)(A-B) = \begin{pmatrix}
 1+ \lambda^2    &  -\lambda   &           &           &         &     \\
 -\lambda   &         1+ \lambda^2      &     &           &         &     \\
         &         &  \ddots   &   \ddots  &         &   \\
                        &         &    &    &         &   \\
         &         &     \ddots        &      1+ \lambda^2      &  -\lambda   &  \\

         &         &           &  \hspace*{-0.5cm}  -\lambda   & \lambda^2  & \\
 \end{pmatrix}.
\end{equation}
In order to solve (\ref{eq:(A-B)(A+B)}), we propose the following ansatz:
\begin{equation}
\label{eq:Psi_kl}
\Psi_{kl} = \alpha e^{i k l} + \beta e^{-i k l},
\end{equation}
with $\alpha$  and $\beta$ constants to be determined. The eigenvalue equation will then give 3 independent equations.
The first one is obtained by equating the $j^{th}$ component of $(A+B)(A-B)|\Psi_k\rangle$ with $\Lambda_k^2\Psi_{kj}$,
for $j=2,\ldots, N-1$. All these choices of $j$ yield the same equation (by enforcing the vanishing of the coefficients
of $\alpha$ and $\beta$), namely,
\begin{equation}
\label{eq:Lambda_k}
\left(\frac{\Lambda_k}{2}\right)^2   =  \lambda^2 + 1 - 2 \lambda \cos k.
\end{equation}
This is (almost) the spectrum of our problem. We still need to find what are the allowed values of the label ``$k$''.
This is done by considering the two other cases ($j=1$  and $j=N$), that give a system of equations for $\alpha$ and
$\beta$:
\begin{eqnarray}
0& = & \;\;\;\alpha \;\;\; + \;\;\;\beta, \nonumber \\
0&=&\left(\lambda e^{ik(N+1)} -e^{i k N}\right) \alpha + \left(\lambda e^{-ik(N+1)} -e^{-i k N}\right)\beta.\nonumber
\end{eqnarray}
The non-trivial solution $\alpha=-\beta$ is obtained provided the determinant of this matrix vanishes. This condition
is equivalent to $k$ being solution of the following transcendental equation, for which $\lambda\neq 0$ has to be
assumed:
\begin{equation}
\label{eq:open-transcendental}
\sin kN = \lambda  \sin k(N+1).
\end{equation}
Notice that, since we must have $N$ eigenvectors, we expect this equation to have $N$ roots. The behavior of these
roots as a function of $\lambda$ is quite relevant; comparison with the periodic chain allows for the recognition of edge states.
\subsection{Periodic boundary conditions} In the case of periodic boundary conditions,
we extend the sums in  (\ref{eq:H-Ising-transv}) to $i=N$, adopting the convention that $\sigma_{N+1}^\alpha\equiv\sigma_{1}^\alpha$ ($\alpha=x,y,z$). This generates a boundary term that couples the first spin operator to the last one:
\begin{equation}
\label{eq:H-Ising-transv-pbc}
H = -  \sum_{i=1}^{N-1} \sigma_i^x \sigma_{i+1}^x  - \lambda  \sum_{i=1}^N \sigma_i^z -\sigma_N^x \sigma_{1}^x.
\end{equation}
We already know how to write the first two terms  of this Hamiltonian in terms of creation and annihilation operators. Let us therefore consider the last term: $\sigma_N^x \sigma_{1}^x$. Here it is convenient to consider the \emph{parity} operator,
$e^{i\pi \mathcal N}$,
where $\mathcal N= \sum_j a_j^\dagger a_j$ is the \emph{number} operator. Using the  identities
\begin{equation}
1- 2a_j^\dagger a_j   =  e^{i\pi a_j^\dagger a_j}, \;\;\; (e^{i\pi a_j^\dagger a_j})^2=1,\;\;\;
 e^{i\pi a_j^\dagger a_j}   (a_j^\dagger +a_j)  =   (a_j - a_j^\dagger ),\nonumber
\end{equation}
and  (\ref{eq:Wigner-Jordan-inverse}) we obtain:
\begin{eqnarray}
\sigma_N^x \sigma_{1}^x &=& \left( \prod_{m=1}^{N-1} (1- 2a_m^\dagger a_m)\right)  (a_N^\dagger +a_N) (a_1^\dagger +a_1) \nonumber\\
&=&\left( \prod_{m=1}^{N-1} e^{i\pi a_m^\dagger a_m}\right)(a_N^\dagger +a_N) (a_1^\dagger +a_1)\nonumber\\
&=& \left(\prod_{m=1}^{N} e^{i\pi a_m^\dagger a_m} \right) e^{i\pi a_N^\dagger a_N} (a_N^\dagger +a_N) (a_1^\dagger +a_1)\nonumber\\
&=& e^{i\pi \mathcal N} (a_N-a_N^\dagger ) (a_1^\dagger +a_1)\nonumber\\
&=& (a_N-a_N^\dagger ) (a_1^\dagger +a_1) e^{i\pi \mathcal N}. \nonumber
\end{eqnarray}
 With this we can write  $H$ in terms of fermionic operators:
 \begin{eqnarray}
 \label{eq:H-with-Bdry-term}
H & = & - \sum_{i=1}^{N-1} (a_i^\dagger a_{i+1}  + a_{i+1}^\dagger a_i +a_i^\dagger a_{i+1}^\dagger - a_i a_{i+1} )
+ 2 \lambda  \sum_{i=1}^N a_i^\dagger a_i \nonumber\\
   &  &\;\;\;\;\; - \lambda N+ \,e^{i\pi \mathcal N}(a_N^\dagger a_1 + a_1^\dagger a_N + a_N^\dagger a_1^\dagger +a_1 a_N).
\end{eqnarray}
An important fact (that can be easily checked) is that the parity operator commutes with $H$:
\begin{equation}
[H, e^{i\pi \mathcal N} ] =0. \nonumber
\end{equation}
Thus, it is possible to diagonalize  the Hamiltonian separately in sectors of even and odd numbers of ``particles''.
The eigenvalues of $e^{i\pi \mathcal N}$ are of the form $\sigma= \pm 1$, with the plus sign for states with an \emph{even} number of particles and the minus sign for states with an \emph{odd} number of particles.
We can, as in the previous case, write $H$ as a quadratic form in the fermion operators (\emph{cf.} (\ref{eq:H-Ising3})), the only difference being the explicit form of the matrices $A$ and $B$, for which we now get:
\begin{equation}
\label{eq:AB-periodic} \hspace{-0.36cm} A=
 \begin{pmatrix}
  2\lambda    &  -1  &  0  &  \cdots & \sigma \\
  -1    &   2\lambda   &  -1  &  0 & \cdots  \\
  0    &  -1  &   2\lambda  &  -1 &  \\
   \vdots   &   0   &  -1  &   2\lambda  & \cdots  \\
  \sigma &  \vdots                   &      &     \vdots   & \ddots
 \end{pmatrix},\;\;\;\;
 B=
 \begin{pmatrix}
  0    &  -1  &  0  &  \cdots & -\sigma \\
  1    &  0  &  -1  &  0 &  \cdots \\
  0    &  1  &  0  &  -1 &  \\
  \vdots  &   0   &  1  &  0 & \cdots  \\
   \sigma & \vdots                     &      &     \vdots   & \ddots
 \end{pmatrix}
 \end{equation}
In contrast to (\ref{eq:(A-B)(A+B)}) or (\ref{eq:(A+B)(A-B)matrix}), the matrix $(A\pm B) (A \mp B)$ is now  Toeplitz:
\begin{equation}
\label{eq:(A-B)(A+B)Toeplitz}
\frac{1}{4}(A\pm B)(A\mp B) =\left(
   \begin{array}{ccccc}
        1\hspace{-0.1cm}+\hspace{-0.1cm}\lambda^2   &    -\lambda   &  0  &\cdots &\sigma \lambda\\
           -\lambda        &   1\hspace{-0.1cm}+\hspace{-0.1cm}\lambda^2 &  -\lambda &  0 & \cdots \\
               0                & -\lambda  & 1\hspace{-0.1cm}+\hspace{-0.1cm}\lambda^2  & \ddots & \\
        \vdots &    0 & \ddots& \ddots & \hspace{-0.3cm}-\lambda\\
     \sigma\lambda & \vdots &  &  -\lambda &1\hspace{-0.1cm}+\hspace{-0.1cm}\lambda^2\\
   \end{array}
 \right)
\end{equation}
The eigenvalue problem (\ref{eq:(A-B)(A+B)}) can again be solved using the ansatz (\ref{eq:Psi_kl}). The eigenvalues are again given by (\ref{eq:Lambda_k}), with the difference that the allowed values of $k$ can now be explicitly given\footnote{This stands in contrast to the case of open boundary conditions, where they are given by the solutions of the transcendental equation (\ref{eq:open-transcendental}).}, as follows from the following exercise.
\begin{xca}\label{ex:1.5}
Show that, in the even parity sector ($\sigma = 1$), the allowed values of $k$ are given by
\begin{equation}
k_m = \frac{(2m+1)}{N} \pi,\;\;\;\;\;\;(m=0,1,\ldots, N-1),\nonumber
\end{equation}
whereas in the odd parity sector ($\sigma = -1$) they are given by
\begin{equation}
k_m = \frac{2m}{N}\pi ,\;\;\;\;\;\;(m=0,1,\ldots, N-1).\nonumber
\end{equation}
\end{xca}
\begin{xca}\label{ex:1.6}
Explain how the $k$'s can be made to take positive and negative values
 and how then we obtain an explicit solution for the matrix $\Phi$, of the form
\[
\Phi_{k,l} \sim \left \{  \begin{array}{c}
                          \sin (kl),\; k>0\\
                           \cos(kl),\; k<0.
                        \end{array}
\right.
\]
\end{xca}
It should by now be clear that the solution of the eigenvalue problem is much easier in the periodic case. We have chosen to discuss the open chain mainly because of the role that boundary effects play in connection to symmetry breaking  and also
for studies of surface effects~\cite{Henkel1999,Karevski2000}. Now we are going to take advantage of the translational symmetry in the periodic case. As discussed below, if instead of applying a Bogoliubov transformation right after transforming the spin operators to fermion operators via the Jordan-Wigner transformation, we perform a Fourier transformation as an intermediate step, we quickly obtain a much more  elegant solution, that can  be easily generalized to treat the limiting case of an infinite chain. Here we will again make use of the fact that the Hamiltonian commutes with the parity operator.

Looking back at (\ref{eq:H-with-Bdry-term}), we realize that it is possible to express $H$ in a more compact form if we introduce boundary conditions of the fermion operators, depending on the parity sector we are interested in. In fact, notice that if we choose $\sigma=-1$, then the boundary term can be included in the first sum, just by setting $a_{N+1}\equiv a_1$. But if we want to do the same in the sector $\sigma = +1$, we will need $a_{N+1}\equiv  -a_1$.
In view of our previous discussion, it is clear that the  choice $\sigma=-1$ is more convenient computationally (the matrices  $(A\pm B)(A\mp B)$ are circulant in that case).
 But if we are interested in the \emph{ground state}, then we must consider the
other choice ($\sigma=1$). This can be seen if we compare the ground state energies of the fermionic Hamiltonian
in each sector. As can be easily checked, the  lowest energy eigenvalue indeed comes from the sector  $\sigma=1$. The difference between them  decreases with increasing $N$ and eventually disappears in the thermodynamic limit. It is for this reason that
the boundary term is usually disregarded. We nevertheless feel that the price to pay if we keep all terms is actually very low, and in return we can have complete control over the spectrum and its degeneracies~\cite{Schmidt1974}.
\begin{xca}\label{ex:1.7}
Making use of the solution of the eigenvalue problem for both periodic and antiperiodic boundary conditions, plot the lowest energy eigenvalues for $\sigma=\pm 1$ as functions of $\lambda$ and $N$.
\end{xca}
We therefore define
\begin{equation}
a_{N+1}:= -e^{i\pi \mathcal{N}} a_1. \nonumber
\end{equation}
 In this way the Hamiltonian becomes
 \begin{equation}
 \label{eq:H-fermions-periodic}
H  =  - \sum_{i=1}^{N} (a_i^\dagger a_{i+1}  + a_{i+1}^\dagger a_i +a_i^\dagger a_{i+1}^\dagger - a_i a_{i+1}- 2 \lambda a_i^\dagger a_i )
- \lambda N.
\end{equation}
Notice that the sum now goes over  $1\leq i\leq N$. So we see that we can describe \emph{both} sectors using the \emph{same} quadratic form, the choice of parity being now encoded in the boundary conditions for the fermion operators.
It should be remarked that when we consider anti-periodic boundary conditions for the fermions we are still
considering periodic boundary conditions for the \emph{spin chain}. Anti-periodicity for the fermions in this context is only related  to the choice of a negative eigenvalue for the parity operator.

Let us now introduce the Fourier transformed operators
\begin{equation}
\label{eq:Fourier-dk}
d_k :=\frac{1}{\sqrt{N}} \sum_{l=1}^N a_l e^{-i \phi_k l}.
\end{equation}
The choice of the phases $\phi_k$ must be made in such a way that the CAR algebra is preserved, but in addition it has to imply the boundary condition $a_{l+N} = -e^{i\pi \mathcal N } a_l$. Therefore, the phases $\phi_k$ will also depend on the parity sector.
\begin{xca}\label{ex:1.8}
Show that in the \emph{even} parity sector ($\sigma=1, a_{l+N} = - a_l$ ), the phases are given by
\begin{equation}
\phi_k = \left(\frac{ 2k + 1}{N}\right) \pi,\nonumber
\end{equation}
whereas in the \emph{odd} parity sector ($\sigma=-1, a_{l+N} =  a_l$ ) they are given by
\begin{equation}
\phi_k = \left(\frac{ 2k + 2}{N}\right) \pi. \nonumber
\end{equation}
\end{xca}
A word of caution is perhaps in order: In principle, $k$ is an integer that takes values in the range $1\leq k\leq N$. Nevertheless, due to the translational symmetry, $k$ actually belongs to a ``Brillouin zone'', meaning with this that there is no difference
between $k$ and $k\pm N$. So, for instance, for $\sigma=1$ we have that $e^{i(\phi_k+\phi_{k'})}=1$ whenever
$2\pi (k+k'+1)/N$ is an integer or, in other words, when
\begin{equation}
\label{eq:mod-N} k+k'+1 =0 \;\;\;(\mbox{mod}\; N).
\end{equation}
For this reason, even if the allowed values of $k$ are originally in the range $\lbrace 1,\ldots, N\rbrace$, we will allow $k$
to take any integer value, provided we interpret this (for $\sigma=1$) in the sense of (\ref{eq:mod-N}).
With this  convention we then obtain, for example:
\begin{equation}
\phi_{-k-1} = -\phi_k.\nonumber
\end{equation}
A similar remark applies to the case $\sigma=-1$.
For the remaining part of this section we will confine ourselves to the even parity sector and so $\sigma=1$ will
be tacitly assumed.

We can now go back to the Hamiltonian (\ref{eq:H-fermions-periodic}) and insert there the new operators defined by (\ref{eq:Fourier-dk}). As a result, we obtain the following form of the Hamiltonian:
\begin{equation}
\label{eq:H-en-terminos-de-d_k}
H = \sum_{k=1}^N \left [  2(\lambda -\cos \phi_k) d_k^\dagger d_k
-i \sin \phi_k (d_k^\dagger d_{-k-1}^\dagger + d_k d_{-k-1}) -\lambda    \right].
\end{equation}
This form of the Hamiltonian is indeed very convenient as it suggests that a diagonalization by means of
a Bogoliubov transformation that only mixes $d_k^{(\dagger)}$ with  $d_{-k-1}^{(\dagger)}$ should be possible.
For this reason we propose the following transformation:
\begin{equation}
\label{eq:c_k}
c_k = \alpha_k d_k  + \beta_k d_{-k-1}^{\dagger},
\end{equation}
with $\alpha_k, \beta_k$ (possibly complex) coefficients to be found.  Imposing the condition $\lbrace c_k, c_{k'}^\dagger \rbrace=\delta_{k,k'}$ we readily obtain $|\alpha_k|^2+|\beta_k|^2=1$. Furthermore, it is easy to check that the choices
\begin{equation}
\alpha_{-k-1}   =  \alpha_k,\;\;\;\;\;\;\beta_{-k-1}   = -\beta_k, \nonumber
\end{equation}
enforce $\lbrace c_k, c_{k'} \rbrace = 0$. Inverting (\ref{eq:c_k}) and inserting the result in (\ref{eq:H-en-terminos-de-d_k}), we obtain:
\begin{eqnarray}
\label{eq:RL2}
\lefteqn{H = \sum_{k=1}^N\left(    c_k^\dagger c_{-k-1}^\dagger \Big[ 2(\cos\phi_k -\lambda)\alpha_k\beta_k
-i \sin\phi_k (\alpha_k^2 + \beta_k^2)\Big] + \mbox{h.c.} \right) {  } }   \nonumber\\
& & { } +2 \sum_{k=1}^N c_k^\dagger c_k \Big[ (\lambda -\cos \phi_k)(  |\alpha_k|^2-|\beta_k|^2) + i \sin\phi_k (\bar{\alpha}_k \beta_k  -   \alpha_k \bar{\beta}_k ) \Big] { }\;\;\;\nonumber\\
&&  { }+ \sum_{k=1}^N\Big(   2(\lambda -\cos\phi_k) |\beta_k|^2  + i \sin \phi_k ( \alpha_k  \bar{\beta}_k-   \bar{\alpha}_k\beta_k ) -\lambda \Big). { }
\end{eqnarray}
Vanishing of the coefficients in the first sum leads to
\begin{equation}
\label{eq:RL1}
\frac{\sin\phi_k}{\lambda-\cos\phi_k} = \frac{2 i \alpha_k \beta_k}{\alpha_k^2 + \beta_k^2}.
\end{equation}
Recalling that $|\alpha_k|^2+|\beta_k|^2=1$, we see that the right hand side of this expression simplifies if we parameterize
 $\alpha_k$ and $\beta_k$ as follows:
 \[
 \alpha_k = \cos\frac{\theta_k}{2},\;\;\; \beta_k = -i \sin\frac{\theta_k}{2}.
 \]
 In fact, with this choice (\ref{eq:RL1}) simplifies to
 \[
 \tan \theta_k = \frac{\sin \phi_k}{\lambda - \cos\phi_k}.
 \]
 We can now introduce a
 normalization factor $\Lambda_k$\footnote{In section \ref{sec:XYmodel} we will multiply  $H$ by a factor  1/2, in order to simplify some expressions.}
  by means of
 \[
 \sin\theta_k = \frac{\sin\phi_k}{(\Lambda_k/2)}, \;\;\;\;\;\;\cos\theta_k = \frac{\lambda-\cos\phi_k}{(\Lambda_k/2)}.
 \]
 This leads to
 \[
 \Lambda_k = 2 \sqrt{(\lambda - \cos\phi_k)^2 + \sin^2\phi_k}.
 \]
 With these definitions, the remaining terms in (\ref{eq:RL2}) simplify further yielding the desired diagonal
 form for the Hamiltonian:
 \begin{equation}
 \label{eq:RL3}
 H = \sum_{k=1}^N \Lambda_k \Big( c_k^\dagger c_k -\frac{1}{2}\Big).
 \end{equation}
 \begin{xca}\label{ex:1.9}
  Verify all the computations leading from (\ref{eq:H-fermions-periodic}) to (\ref{eq:RL3}).
 \end{xca}
 {\bf The Ground State.} Having brought the Hamiltonian to the diagonal form (\ref{eq:RL3}), we now proceed to find
 an expression for the ground state. Since $\Lambda_k$ has been chosen to be positive (or zero) for all values
 of $k$, the ground state $|\Omega(\lambda)\rangle$ will be given by the condition $c_k |\Omega(\lambda)\rangle=0,\;\;
 \forall k\in \lbrace 1,\ldots,N\rbrace$. Let $|0\rangle$ be the ``vacuum'' state for the operators $d_k$, defined by the condition $d_k|0\rangle = 0$. Observe now
 that if we define
 \[
 B_k:= \cos\frac{\phi_k}{2} + i \sin\frac{\phi_k}{2} d_k^\dagger d_{-k-1}^\dagger,
 \]
 then we get $c_k B_k |0\rangle = 0$ for all $k$. Since $[B_k,B_k']=0$, one would think that
 $\prod_{k=1}^N B_k |0\rangle$ does the job, but we must take into account that $B_{k'}= B_k$ whenever
 $k'+k+1=0\;(\mbox{mod}\,N)$.
 If $N$ is even, this means that \emph{each} $B_k$ appears \emph{twice} in $\prod_{k=1}^N B_k$ and then an easy
 calculation shows that $c_{k'}\prod_{k=1}^N B_k |0\rangle$ does not vanish. Hence, we must be careful to include in
 the product only one instance of each $B_k$.

To do this, we distinguish  two cases: $N$ even and $N$ odd. Restricting the domain of $\phi_k$ to the interval $(-\pi,\pi)$ we get,  for $N$ even, $0<\phi_k<\pi$ for $1\leq k\leq
N/2-1$ (also for $k=N$) and  $-\pi <\phi_k<0$ for $N/2 \leq k \leq N-1$. Now, for each $k$ such that
$0<\phi_k<\pi$, there is \emph{exactly} one $k'$ with $-\pi<\phi_{k'}<0$ and such that $B_k=B_{k'}$. In fact, let $k$
be such that $0<\phi_k<\pi$ and put $k'=N-k-1$. It follows that $-\pi<\phi_{k'}<0$ and $k'+k+1 = 0\; (\mbox{mod} N)$,
so that $B_k=B_{k'}$.  We therefore see that, \emph{for even} $N$, the ground state is given by
 \[
|\Omega(\lambda)\rangle=
   \prod_{0<\phi_k<\pi}B_{k} |0\rangle.
 \]
Writing $N=2 M$, we may as well consider $k$ to be such that $-M\leq k \leq M-1$. The ground state, then, takes the form
$|\Omega(\lambda)\rangle=
   \prod_{k=0}^{M-1}B_{k} |0\rangle.$

For the case $N$ odd, write $N=2M +1$. In this case all the $B_k$ appear twice in $\prod_{k=1}^N B_k$, \emph{except}
when $k=M$. When $k=M$ we have $\phi_k=\pi$. This, together with  the fact that $d_M^\dagger\equiv d_{-M-1}^\dagger$,
implies that $B_M\equiv 0$, so we must exclude it from the product.  In analogy to the previous case, for each $k$ for
which $0<\phi_k <\pi$ we can find exactly one $k'$ such that $-\pi< \phi_{k'}< 0$ and $B_k=B_{k'}$, namely $k'=2M-k$.
Now, for $0\leq K\leq M-1$, we have $0< \phi_k < \pi$, whereas for $M+1\leq k\leq 2M$, we have $\phi_k\in (-\pi,0)$.
The result is that, \emph{for odd} $N$, the ground state is given by
 \begin{equation}
 \label{eq:GS-Ising-odd}
|\Omega(\lambda)\rangle=
   \prod_{0<\phi_k<\pi}B_{k} d_M^\dagger |0\rangle.
 \end{equation}
 The operator $d_M^\dagger$ has to be included in (\ref{eq:GS-Ising-odd}) in order to ensure that $c_M|\Omega(\lambda)\rangle=0$ (recall that in this case
 $c_M=-id_M^\dagger$).
 \subsection{Entanglement properties of the ground state}
 \label{sec:entanglement-properties-GS}
 Having obtained the exact spectrum of the model, as well as the ground state both for open and periodic boundary
 conditions, it is now possible -in principle- to obtain all relevant correlation functions and relate them to physical
 observables like magnetization, susceptibilities, and so on. As is well-known, the quantum Ising chain does not
 present any phase transition at finite temperature, but it certainly is one of the paradigmatic examples of a system
 displaying a quantum critical point. In fact, in the thermodynamic limit there is a quantum phase transition
 occurring when the external field $\lambda$ approaches the critical value $\lambda_c=1$. In this notes we will review
 an approach to the study of this quantum phase transition from the point of view of \emph{entanglement}. We will also
 be interested in certain geometric and topological properties of the ground state that are relevant for this quantum
 phase transition.

 Let us start by considering the model (\ref{eq:H-Ising-transv}) for an open,  two-site chain ($N=2$). Putting
 $\alpha=\sqrt{1+4 \lambda^2}$, we find that the ground state is given in the spin basis by
 \begin{equation}
 \label{eq:GS-Ising-two-sites}
 |\Omega(\lambda)\rangle = \frac{1}{\sqrt{2 \alpha (\alpha - 2\lambda)}} \left(|++\rangle +
 (\alpha-2\lambda)|--\rangle\right).
 \end{equation}
 Taking the partial trace with respect to one of the sites, we obtain
 \begin{equation}
\rho \equiv \mbox{Tr}_1 |\Omega(\lambda)\rangle  \langle \Omega(\lambda)| =
 \frac{1}{2 \alpha (\alpha - 2\lambda)} \left(|+\rangle\langle +| + (\alpha -2\lambda)^2|-\rangle\langle
 -|\right), \nonumber
 \end{equation}
 which in terms of the basis $\lbrace |+\rangle, |-\rangle\rbrace$ takes the following simple matrix form:
 \begin{equation}\label{eq:rho-N=2}
 \rho =\left( \begin{array}{cc}
   \frac{1}{2} +\frac{\lambda}{\alpha} & 0 \\
   0 & \frac{1}{2} -\frac{\lambda}{\alpha} \\
 \end{array}
 \right).
 \end{equation}
 Figure \ref{fig:1.1} displays a plot of the eigenvalues of $\rho$ as a function of $\lambda$.
\begin{figure}[t!]
\centering
\includegraphics[width=0.6\linewidth]{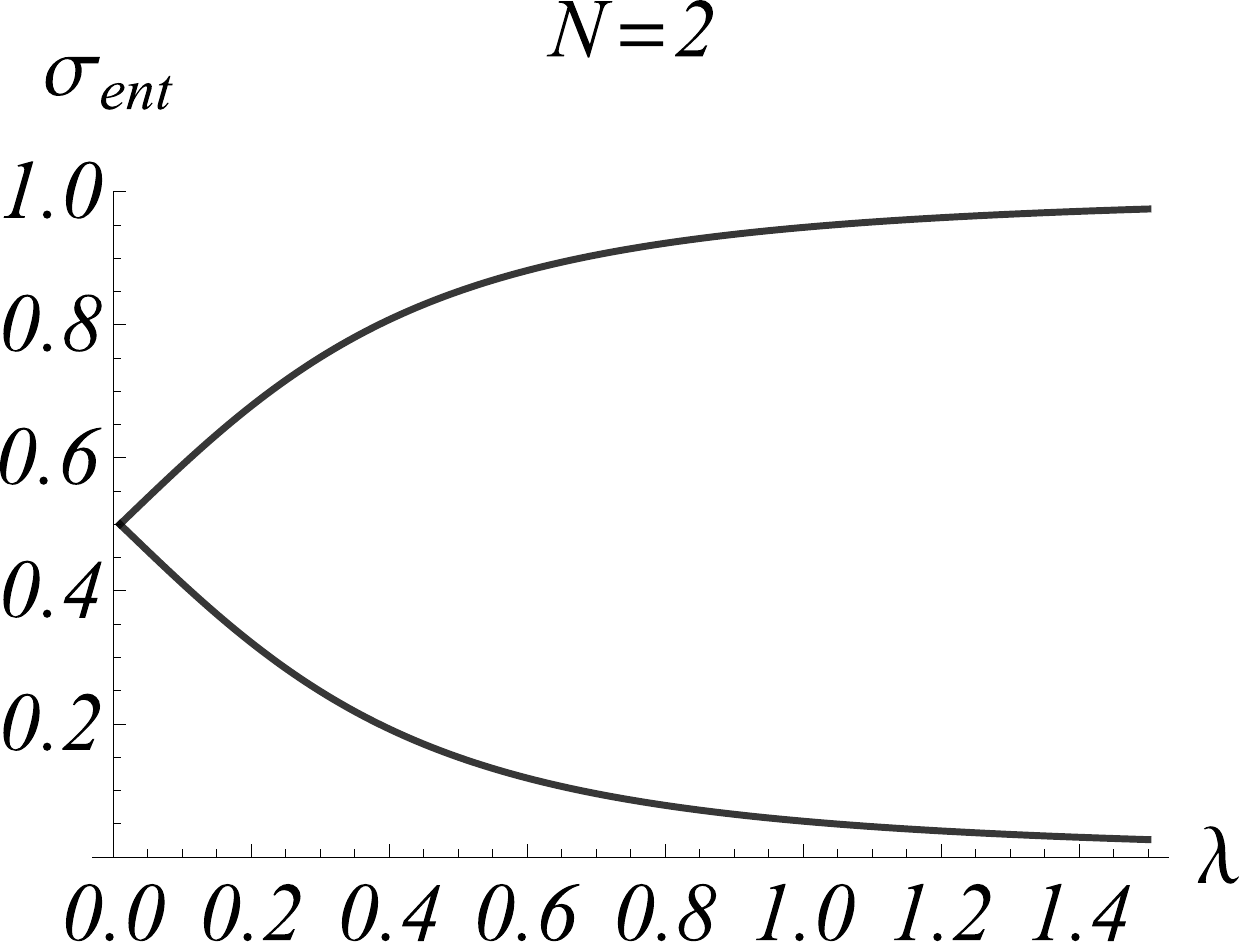}
\caption{Entanglement spectrum for two-site chain as a function of the external field  $\lambda$} \label{fig:1.1}
\end{figure}
As we shall see, the behavior of this \emph{entanglement spectrum} as the  size of the  chain increases will provide
important information regarding the quantum phase transition in this model.
\begin{xca}\label{ex:1.10}
Follow the steps outlined in section \ref{sec:OpenBoundaryConditions} to show that for $N=2$ the matrices $g$ and $h$
introduced in (\ref{eq:g-h-defined}) are given by
\[
g=\frac{1}{\sqrt{8\alpha}}\left(%
\begin{array}{cc}
  \frac{2\lambda +\alpha -1}{\sqrt{\alpha -1}} &  \frac{2\lambda +\alpha -1}{\sqrt{\alpha -1}}\\
  \frac{-2\lambda -\alpha -1}{\sqrt{\alpha +1}} & \frac{2\lambda +\alpha +1}{\sqrt{\alpha +1}}  \\
\end{array}%
\right),\qquad
h=\frac{1}{\sqrt{8\alpha}}\left(%
\begin{array}{cc}
  \frac{2\lambda -\alpha +1}{\sqrt{\alpha -1}} & \frac{\alpha -1-2\lambda}{\sqrt{\alpha -1}}\\
  \frac{-2\lambda +\alpha +1}{\sqrt{\alpha+1}} & \frac{\alpha +1-2\lambda}{\sqrt{\alpha +1}}  \\
\end{array}%
\right).
\]
Use this in order to show that the condition $c_k|\Omega(\lambda)\rangle\stackrel{!}{=}0\;\;(k=1,2)$ leads exactly to
(\ref{eq:GS-Ising-two-sites}).
\end{xca}
As we have seen, in this simple case ($N=2$) the density matrix for the reduced state can be easily computed. In spite
of the fact that the model can be exactly solved, the computation of correlation functions (as those involved in the
computation of the reduced density matrix) for arbitrary values of $N$ is a non-trivial task. We will therefore
present, following~\cite{Peschel2009}, a method that allows us to obtain the spectrum of $\rho$ for arbitrary values of
$N$ and that can be used for any quadratic Hamiltonian. As a preparation for the general case, we first explain the
idea using our very simple example of a two-site chain, where all calculations can be explicitly carried out.

Let $\mathcal A$ be the CAR algebra generated by the two fermionic operators $a_1,a_2$ and the unity $\mathds 1$. We
consider the ``Fock representation'' where $a_i^\dagger$ acts on the state $|0\rangle \equiv |++\rangle$. This state is
cyclic and therefore we can write the ground state (\ref{eq:GS-Ising-two-sites}) in terms of it (\emph{cf.}~exercise
\ref{ex:1.10}):
\begin{equation}
\label{eq:GS-N=2}
|\Omega(\lambda)\rangle = \frac{1}{\sqrt{2\alpha(\alpha-2\lambda)}} \left(|0\rangle +(\alpha-2\lambda)a^\dagger_1
a^\dagger_2|0\rangle\right).
\end{equation}
As discussed in section \ref{sec:algebras-and-states}, we can use it to define a state $\omega_\lambda$ in the \emph{algebraic sense}:
\begin{eqnarray}
\omega_\lambda: \mathcal A &\longrightarrow & \mathds C \nonumber\\
A & \longmapsto & \omega_\lambda(A):= \langle\Omega(\lambda)|\,A\,|\Omega(\lambda)\rangle.\nonumber
\end{eqnarray}
This is, of course, a pure state. Nevertheless, if we \emph{restrict} it to the subalgebra $\mathcal A_1$ that is
generated by $\mathds 1$ and $a_1$ (``a half-chain''), then we find that in general the resulting state $\omega_\lambda
|_{\mathcal A_1}$ will be a \emph{mixed} state. Let us consider this restriction in more detail. Defining
$\omega_{\lambda,1}:=\omega_\lambda |_{\mathcal A_1}$ we expect that for any $A\in \mathcal A_1$
\begin{equation}
\label{eq:A1-N=2} \omega_{\lambda,1}(A)= \mbox{Tr}(\rho A)
\end{equation}
 will hold, with
\begin{equation}
\rho=\frac{e^{-\varepsilon_1 a_1^\dagger a_1}}{Z},\nonumber
\end{equation}
$Z=\mbox{Tr} e^{-\varepsilon_1 a_1^\dagger a_1}$,
 both traces being  taken on the Hilbert space corresponding to the half-chain (which in this case has a
 basis given by $\lbrace |0\rangle, a_1^\dagger|0\rangle  \rbrace$). Finding $\varepsilon_1$ amounts to finding the
 spectrum of $\rho$ and this, in turn, can be used to compute the entanglement entropy. The specific form of $\rho$
 as the exponential  of a quadratic form is to be expected in view of the fact that the ground state is Gaussian.
 In order to find $\varepsilon_1$, we notice that it is possible to compute the left hand side of (\ref{eq:A1-N=2}).
 The only relevant case is $A=a_1^\dagger a_1$, so let us define $C_1:=\omega_{\lambda,1}(a_1^\dagger a_1)$. Using
(\ref{eq:GS-N=2}), we get:
\begin{eqnarray}
C_1 & = & \omega_{\lambda,1}(a_1^\dagger a_1) =  \omega_{\lambda}(a_1^\dagger a_1) \nonumber\\
    & = & \langle \Omega(\lambda) |a_1^\dagger a_1| \Omega(\lambda)\rangle\nonumber\\
    &=& \frac{1}{2} -\frac{\lambda}{\alpha},\nonumber
\end{eqnarray}
where we still keep the notation $\alpha=\sqrt{1+4\lambda^2}$. On the other hand, we have
\begin{equation}
\mbox{Tr}(\rho a_1^\dagger a_1)=\frac{1}{1+e^{\varepsilon_1}}.\nonumber
\end{equation}
It follows that
\begin{equation}
1-2C_1 =\frac{2\lambda }{\alpha} = \tanh \frac{\varepsilon_1}{2} \nonumber
\end{equation}
or, equivalently,
\begin{equation}
\varepsilon_1 = \ln \left( \frac{1/2 +\lambda/\alpha}{1/2 -\lambda/\alpha}\right). \nonumber
\end{equation}
Since $\rho$ is already diagonal in the basis $\lbrace |0\rangle, a_1^\dagger|0\rangle  \rbrace$, we obtain:
\[
\rho=\left(%
\begin{array}{cc}
  \frac{1}{1+e^{-\varepsilon_1}} &            0                   \\
                0                & \frac{1}{1+e^{\varepsilon_1}}  \\
\end{array}%
\right) = \left(%
\begin{array}{cc}
  \frac{1}{2} +\frac{\lambda}{\alpha} & 0 \\
  0 &  \frac{1}{2} -\frac{\lambda}{\alpha} \\
\end{array}%
\right),
\]
which coincides with (\ref{eq:rho-N=2}). Notice, in passing, that $C_1$ can also be obtained directly from $h$. In
fact, a small computation shows that $C_1 = (h^t h)_{11}$.

This last remark gives a hint towards the solution for the general case: As we shall see, the spectrum of $\rho$ in the
general case can be extracted from an eigenvalue problem that uses the matrices $h$ and $g$ as input. Let us, then,
start by considering the algebra $\mathcal A$ generated by $\mathds 1$ and by all operators $a_i$, with $i=1,\ldots,N$.
Again, the ground state $|\Omega(\lambda)\rangle$ gives rise to a state $\omega_\lambda:\mathcal A\rightarrow \mathds
C$. We are interested in the restriction of this state to the subalgebra $\mathcal A_L$ generated by $\lbrace \mathds
1,a_1,\ldots,a_L \rbrace$. By the same argument as before, we expect this new state
$\omega_{\lambda,L}:=\omega_\lambda|_{\mathcal A_L}$ to be of the form $\omega_{\lambda,L}(A)=\mbox{Tr}(\rho A)\;$
($A\in \mathcal A_L$), with
\begin{equation}
\rho =\frac{e^{-H_\rho}}{Z}, \;\;\;\;\; Z=\mbox{Tr}e^{-H_\rho} \nonumber
\end{equation}
and
\begin{equation}
H_\rho=\sum_{i,j=1}^L\left( a_i^\dagger K_{ij}a_j +\frac{1}{2}(a_i^\dagger M_{ij}a_j - a_i M_{ij}a_j) \right).\nonumber
\end{equation}
For the Ising model, the matrices $K$ and $M$ (yet to be determined) will be real, so that $K$ has to be symmetric, and
$M$ antisymmetric. Now we define~\cite{Peschel2009}, for $i,j=1,\ldots , L$ the correlation functions
\begin{equation}\label{eq:C_ij-F_ij}
C_{ij}:= \omega_{\lambda,L}(a_i^\dagger a_j),\;\;\;\; F_{ij}:= \omega_{\lambda,L}(a_i^\dagger a_j^\dagger).
\end{equation}
Using (\ref{eq:g-h-defined}) one shows that $C_{ij}=(h^t h)_{ij}$ and $F_{ij}=(h^t g)_{ij}$. Let us first consider the
simpler case $M=0$. In this case $H_\rho$ can be diagonalized by a simple transformation of the form
\begin{equation}\label{eq:b-gamma-a}
b_l = \sum_{i=1}^L \gamma_{li}a_i,\;\;\;\;l=1,\ldots,L,
\end{equation}
where the matrix $\gamma$ is such that $\gamma^t \gamma =\mathds 1_L$. Solving the eigenvalue problem
\begin{equation}
K|\varphi_l\rangle = \varepsilon_l |\varphi_l\rangle \nonumber
\end{equation}
and introducing the convention $|i\rangle= (0,0,\ldots,1,\ldots,0)$ ($1$ in the $i^{th}$ entry and $0$ everywhere
else), we set $\gamma_{li}:= \langle i| \varphi_l\rangle$, \emph{i.e.}, the rows of  $\gamma$ are the eigenvectors
$|\varphi_l\rangle$. We then obtain
\begin{equation}
H_\rho =\sum_{l=1}^L\varepsilon_l b_l^\dagger b_l.  \nonumber
\end{equation}
It follows that $Z=\prod_{l=1}^L (1+ e^{-\varepsilon_l})$, so that $\omega_{\lambda,L}(b_i^\dagger
b_j)=\delta_{ij}(1+e^{\varepsilon_i})^{-1}$. Using (\ref{eq:b-gamma-a}) we then obtain
$C_{ij}=\sum_{l=1}^L \gamma_{li}\gamma_{lj}(1+e^{\varepsilon_l})^{-1}$, which is equivalent to the matrix identity
$C=(1+e^K)^{-1}$. This can also be written in the following form:
\[
1-2C=\tanh\frac{K}{2}.
\]
If $M\neq 0$, we replace ($\ref{eq:b-gamma-a}$) by
\[
b_l =\sum_{i=1}^L (\gamma_{li} a_i +\eta_{li} a_i^\dagger),
\]
where now we require (\emph{cf.}~exercise~\ref{ex:1.3}):
\[
\gamma^t \gamma +\eta^t \eta =\mathds 1_L,\;\;\;\gamma^t \eta +\eta^t \gamma = 0.
\]
Once again, the matrices $\gamma$ and $\eta$ should be such that $\rho$ becomes diagonal, \textit{i.e.}, $H_\rho=\sum_{l=1}^L
\varepsilon_lb_l^\dagger b_l$ should hold. The eigenvalue problem to be solved is
\begin{eqnarray}
\label{eq:(K+M)(K-M)}
(K-M)|\psi_l\rangle &=& \varepsilon_l |\varphi_l\rangle, \nonumber\\
(K+M)|\varphi_l\rangle &=& \varepsilon_l |\psi_l\rangle.
\end{eqnarray}
Using these eigenvectors as rows for the matrices $\psi$ and $\varphi$ and putting
\[
\gamma=\frac{1}{2}(\varphi+\psi),\;\;\;\;  \eta=\frac{1}{2}(\varphi-\psi),
\]
we obtain the desired diagonal form for $H_\rho$. As in the previous case, we can now write all correlation functions
that involve the operators $a_i$ in terms of correlation functions that only involve the operators $b_l$.
\begin{xca}\label{ex:1.11}
Obtain the following two relations:
\begin{eqnarray}
C_{ij} & = & \sum_{l=1}^L \left(\frac{\gamma_{li} \gamma_{lj}}{1+e^{\varepsilon_l}} + \frac{\eta_{li}
\eta_{lj}}{1+e^{-\varepsilon_l}}\right),\nonumber\\
F_{ij} & = & \sum_{l=1}^L \left(\frac{\gamma_{li} \eta_{lj}}{1+e^{\varepsilon_l}} + \frac{\gamma_{lj}\eta_{li}
}{1+e^{-\varepsilon_l}}\right), \nonumber
\end{eqnarray}
and use them in order to prove the following identity:
 \[
 \frac{1}{1+e^{(K\pm M)}} -\frac{1}{1+e^{-(K\pm M)}} = 2C -\mathds 1_L \pm 2F,
 \]
where $C$ and $F$ are the $L\times L$ matrices with components given by (\ref{eq:C_ij-F_ij}).
\end{xca}
From this exercise we obtain, making use of (\ref{eq:(K+M)(K-M)}),
\begin{equation}
\label{eq:2C-2F}
(2C -\mathds 1_L - 2F)(2C -\mathds 1_L + 2F) |\varphi_l\rangle = \tanh^2 \frac{\varepsilon_l}{2} |\varphi_l\rangle.
\end{equation}
\begin{figure}[b!]
\centering
\includegraphics[width=0.6\linewidth]{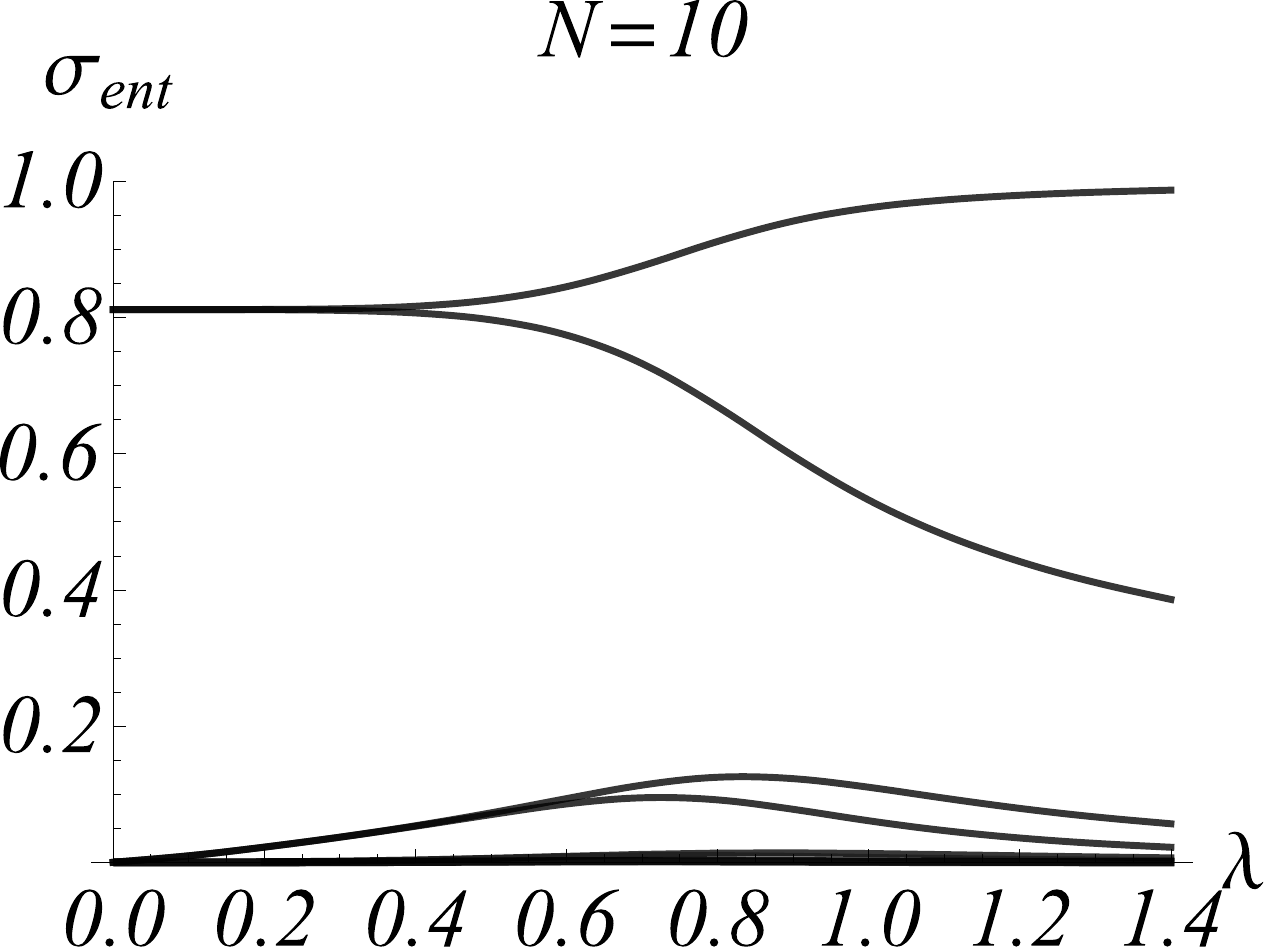}
\caption{Entanglement spectrum for an Ising  chain of $N=10$ sites as a function of the external field  $\lambda$}
\label{fig:1.2}
\end{figure}
We will now make use of this expression in order to obtain the entanglement spectrum and the entanglement entropy of
the ground state of the Ising chain when $L=N/2$, for various values of $N$ and $\lambda$. In Fig.~\ref{fig:1.2} we
plot the entanglement spectrum (in logarithmic scale) for a chain of $N=10$ sites, when the ground state is restricted
to a half-chain. We notice that, in contrast to the case $N=2$, the gap between the first two eigenvalues is now closed
for almost all values of $\lambda$ below $\lambda_c=1$.

In Fig.~\ref{fig:subfigures} we compare the entanglement spectrum for an Ising chain of $N=10$ sites against
the spectrum of a chain of $N=100$ sites. It is apparent that for $N=100$ the gap near $\lambda_c$ has almost closed.

\begin{figure}[t!]%
\hspace{0.8cm}
\parbox{1.2in}{\includegraphics[width=0.4\textwidth]{Fig-entangl-spec-Ising-N=10-nuevo}}%
\qquad \qquad \qquad
\begin{minipage}{1.2in}%
\includegraphics[width=1.5\textwidth]{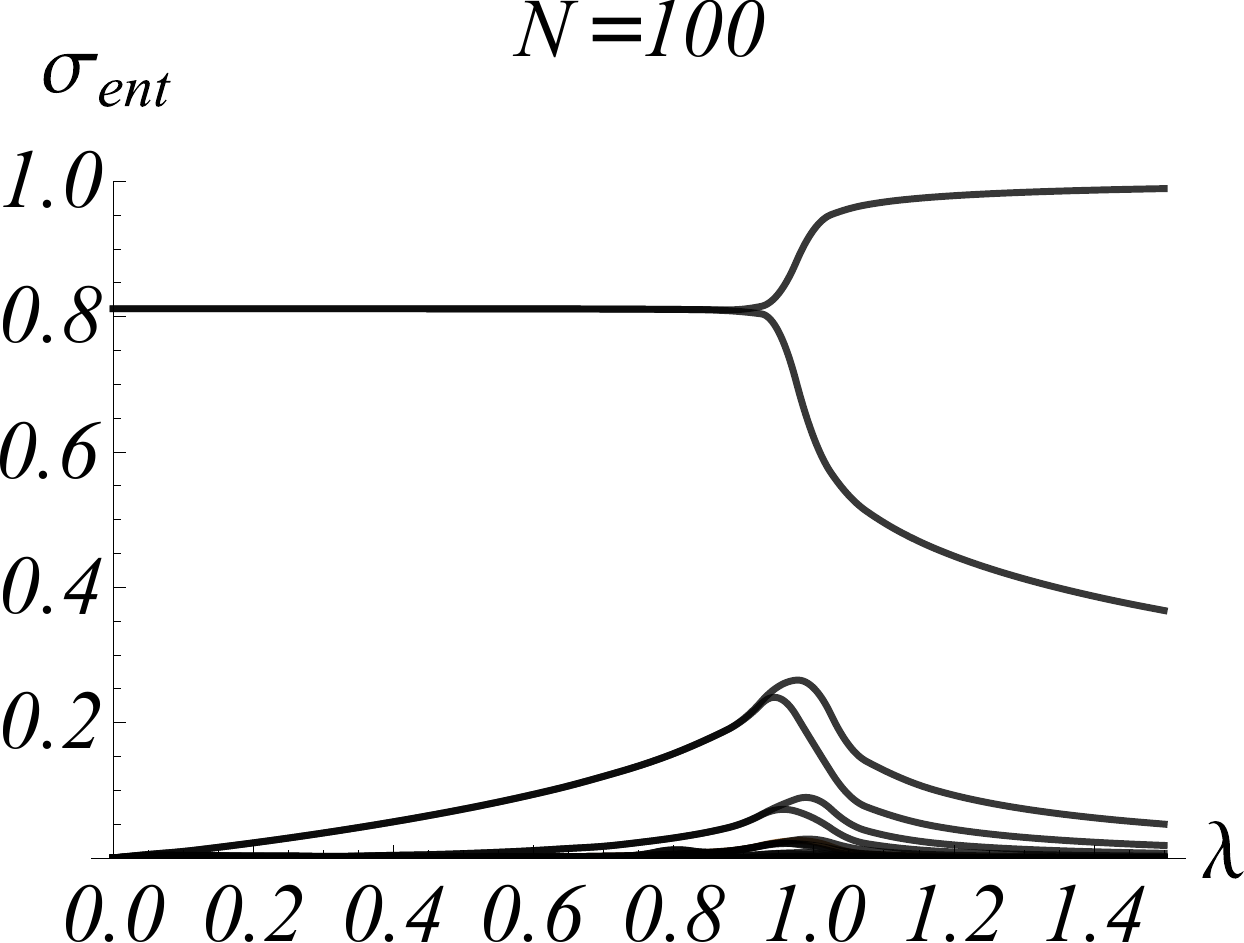}
\end{minipage}%
\caption{Entanglement spectrum for an Ising  chains of $N=10$ and $N=100$, showing how the gap at $\lambda=1$ closes as the number of sites increases.}%
\label{fig:subfigures}%
\end{figure}

This gap, when evaluated at the critical value $\lambda_c$,
bears the name of \emph{Schmidt gap}. We will use the notation $\Delta_S$ for the Schmidt gap. De Chiara and collaborators have shown, using finite-size scaling, that the Schmidt gap can be interpreted as an \emph{order parameter}~\cite{DeChiara2012}.
In fact, from a finite size scaling analysis, one can in fact obtain numbers $\mu_1$ and $\mu_2$ such that when we plot
$\Delta_S N^{\mu_1}$ versus $|\lambda-\lambda_c|N^{\mu_2}$, all points collapse to a single curve, irrespective of the value of $N$ chosen. This result was obtained in~\cite{DeChiara2012} using a DMRG algorithm. Figure \ref{fig:catalina}, obtained by C. Rivera using the same DMRG technique~\cite{Rivera2014} depicts the resulting curve, for which points corresponding to different lengths of the chain (ranging from $N\approx 700$ to $N=8000$) are seen to collapse to a single curve. The values of $\mu_1$ and $\mu_2$ for which this is achieved give  values for the critical exponents which are very close to the actual values ($\nu=1$ and $\beta=1/8$).
\begin{figure}
\centering
\includegraphics[width=0.8\linewidth]{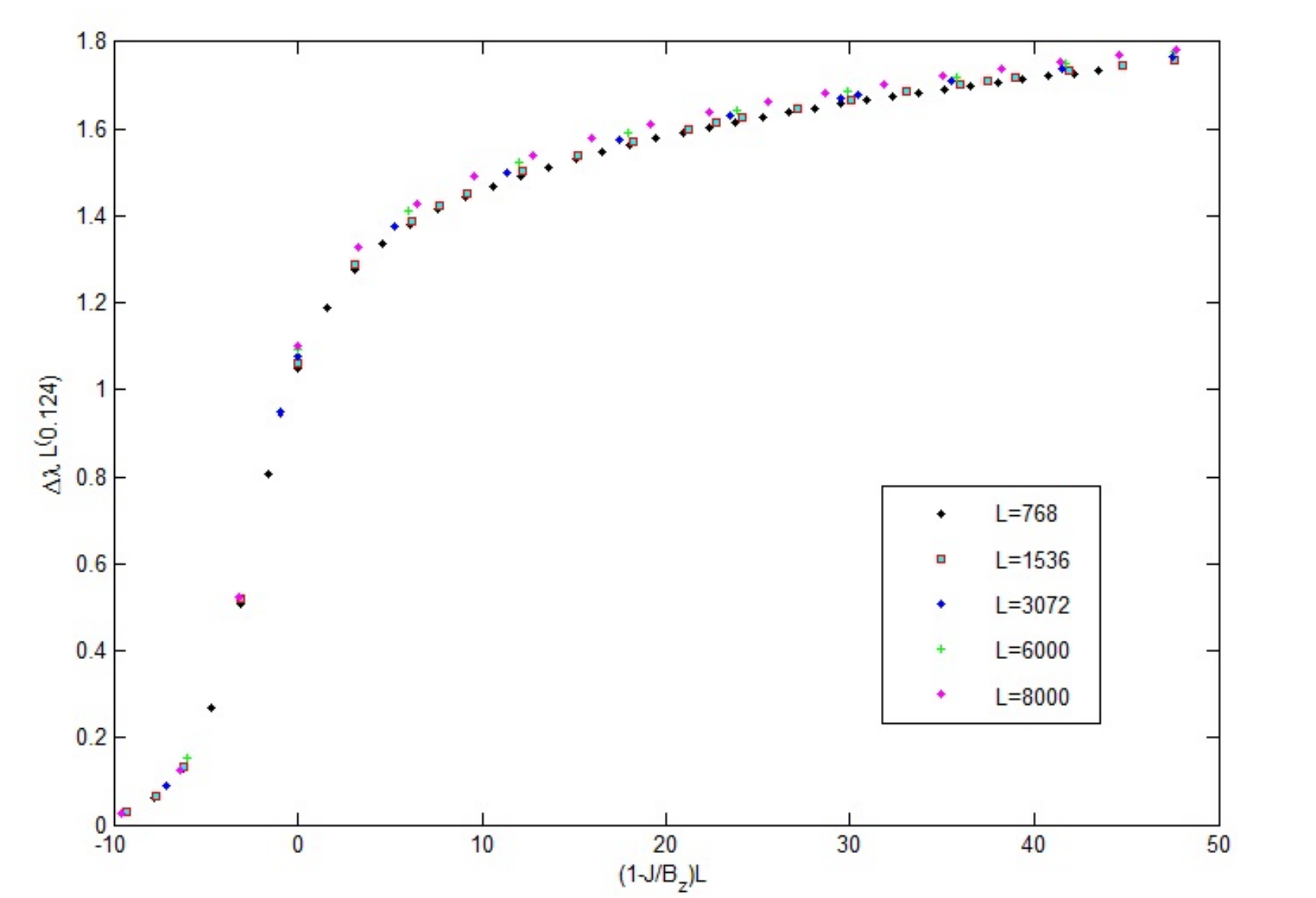}
\caption{Finite-size scaling}
\label{fig:catalina}
\end{figure}
The entanglement spectrum can also be used to compute the entanglement entropy  and its dependence on the size of the subsystem chosen.
A numerical evaluation of this entropy gives the result expected from conformal field theory, that is, at the critical point, the entanglement entropy follows a logarithmic behavior, of the form
\[
S(l) \sim \frac{c}{3} \log l,
\]
where $c$ denotes the central charge of the corresponding CFT. In this case (Ising chain), numerical evaluation of the entropy
for different chain sizes gives a value of $c=1/2$, in accordance with the predictions of CFT.
\subsection{The XY model}
\label{sec:XYmodel}
The XY model is a generalization of the Ising model. The model also represents a spin chain with nearest-neighbor interactions. Its Hamiltonian is given by
\[
 H= -\frac{1}{2} \sum_{j=1}^N \left(\frac{1+\gamma}{2} \sigma_j^x \sigma_{j+1}^x  + \frac{1-\gamma}{2}
\sigma_j^y \sigma_{j+1}^y +\lambda \sigma_j^z\right).
\]
There are two parameters in this model. As in the Ising case, $\lambda$ represents an external field. The parameter $\gamma$ is an anisotropy parameter. Notice that if $\gamma=1$, we recover, up to an overall factor of $1/2$, the Ising Hamiltonian (\ref{eq:H-Ising-transv}). Our interest in this model lies in the fact that according to the values of the parameters $\gamma$ and $\lambda$, we may study critical regions corresponding to \emph{different} universality classes: If $\gamma>0$, there is a critical line determined by $\lambda = 1$ that corresponds to the Ising universality class, with central charge $c=1/2$. On the other hand, for $\gamma=0$ the whole line $0<\lambda< 1$ corresponds to the $XX$ universality class, with central charge $c=1$.
\subsection{Solution of the model}
 The solution of this model can be obtained in the same fashion as we did in the case of the Ising model. The first step, then, consists in using the Wigner-Jordan transformation to bring the Hamiltonian to a quadratic form in fermionic operators. The result one obtains is
 \begin{eqnarray}
 \label{eq:H_XY}
 \lefteqn{H= -\frac{1}{2}\sum_j (a_j^\dagger a_{j+1} + a_{j+1}^\dagger a_j + \gamma a_j^\dagger a_{j+1}^\dagger
-\gamma a_{j} a_{j+1}) { }} \\
&& { } +\lambda \left( \mathcal N -\frac{N}{2}\right),{ }\hspace{5cm}\nonumber
\end{eqnarray}
  with $\mathcal N\equiv \sum_j a_j^\dagger a_j$, the fermionic number operator.
 Since we are interested in the ground state, we consider the even parity sector and proceed, as with the Ising model, with  a Fourier transformation, followed by a Bogoliubov transformation. For odd $N$ ($=2M+1$) and with
\[
d_k=\frac{1}{\sqrt N}\sum_{l=1}^N a_l e^{-i\phi_k  l},\;\; \;\;\phi_k:=\left( \frac{2k+1}{N} \right )\pi,\;\; -M\leq k
\leq M,
\]
we obtain:
\[
H =\sum_k ((\lambda -\cos \phi_k) d_k^\dagger d_k -\frac{i \gamma}{2} \sin \phi_k (d_k^\dagger d_{-k-1}^\dagger
+d_k d_{-k-1})) -\frac{\lambda N}{2}.
\]
We can now consider a Bogoliubov transformation, as in (\ref{eq:c_k}),
\[
c_k=\cos\frac{\theta_k}{2} d_k - i \sin\frac{\theta_k}{2} d^\dagger_{-k-1}, \;\;\; \tan\theta_k = \frac{\gamma \sin
\phi_k}{\lambda-\cos\phi_k},
\]
and obtain the following diagonal form for the Hamiltonian:
\[
H = \sum_{k=-M}^M \Lambda_k\left(c_k^\dagger  c_k  -\frac{1}{2}\right),
\]
with
\[
\Lambda_k =\sqrt{(\lambda - \cos\phi_k)^2 + \gamma^2\sin^2\phi_k}.
\]
Defining $B_k:= \cos\frac{\phi_k}{2} + i \sin\frac{\phi_k}{2} d_k^\dagger d_{-k-1}^\dagger$ we obtain the following explicit
form for the ground state $|\Omega(\lambda,\gamma)\rangle$:
\[
|\Omega(\lambda,\gamma)\rangle =\left(\prod_{0\leq k <M}B_k\right)d_M^\dagger|0\rangle,
\]
where $|0\rangle$ is the vacuum state with respect to the operators $d_k$.
 \subsection{Criticality}
The XY model has a rich phase diagram, and this provides an opportunity to check the behavior of entanglement entropy at the critical point and its relation to the central charge. In the previous section we focused on a finite-size scaling analysis for the Ising model, which in turn brought up the difficulty of having to compute correlation functions for finite size chains. In this section, we will only be interested in properties in the thermodynamic limit, \textit{i.e.}, we will discuss the behavior of the entanglement entropy of a sub-chain of length $L$, considered as a subsystem of a chain of infinite length, or the geometric phase of the ground state associated to closed loops around the critical point. Calculations in the thermodynamic limit are greatly simplified because discrete sums become integrals that can be readily evaluated. Let us now very briefly describe how to obtain the entanglement entropy of a sub-chain of size $L$ in the thermodynamic limit. We will follow the approach of Vidal \emph{et al.}~\cite{Vidal2003}, which amounts to the real (\textit{e.g.}~Majorana) version of the method discussed in section \ref{sec:entanglement-properties-GS}.
We begin by replacing the fermion operators ($a_i$) appearing in (\ref{eq:H_XY}) by new, Majorana fermion operators ($\tilde c_i$),
 defined as follows:
\begin{eqnarray}
\tilde c_{2n-1} &:=& (a_n+a_n^\dagger) \nonumber \\
\tilde c_{2n} &:=&i(a_n^\dagger-a_n). \nonumber
\end{eqnarray}
These Majorana operators satisfy commutation relations of the form $\lbrace \tilde c_m ,\tilde c_n\rbrace =
2\delta_{mn}$. Now, consider the correlation matrix $\langle \tilde c_m \tilde c_n\rangle = \delta_{mn} +i
\Gamma_{mn}$, where the
 expectation value is taken with respect to the ground state. The matrix $\Gamma$ can be expressed as a block matrix, with the block $(i,j)$  being given by the $2\times 2$ matrix $\Pi_{j-i}$  ($i$ and $j$ range from $0$ to $N-1$, $N$ being the size of the full chain), where
\[
\Pi_l :=\left(
          \begin{array}{cc}
            0 & g_l \\
            -g_l & 0 \\
          \end{array}
        \right).
\]
In the thermodynamic limit ($N\rightarrow \infty$) the coefficient function $g_l$ takes the following compact form:
\[
g_l=\frac{1}{2\pi} \int_0^{2\pi} d\phi e^{-il\phi} \frac{(\cos\phi -\lambda -i\gamma \sin\phi )}{|\cos\phi-\lambda -i\gamma \sin\phi|}.
\]
The restriction of the correlation matrix to the sites belonging to a sub-chain of size $L$ is then a $L\times L$ matrix $\Gamma_L$ which is obtained just by deleting the entries of $\Gamma$ that do not correspond to sites in the sub-chain.
\begin{xca}
Following the method of section \ref{sec:entanglement-properties-GS} as well as Vidal et al.~\cite{Vidal2003}, show that the entanglement entropy corresponding to the restriction of the full chain to a sub-chain of $L$ sites is given by the formula
\begin{equation}
\label{eq:S_L}
S(L)=\sum_{m=0}^{L-1} H\left(\frac{1+\nu_m}{2}\right),
\end{equation}
where $H$  denotes the Shannon entropy as defined in  (\ref{eq:H(lambda)}) and $\lbrace\nu_m\rbrace_m$ are the positive eigenvalues of $i\Gamma_L$.
\end{xca}
Figure \ref{fig:vidal} displays the behavior of the entropy as a function of $L$. One of the curves (dashed)
shows the results for the XX chain, which corresponds to $\gamma=0$. From the logarithmic dependence of the entropy
we extract the value $c=1$ for the central charge in this model. The continuous curve  shows the entropy
for the critical case $\lambda=1,\gamma=0.5$, which belongs to the Ising universality class. In this case, the
value obtained for the central charge is $c=1/2$. The dots correspond to the non-critical case $\lambda=0.5, \gamma=0.5$. We can
see that in this case the entropy saturates with $L$. Thus, the results confirm the expected behavior of the entropy as predicted by CFT.
 \begin{figure}[t!]
 \centering
\includegraphics[width=0.9\linewidth]{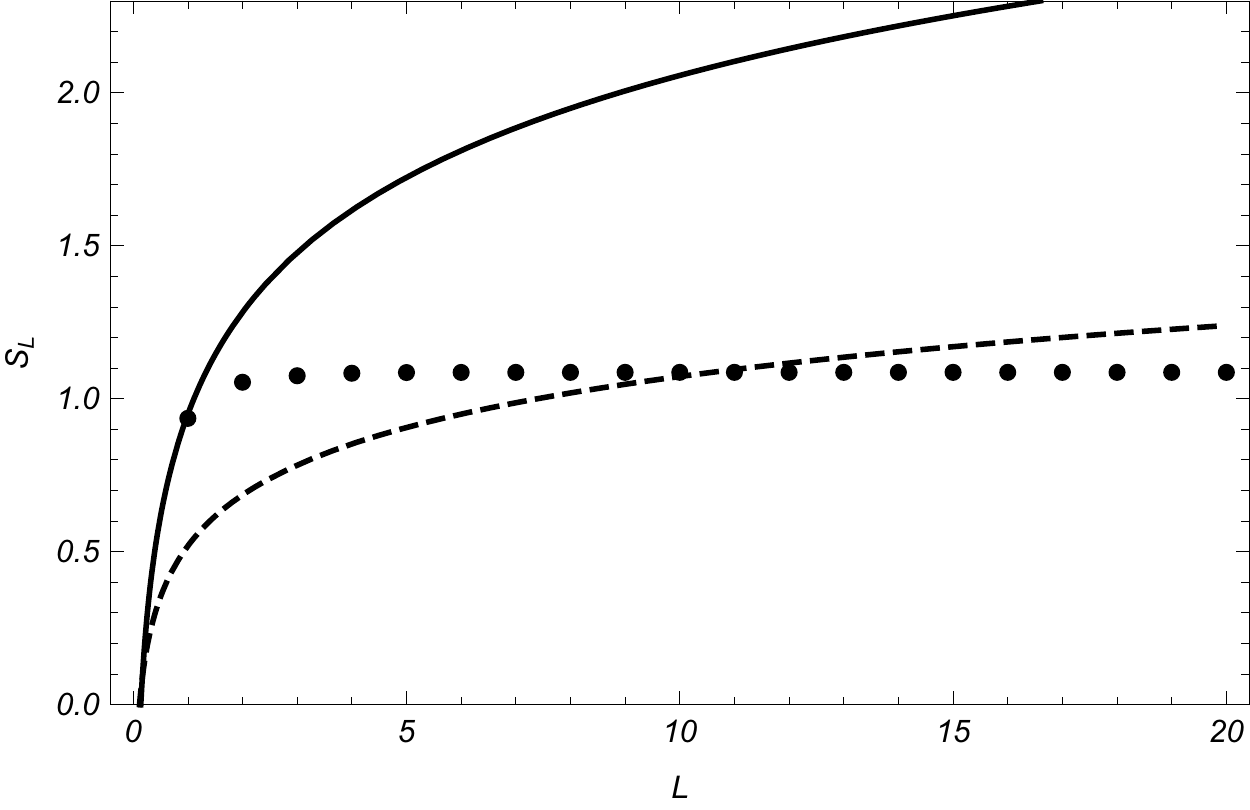}
\caption{Entanglement entropy as a function of the sub-chain size. $L$ is represents the size of the sub-chain. For model parameters
corresponding to critical phases, we observe a logarithmic behavior for the entanglement entropy. Outside the critical point, the entropy saturates.}
\label{fig:vidal}
\end{figure}




\section{Quantum fields}
\subsection{The scalar field}
In section \ref{sec:2} we mentioned that the \emph{canonical commutation relations} (CCR) lead to a $C^*$-algebra, whose representation theory plays a very important role in quantum field theory. Let us recall that the operators $U(a)$ and $V(b)$ defined in (\ref{eq:U(a)V(b)def}) satisfy the so-called Weyl form of the CCR (\emph{cf.}~exercise~\ref{ex:Weyl-CCR}). The first thing we want to do is to recognize that any implementation of the CCR in their Weyl form corresponds in fact to the construction of a representation of a certain group (the Heisenberg group). The Heisenberg group of $\mathbb R^n$, whose underlying set is $(\mathbb R^n\times\mathbb R^n\times \mathbb R)$, is the group defined by the following operation:
\[
(\vec a_1,\vec b_1,r_1)\cdot (\vec a_2,\vec b_2,r_2):= (\vec a_1+\vec a_2,\vec b_1+\vec b_2,r_1+r_2+\frac{1}{2}(\vec a_1\cdot \vec b_2-\vec a_2\cdot \vec b_1)).
\]
\begin{xca}
Study the Lie algebra of the Heisenberg group and explain how this Lie algebra is related to the CCR.
\end{xca}
\begin{xca}
Show that the operators $\mathcal R (\vec a,\vec b,r)$ defined through
\[
\mathcal R (\vec a,\vec b,r):= V(\vec b) U(\vec a)e^{i\hbar(r+\frac{1}{2}\vec a\cdot \vec b)},
\]
with $U$ and $V$ as in (\ref{eq:U(a)V(b)def}), furnish a representation of the Heisenberg group.
\end{xca}
The key point regarding these commutation relations is actually their relation to the (canonical) symplectic structure of the underlying classical phase space, $T^*\mathbb R^n$. In fact, if for $u=(\alpha,\beta)\in T^*\mathbb R^n$ we define
\begin{equation}
W (\alpha, \beta):= e^{-i(\alpha \hat q +\beta \hat p)}, \nonumber
\end{equation}
we obtain the following identity:
\[
W(u) W(v) = e^{-\frac{i}{2}\sigma(u,v)} W(u+v),
\]
where $u$ and $v$ denote  elements of the \emph{symplectic vector space} $T^*\mathbb R^n$, and $\sigma$ the standard symplectic form.
Notice that we have the following  relations between all these operators:
\begin{equation}
W(\alpha,\beta)=e^{i\frac{\hbar}{2}\alpha\beta} V(\alpha) U(\beta)= \mathcal R (\beta,\alpha, 0). \nonumber
\end{equation}
\begin{definition}{(\emph{cf.}~\cite{Moretti2012})}
Let $V$ be a real vector space and $\sigma:V\times V \rightarrow \mathbb R$ a symplectic form (\textit{i.e.} a bilinear, skew-symmetric, non-degenerate form). A $*$-algebra $\mathcal W(V,\sigma)$ is called a Weyl $*$-algebra of $(V,\sigma)$ if there is a family $\lbrace W(u)\rbrace_{u\in V}$ of ``generators'' such that
\begin{itemize}
\item[({\it i}\,)]
$W(u) W(v) = e^{-\frac{i}{2}\sigma(u,v)} W(u+v)$, $\;W(u)^*=W(-u),\;\;$ $u,v\in V$.
\item[({\it ii}\,)] $\mathcal W(V,\sigma)$ is \emph{generated} by the family $\lbrace W(u)\rbrace_{u\in V}$, \textit{i.e.}, it is the span of finite linear  combinations  of finite product of the $W(u)$.
\end{itemize}
\end{definition}
In~\cite{Moretti2012} it is proved that every symplectic vector space $(V,\sigma)$ determines  uniquely  a  Weyl $*$-algebra, up to
$*$-isomorphism. Also, if $\phi: V_1\rightarrow V_2$ is a symplectomorphism between two symplectic vector spaces, then the corresponding Weyl $*$-algebras are isomorphic. A most important fact is that $\mathcal W(V,\sigma)$ can be completed to a $C^*$-algebra, the Weyl $C^*$-algebra.

It is a fundamental result, due to Stone and von Neumann, that when the dimension of $V$ is finite (and hence necessarily even), there is essentially only one representation of the CCR, which can be taken to be the standard Schr\"odinger representation~\cite{Moretti2012}. But in infinite dimensions uniqueness is lost, and so \emph{inequivalent} representations do exist. This fact is closely  related to the non-uniqueness of a vacuum state for a free quantum field in a curved spacetime background.
The following example illustrates how the  construction of a (free, bosonic) quantum field runs in parallel to the
previous discussion of CCR, the ``only'' (!) difference here being that the symplectic vector space we are dealing with is infinite dimensional.
\begin{example}[The Klein-Gordon field, \emph{cf.}~\cite{Wald1994}]  Let $(M,g)$ denote a globally hyperbolic spacetime, and consider the Klein-Gordon equation on that background:
\begin{equation}
(\Box_g +m^2)\varphi=0. \nonumber
\end{equation}
Let $\Sigma_0$ denote a fixed  Cauchy surface, to which the initial data for the solutions of the Klein-Gordon equation will be referred and consider the following vector space, a \emph{space of solutions of the K-G equation}:
\begin{equation}
\mathcal S:= \lbrace \varphi\in C^{\infty}(M,\mathbb R):\; (\Box_g +m^2)\varphi=0\; \mbox{and}\;\varphi|_{\Sigma_0}\in C_0^\infty{(\Sigma_0)} \rbrace. \nonumber
\end{equation}
It is a most remarkable fact that the classical field equation comes equipped with a natural symplectic structure. In fact, if $\Sigma$ is any space-like (Cauchy) hypersurface, and $n$ denotes the unit normal, then on the vector space $\mathcal S$ defined above we can define the following symplectic form:
\begin{equation}
\sigma(\varphi_1,\varphi_2):= \int_\Sigma (\varphi_1 \nabla_\mu\varphi_2-\varphi_2 \nabla_\mu\varphi_1) n_\mu d\mbox{vol}_g. \nonumber
\end{equation}
For $f\in C_0^{\infty}(M)$, let $A$  and $R$ denote the advanced and retarded fundamental solutions of the K-G equation. That is,
for $f\in C^\infty_0(M)$ we have
\begin{equation}
(\Box_g+m^2)Af = f,\;\; (\Box_g+m^2)Rf = f, \nonumber
\end{equation}
where the support of $Af$ lies in the causal past of the support of $f$ (and analogously  for $Rf$). Then $E:= A-R$ gives a map
\begin{equation}
E:C_0^\infty(M) \rightarrow \mathcal S.  \nonumber
\end{equation}
The CCR for the quantum field (in its \emph{smeared form}) then take the form
\begin{equation}
\, [\hat \varphi (f), \hat \varphi (g)]=-i \sigma(Ef, Eg).  \nonumber
\end{equation}
Just as in  non-relativistic quantum mechanics, the operators obeying CCR  are unbounded, and so it is convenient to go to the Weyl form, which leads us back to the Weyl $C^*$-algebra, this time the one corresponding to the infinite dimensional vector space
$(\mathcal S,\sigma)$.
\end{example}
\subsection{Fermionic second quantization and Clifford algebras}
Representation spaces for the  Weyl algebra of the previous example are bosonic Fock spaces.
For the classical Dirac  field, the transition  to the corresponding quantum description
can be performed in a completely analogous way, only that in this case, instead of an antisymmetric form, a symmetric one is used which, again,
comes ``for free'' as part of the structure of the (vector space) of solutions of the classical field equation. It is an interesting fact that the existence of a natural bilinear antisymmetric form in the scalar field case (and of a symmetric one in the case of the Dirac field) leads to quantum fields that obey  the physically correct relation between spin and statistics.

In very general terms, quantization of free fields can be formulated as follows. Suppose we are given a real
vector space $V$ together with a non degenerate (anti-)symmetric
bilinear form $s:V \times V \rightarrow \mathbb R$.
If in addition we  choose a complex structure $J$ which is compatible with $s$ (meaning
$s(u,v)=s(Ju,Jv)$ holds), we can complexify and  obtain a complex vector space $V_J$.
 The bilinear form $s$ can then be used to define an inner product on $V_J$, making it a Hilbert space.
 According to whether the bilinear form is symmetric or antisymmetric, one then considers the antisymmetric
 or the symmetric subspace of the tensor algebra of $V_J$, call it $\mathcal{F}_J$. Endomorphisms  of $V_J$ may be
lifted to $\mbox{End}(\mathcal{F}_J)$ in a natural (\emph{i.e.} functorial) way.
These liftings are then interpreted as
the ``second quantization'' of observables (self-adjoint operators in $V_J$) and symmetries
(unitary operators).
\medskip\\
In this section we will  focus on the fermionic case, following the  beautiful treatment of the subject
presented in~\cite{Gracia-Bond'ia2001}, which makes strong use of Clifford algebras and spin group representations\footnote{We  urge the reader to consult~\cite{Gracia-Bond'ia2001} for details and explicit computations.}.
In the next section we will explore the connection between second quantization (as presented here) and  quantum criticality in spin chains (as discussed previously).
Let us start by  first considering  the simple case where $\dim V < \infty$.
If  $V$ is  a finite dimensional real vector space, and  $g$ a positive definite
symmetric bilinear form on it, we can construct the corresponding  Clifford algebra $C l(V,g)$.
 As shown in~\cite{Gracia-Bond'ia2001}, it is possible to obtain a concrete realization of the complexified
 Clifford algebra $\mathbb Cl (V)$ as a subalgebra of $\mbox{End}(\Lambda^\bullet V^\mathbb C)$, but
the corresponding representation is not irreducible. The situation changes if we assume $V$ has
even dimension, say $2m$. Then we may choose an orthogonal complex structure, that is, a linear
operator $J \in \mbox{End}_\mathbb R(V)$ such that $J^2=-1$ and $g(J u,Jv)=g(u,v)$ for all $u,v \in
V$, and regard $V$ as a complex  vector space $V_J$ of dimension $m$.
The  scalar product on $V_J$ is given by
$
\langle u|v\rangle_J := g(u,v) + i g(Ju,v).
$
It then turns out that the exterior algebra $\Lambda^\bullet V_J$ is an irreducible Clifford
module. We call this module  {\it Fock space} and use the notation $\mathcal{F}_J(V)$ for it.
The  action of $\mathbb Cl (V)$ on $\mathcal{F}_J(V)$ is obtained by defining ``creation" and
``annihilation" operators $ a^{\dagger}_J,\; a_J: V\rightarrow \mbox{End}(\mathcal{F}_J(V) )$
 acting on  $\mathcal{F}_J(V)$ as follows:
\begin{eqnarray}
 a_J(v) (u_1\wedge\ldots\wedge u_k) & := & \sum_{j=1}^k (-1)^{j-1}\langle v|u_j\rangle_J
u_1\wedge\ldots\wedge \widehat{u}_j\wedge \ldots\wedge u_k,\nonumber\\
 a^{\dagger}_J(v)(u_1\wedge\ldots\wedge u_k) & := & v \wedge u_1 \wedge\ldots\wedge u_k \,, \;\;\;\;\;\; a_J(v)(1)  := 0.\nonumber
\end{eqnarray}
 From
these definitions, we obtain the anti-commutation relations
\begin{equation}
  \lbrace a^{\dagger}_J(u),a_J(v)\rbrace   =   \langle u | v \rangle_J,\;\;\;\;\lbrace a_J(u),a_J(v)\rbrace  = 0=  \lbrace a^{\dagger}_J(u),a^{\dagger}_J(v)\rbrace,\nonumber
\end{equation}
from which $(a^{\dagger}_J(v)+a_J(v))^2= \langle v| v \rangle_J $ follows. Hence, the map
\begin{eqnarray}
\pi_J:V & \longrightarrow & \mbox{End}(\Lambda^\bullet V_J) \nonumber\\
 v &\longmapsto &\pi_J(v):= a_J(v) +a^{\dagger}_J(v)\nonumber
\end{eqnarray}
may be extended to $\mathbb Cl(V)$ by making use of the universal property of the Clifford algebra.
The scalar product $\langle \:|\:\rangle_J$ on  $V_J$ induces one on $\mathcal{F}_J(V)$, which
on basis elements is given by
\begin{equation}
\label{eq:scalar-prod}\langle u_1\wedge\ldots\wedge u_k| v_1\wedge\ldots\wedge v_{k'}\rangle
=\delta_{kk'}\det\langle u_i|v_j\rangle_J,
\end{equation}
thus making $\mathcal{F}_J(V)$ a (finite dimensional) Hilbert space. It may be checked that
$a^{\dagger}_J(v)$ is really the adjoint of $a_J(v)$, with respect to this scalar product.
 Our interest is to lift (or quantize)
operators $A\in \mbox{End}( V_J)$ to  operators acting on Fock space. In the present context,
this is done by means of maps $\Gamma, \; d \Gamma: \mbox{End}(V_J) \rightarrow
\mbox{End}(\mathcal{F}_J(V))$ defined as the graded operators obtained from
\begin{equation}
\Gamma^{(k)} A(u_1\wedge\ldots \wedge u_k):= A u_1\wedge \ldots
\wedge A u_k \nonumber
\end{equation}
and
\begin{equation}
d\Gamma^{(k)} A(u_1\wedge\ldots\wedge u_k):=\sum_{j=1}^k u_1 \wedge\ldots\wedge u_{j-1} \wedge
(A u_j)\wedge u_{j+1}\wedge\ldots\wedge u_k. \nonumber
\end{equation}
\begin{xca}
Check the correspondence between the definition of $d\Gamma$ given here and the one given by (\ref{eq:dGamma^k}) and
(\ref{eq:dGamma}).
\end{xca}
\begin{xca}[\textit{cf.}~\cite{Gracia-Bond'ia2001}]
 \label{ex:[dLambda A,v]=Av}
Derive the following identity:
$ [d\Gamma (A),a^{\dagger}(v)]=a^{\dagger}(Av)$.
\end{xca}
Given any $x\in V$ and $v \in V^\mathbb C$ any {\it unitary} vector, we see that
 $\phi(v) (x):=\chi(v)\cdot x \cdot v = x -2g(v,x)v$,  \textit{i.e.} the
twisted conjugation by $v$ using the grading automorphism $\chi$ is a reflection on the
hyperplane perpendicular to $v$. But reflections generate the orthogonal group $O(V,g)$ so that
the map $\phi$, being  defined for all invertible elements of the Clifford algebra,
 when  restricted to the subgroup generated by all products of an even number of unitary
 vectors ($\mbox{Spin}^c(V)$),   reduces to  a homomorphism onto $SO(V)$. This homomorphism
  may be restricted to $\mbox{Spin}(V)$.

Let us recall that there is a canonical vector space isomorphism between the Clifford algebra and the exterior
algebra, given explicitly in terms of an orthonormal basis by:
\begin{eqnarray}
\label{eq:Q}
Q : & \Lambda V^{\mathbb C} & \longrightarrow  \mathbb Cl(V) \\
 & v_1 \wedge\ldots\wedge v_k & \longmapsto  v_1 \cdot v_2\cdots v_k.\nonumber
\end{eqnarray}
Direct computation then shows that $\mathfrak{spin}(V)$ coincides
with $Q(\Lambda^2 V)$ and hence $Q(\Lambda^2 V)\cong \mathfrak{so}(V)$, the isomorphism being given by
\begin{eqnarray}
\rm{ad}: & Q(\Lambda^2 V)& \longrightarrow \mathfrak{so}(V) \nonumber \\
 & a& \longmapsto [a,\cdot \;]\nonumber
\end{eqnarray}
Let $c:\mathbb Cl(V)\rightarrow \mbox{End}(S)$ be an irreducible representation of $\mathbb Cl(V)$. Since
$\mathfrak{spin}(V)$ is realized as a subspace of $\mathbb Cl(V)$, we can compose $c$ with the
isomorphism $\rm{ad}^{-1}$ to obtain a representation, called the {\it infinitesimal spin
representation}, of $\mathfrak{so}(V)$ on $S$. Setting $\widetilde{B}:=ad^{-1}(B)$:
\begin{eqnarray}
\dot{\mu}:&\mathfrak{so}(V)&\rightarrow
Q(\Lambda^2 V)\hookrightarrow\mathbb Cl(V)\rightarrow\rm{End}(S)\nonumber\\
& B&\mapsto\hspace*{0.4cm}\widetilde{B}\hspace*{0.6cm}\mapsto\hspace*{0.2cm}
\widetilde{B}\hspace*{0.4cm}\mapsto\hspace*{0.2cm} c(\widetilde{B}),\nonumber
\end{eqnarray}
\textit{i.e.},
\begin{equation}
\dot{\mu}(B):=c(\widetilde{B}),\hspace*{1cm}B\in \mathfrak{so}(V).\nonumber
\end{equation}
The infinitesimal spin representation satisfies the following basic relation:
\begin{equation}
\label{basic_rel} [\dot{\mu}(B),c(v)]=c(Bv).
\end{equation}
Comparing  (\ref{basic_rel}) with the identity from exercise \ref{ex:[dLambda A,v]=Av}, one realizes that $d\Gamma$ and
$\dot{\mu}$ must be closely related. In order to express this relation, note that any real
linear operator $R$ on $V$ can be written as $R=R_{+}+R_{-}$, with $R_{+}$ linear in $V_J$ and
$R_{-}$ antilinear in $V_J$. For an element $B\in \mathfrak{so}(V)$, $B_{+}$ is skewadjoint and
$B_{-}$ antilinear and skewsymmetric (in $(V_J,\langle\;|\;\rangle_J)$).
Then, it can be shown that for any $B\in \mathfrak{so}(V)$, the following relation holds:
\begin{equation}
\label{eq:mu=dLambda+Tr} \dot{\mu}(B)=d\Gamma(B_{+})+\frac{1}{2}(a^{\dagger}B_{-}a^{\dagger}-aB_{-}a)
-\frac{1}{2}\rm{Tr}B_{+},
\end{equation}
where $a^{\dagger} B_- a^{\dagger}  :=  \sum_{k,l}\langle u_k| B_- u_l\rangle_J a^{\dagger}_k a^{\dagger}_l$
and $a T a  :=  \sum_{k,l} \langle T u_l | u_k \rangle_J a_l a_k$
(with $\lbrace u_k\rbrace_k$  any orthonormal basis in $V_J$)~\cite{Gracia-Bondia1994,Gracia-Bond'ia2001}.
This equation plays an important role for the definition of the quantization map in the
infinite dimensional case, discussed below.

In the infinite dimensional context, we begin with a separable real vector space $V$ on which a
positive definite symmetric bilinear form $g$ is defined. It is assumed that $V$ is complete in
the metric induced by $g$, so that $(V^\mathbb C,\langle\langle\;|\;\rangle\rangle\:)$ is a separable
Hilbert space, where $\langle\langle u|v\rangle\rangle:=2g(\bar{u},v)$.
 The Clifford algebra is constructed as follows.  Consider the algebra  $\mathbb Cl_{\mbox{\tiny fin}}(V)$ obtained from the  union of all algebras $\mathbb Cl(W)$, where $W$ runs through all finite
dimensional subspaces of $V$. There is a unique trace $\tau$
 on $\mathbb Cl_{\mbox{\tiny fin}}(V)$, inherited from the traces on each $\mathbb Cl(W)$.
 The scalar product induced by the trace makes $\mathbb Cl_{\mbox{\tiny fin}}(V)$ a prehilbert space.
  Its completion $\mathcal{H}_\tau$ allows one to represent $\mathbb Cl_{\mbox{\tiny fin}}(V)$
 as a subalgebra of $\mathcal{B}(\mathcal{H}_\tau)$ via the GNS construction.
  The closure of this algebra
 is then a $C^*$-algebra, which is defined to be the Clifford algebra $\mathbb Cl(V)$.
 The universal property remains valid in this context and in particular an orthogonal map
 $h\in O(V,g)$ extends to a $C^*$-algebra automorphism $\theta_h$ of $\mathbb Cl(V)$ (Bogoliubov automorphsim).

 In order to construct a representation of the Clifford algebra on Fock space we need, as before, an orthogonal complex structure $J$. The Fock space $\mathcal{F}_J(V)$ is now defined as
 the completion of the exterior algebra $\Lambda^\bullet V_J$ with respect to the scalar product
  (\ref{eq:scalar-prod}). The action of $\mathbb Cl(V)$ on the exterior algebra is defined as in the finite
  dimensional case, its extension giving rise to an irreducible representation
 $\pi_J:\mathbb Cl(V)\rightarrow \mbox{End}(\mathcal{F}_J(V))$. As in the finite dimensional case, the action of the orthogonal group $O(V)$ on the set
  of orthogonal complex structures is transitive, for  given two complex structures $J$ and $K$,
  there is a unitary map $h:V_J \rightarrow V_K$ such that $K=hJh^{-1}$. As can be shown,
  if the Bogoliubov automorphism $\theta_h$ is implementable by a unitary operator on Fock space,
  then  $\pi_J$ and $\pi_K$ are equivalent.  In general however, they will be inequivalent.
  The converse also holds: Implementability of $\theta_h$ follows from equivalence of the
  representations.
  Necessary and sufficient conditions on $h$ for $\theta_h$ to be implementable are afforded by
   the Shale-Stinespring theorem (\emph{cf.}~\cite{Gracia-Bond'ia2001}, thm. 6.16), namely that the antilinear part $\frac{1}{2}(h+JhJ)$ of
   $h$ be Hilbert-Schmidt. In terms of the representations, it says that $\pi_J$ and $\pi_K$
   are unitarily equivalent if and only if $(J-K)$ is Hilbert-Schimdt.
\begin{xca}
Provide a physical interpretation of the  Shale-Stinespring theorem. Hint: Try relating the vacua associated to two complex structures
$J$ and $K$.
\end{xca}
For the construction of quantized currents (like the charge or the number operator) we need a map
taking self-adjoint operators on $\mathcal H := V_\mathds C$ to self-adjoint operators on $\mathcal F_J (V )$. The explicit construction
of this map, which can  be found in~\cite{Gracia-Bond'ia2001}, will not be presented here. Instead, we will try to describe some of the motivations behind the construction.

Recall that $O(V)$ acts transitively on the set of orthogonal
complex structures, $\mathcal{J}(V)$. Since we have chosen a fixed complex structure $J$ to
construct the representation $\pi_J$ on Fock space, then any other complex structure will be
given in terms of $J$ and of an element $h\in O(V)$ ($J'=hJh^{-1}$). Recall also that an element of
the unitary group $U_J(V)$ is just an orthogonal map $h\in O(V)$  that commutes with $J$. A
unitary transformation $U\in U_J(V)$ is implemented in Fock space by a map $\Gamma_J$ such
that:
\begin{equation}
\Gamma_J(U)\pi_J(v)\Gamma_J(U)^{-1}=\pi_J(Uv).\nonumber
\end{equation}
Now, the observables of the 1-particle theory correspond to elements of the Lie algebra
${\mathfrak o}_J(V)$ of the restricted orthogonal group, which is defined as
$O_J(V):=\lbrace h\in O(V)\,|\,\left[h,J\right]\; \mbox{is Hilbert-Schmidt}\rbrace$. If an element
$X\in{\mathfrak o}_J(V)$ is such that $\left[X,J\right]=0$, then it follows that $JX$ is
self-adjoint on $V_J$ and it can be quantized by means of $d\Gamma_J$. The question of
(unitary) equivalence of the quantization obtained by means of $J$ on one side and by means of
$J'=hJh^{-1}$ on the other, may be formulated as follows. The complex structure $J'$ is
determined by $h$. The map $v\mapsto \pi_J(hv)$ extends to an automorphism of the Clifford
algebra and the question is then whether this automorphism is implementable on Fock space. So,
given $h\in O(V)$, we are looking for a unitary operator $\mu(h)$ on Fock space such that:
\begin{equation}
 \mu(h)\pi_J(v)\mu(h)^{-1}=\pi_J(hv).\nonumber
\end{equation}
In view of the inclusion $U_J(V)\subseteq O(V)$, we see that we are looking for an extension of
$\Gamma_J$ to $O(V)$. If $\left[J, h\right]$ is Hilbert-Schmidt, that is,
if $h\in O_J(V)$, then the map $\mu$ can be constructed (it is the pin representation of
$O_J(V)$ on Fock space). Its restriction to $SO_J(V)$ is the spin representation, from which an
infinitesimal version $\dot{\mu}$ can be defined on a dense domain of Fock space and turns out
to be given by (\ref{eq:mu=dLambda+Tr}) without the trace term and, therefore, it is not a Lie algebra
homomorphism. This gives rise the so-called \emph{Schwinger term}, defined as
\begin{equation}
\alpha(A,B):= [\dot \mu (A),\dot \mu (B)]-\dot \mu ([A,B]). \nonumber
\end{equation}
The Schwinger term is a cocycle that gives an obstruction for $\dot \mu$ to be a Lie algebra homomorphism. As explained in~\cite{Gracia-Bond'ia2001},
the Schwinger term  is related to anomalies in quantum field theory and also to noncommutative geometry.
 It appears naturally when studying the Virasoro algebra, and therefore is also related to
the notion of \emph{central charge} in conformal field theory.
 As we have seen, central charges do also appear when we study entanglement entropy
in spin chains that  display critical behavior.
Our aim in the next section is  to give an idea of how all these concepts are interrelated.
 This will allow us to close the circle by providing a new topological interpretation of quantum criticality,
 for which the underlying geometry  is not anymore the geometry of the space of external parameters,
 but a geometry in the sense of noncommutative geometry.
 This geometric interpretation of anomalies, central charges, Schwinger cocycles  \textit{etc.}, has been known for
  several decades (to the best of my knowledge the first person to realize this was Araki~\cite{Araki1987,Araki1988}).
  But the connection to  \emph{quantum critical phenomena} does not seem to have been exploited.
\subsection{Cyclic cocycles and quantum criticality}
Previously in these notes we discussed the Gelfand-Naimark theorem, which serves as a motivation for the notion of a noncommutative space.
It turns out that the Schwinger term described in the previous section also has a nice geometric interpretation, in the context of noncommutative geometry~\cite{Gracia-Bond'ia2001,Langmann2001}.
In order to give at least a glimpse of how this interpretation appears we will present, following~\cite{Connes1995}, a very basic (yet instructive) example, namely the Gauss-Bonnet theorem, in the language of noncommutative geometry.  But before that we quickly review some aspects of the geometry of surfaces, following the beautiful presentation of the subject by Pressley~\cite{Pressley2010}.
We thus start with a given 2-dimensional surface, described locally through a parametrization
\begin{eqnarray}
\sigma: U\subseteq \mathbb R^2 &\longrightarrow & \;\;\;\;\mathbb R^3 \nonumber\\
(u,v) &\longmapsto  & \sigma(u,v) = (x(u,v), y(u,v),z(u,v)).\nonumber
\end{eqnarray}
At each point, the tangent plane is generated by $\sigma_u\equiv \frac{\partial \sigma}{\partial u}$ and $\sigma_v\equiv \frac{\partial \sigma}{\partial v}$ and we can  define the normal unit vector as
\begin{equation}
n:=\frac{\sigma_u\times \sigma_v}{\|\sigma_u\times \sigma_v\|}.\nonumber
\end{equation}
The infinitesimal length element on the surface is given by
\begin{equation}
ds^2 =E du^2  + 2 F du dv + G dv^2,\nonumber
\end{equation}
where
\[ E=(\sigma_u,\sigma_u),\;F=(\sigma_u,\sigma_v), \; G=(\sigma_v,\sigma_v)\]
are the components of the \emph{first fundamental form}:
\begin{equation}
\mathcal F_I= \left ( \begin{array}{cc}  E & F \\ F & G \end{array}\right ).\nonumber
\end{equation}
Recall that the first and second fundamental forms enter the definition of Gaussian curvature of a surface, which we now recall, starting from the curvature of a curve.
For a plane curve, parametrized by $\gamma (t)$ with unit speed ($\|\dot \gamma(t)\|=1$), let $n(t)$ denote the  unit  normal at $\gamma(t)$. Then
the (signed) curvature $\kappa$ is defined through the relation $\ddot \gamma =\kappa n$. The curvature gives us an idea of how much the curve deviates from a straight line~\cite{Pressley2010}. In fact, we may compute
  \begin{equation}
  (\gamma(t+\Delta t )-\gamma(t ))\cdot n \approx (\dot \gamma(t) \Delta t +\frac{\ddot \gamma(t)}{2}\Delta t^2 +\cdots )\cdot n = \frac{1}{2}\kappa \Delta t^2. \nonumber
  \end{equation}

Now we repeat the same reasoning, but  with a surface. That is, we  want to compute $(\sigma(u+\Delta u ,v+\Delta v)-\sigma(u,v))\cdot n$, where now $n$ denotes the unit normal to the surface.
From a second order Taylor expansion and using the fact that $\sigma_u\cdot n=0=\sigma_v\cdot n$, we get
\begin{eqnarray}
\lefteqn{
\sigma(u+\Delta u ,v+\Delta v)-\sigma(u,v)\approx  { }} \nonumber\\
&& { }\hspace{0.0cm}\approx \frac{1}{2} (L \Delta u^2 +2 M\Delta u \Delta v + N\Delta v ^2 ) =
\frac{1}{2}(\Delta u ,\Delta v)\left( \begin{array}{cc}  L & M \\ M & N \end{array}\right) \left( \begin{array}{c}
                                \Delta u\\
                                \Delta v \end{array}  \right),{ } \nonumber
 \end{eqnarray}
where $L=\sigma_{uu}\cdot n$, $M=\sigma_{uv}\cdot n$ and  $N=\sigma_{vv}\cdot n$ are the coefficients of the
\emph{second fundamental form}:
\begin{equation}
 \mathcal F_{II}=\left( \begin{array}{cc}  L & M \\ M & N \end{array}\right).\nonumber
\end{equation}
Now, for any smooth curve $\gamma$ we have a decomposition of the form
\begin{equation}
\ddot \gamma = \kappa_n n + \kappa_g n\times \dot \gamma,\nonumber
\end{equation}
where $\kappa_n$ and $\kappa_g$ denote, respectively, the \emph{normal} and \emph{geodesic} curvatures of $\gamma$.
To define the curvature of the surface we need to consider only curves with $\kappa_g=0$, for which the normal curvature takes the form
\begin{equation}
\kappa_n = \ddot \gamma \cdot n = L \dot u^2 +2 M \dot u \dot v + N \dot v^2. \nonumber
\end{equation}
The value of $\kappa_n$ changes depending on the direction of the curve at the evaluation point, and it is clear it must attain maximum and minimum values. They are the \emph{principal curvatures}, denoted as $\kappa_1$ and $\kappa_2$.
The Gaussian curvature is defined as
$K:=\kappa_1 \kappa_2$. We can obtain several useful formulas for the curvature using the fundamental forms. The principal curvatures $\kappa_1$ and $\kappa_2$ are obtained from an extremization problem, where we look for extremal values of the principal curvature as we change the direction of the curve -given by $(\dot u, \dot v)$- subject to the unit-speed  restriction $\|\dot \gamma\|=1$.
So, with $x=(x_1,x_2)$, $x_1=\dot u, x_2=\dot v$, we want to extremize
\begin{eqnarray}
Q: \mathbb R^2 &  \longmapsto & \mathbb R \nonumber\\
       x       &  \longmapsto &  x^t\mathcal F_{II} x,\nonumber
\end{eqnarray}
 subject to the restriction $\|\dot \gamma\|^2=1$. But we also have
 \[ \|\dot \gamma\|^2 = E \dot u ^2 + 2F \dot u \dot v + G\dot v^2= x^t \mathcal F_{I} x.\]
\begin{xca}
Solve the extremization problem posed above, subject to the constraint $\|\dot \gamma\|^2=1$, and show that the solution
is given by the solution of the eigenvalue problem
\[
(\mathcal F_{I}^{-1} \mathcal F_{II})\, x = \lambda x,
\]
where the eigenvalue $\lambda$ is the Lagrange multiplier of the problem. Show furthermore that, for a given solution $(x,\lambda)$, the corresponding
normal curvature coincides precisely with $\lambda$ and conclude that the Gaussian curvature can be expressed as follows:
\begin{equation}
K=\kappa_1\kappa_2= \det \left( \mathcal F_{I}^{-1}\mathcal F_{II}\right)= \frac{L N - M^2}{E G -F^2}.\nonumber
\end{equation}
\end{xca}
There are other ways we can characterize $K$. One of them uses the Gauss map: For any point $p$ in the surface, consider the unit normal, $n_p$. Regarding this vector as a point in the unit sphere, we obtain a map $p\mapsto n_p$ from the surface to the unit sphere, the Gauss map. As we go around a small loop centered at $p$, the Gauss map produces an image loop on the unit sphere. The quotient of the areas enclosed by the loops converges, in the limit where the original loop shrinks to a point, to the Gaussian curvature.
Yet another way to obtain the Gaussian curvature is by means of the \emph{Weingarten map}. First notice that we have $n_u\cdot \sigma_u=-L$,
$n_u\cdot \sigma_v=n_v\cdot \sigma_u=-M$ and $n_v\cdot \sigma_v=-N$. This can be easily checked; for instance, the first relation is obtained by taking
 the partial derivative of  $\sigma_u\cdot n=0$ with respect to $u$. The others are obtained in the same way.

 The Weingarten map $\mathcal W$ is most easily defined in the basis $\lbrace \sigma_u, \sigma_v \rbrace$, where it has a matrix given
  by  $\mathcal F_{I}^{-1} \mathcal F_{II}$. In fact, if we define
 \[
 \mathcal W (\sigma_u)= -n_u, \; \mathcal W (\sigma_v)= -n_v
 \]
and write
\begin{eqnarray}
  -n_u & = & \lambda_{11} \sigma_u + \lambda_{12} \sigma_v \nonumber\\
  -n_v & = & \lambda_{21} \sigma_u + \lambda_{22} \sigma_v,\nonumber
\end{eqnarray}
then it follows that
\[
\Lambda:=\left(\begin{array}{cc}
        \lambda_{11} & \lambda_{12} \\
        \lambda_{21} & \lambda_{22}
      \end{array}
 \right)
 = \mathcal F_{I}^{-1} \mathcal F_{II}.
\]
Now we compute
\begin{eqnarray}
n_u\times n_v &= &(\det \Lambda ) \,\sigma_u \times \sigma_v \nonumber\\
&=& (\det \mathcal F_{I}^{-1} \mathcal F_{II} )\, \sigma_u \times \sigma_v, \nonumber
\end{eqnarray}
from which the following identity is obtained:
\begin{equation}
n \cdot (n_u \times n_v) = K \|  \sigma_u \times \sigma_v  \|.\nonumber
\end{equation}
Now, as the components of the unit normal $n=(n_1,n_2,n_3)$ are smooth functions of $(u,v)$,
 we can compute their exterior derivatives
\begin{equation}
d n_i = (\partial_u n_i) du + (\partial_v n_i) dv,\nonumber
\end{equation}
in order to obtain
\begin{equation}
\varepsilon_{ijk} n_k d n_i\wedge d n_j = 2 n \cdot (n_u \times n_v) du\wedge d v.\nonumber
\end{equation}
This leads to the following formula for the Euler characteristic:
\begin{equation}
\chi(\Sigma) = \frac{1}{4\pi} \int \varepsilon_{ijk} n_k d n_i\wedge d n_j.\nonumber
\end{equation}
This formula is particularly well suited to explain the noncommutative viewpoint, to which we now turn, following the example presented by Connes in~\cite{Connes1995}. Let us consider an algebra $\mathcal A$ (without assuming other properties, for the moment being) together with a trilinear map
\[
\tau: \mathcal A\times \mathcal A\times \mathcal A \rightarrow \mathbb C
\]
such that
\begin{equation}
\label{eq:cond(i)}
\tau(a_0,a_1,a_2)=\tau(a_2,a_0,a_1)
\end{equation}
and
\begin{equation}
\label{eq:cond(ii)}
\tau(a_0a_1,a_2,a_3)-\tau(a_0,a_1a_2,a_3)+\tau(a_0,a_1,a_2a_3)
-\tau(a_3a_0,a_1,a_2) =0.
\end{equation}
Now let $t\mapsto e_t$ be a (continuous) family of idempotents, $e_t^2 = e_t$. Then it can be shown that
the quantity $\tau(e_t,e_t, e_t)$ does not depend on $t$, providing an ``invariant'' of the algebra. A quick way to see what is going on is to assume that we can differentiate with respect to $t$. In this case we have
\begin{equation}
\dot e_t = [x_t,e_t],\nonumber
\end{equation}
with $x_t= [\dot e_t, e_t]$. Using this fact and  (\ref{eq:cond(i)}), we obtain
\begin{equation}
\frac{d}{dt} \tau (e_t,e_t,e_t) = 3 \tau ([x_t,e_t], e_t, e_t).\nonumber
\end{equation}
But the last expression vanishes exactly. To see this, we put $a_0= x_t$ and $a_1=a_2=a_3= e_t$ in  (\ref{eq:cond(ii)})
and use $e_t^2= e_t$. Now, coming back to Gauss-Bonnet, if we take the algebra to be $\mathcal A= C^{\infty}(\Sigma)$, we see that the trilinear map $(f_0,f_1,f_2)\mapsto \int_\Sigma f_0 df_1\wedge df_2$ satisfies  (\ref{eq:cond(i)}) and (\ref{eq:cond(ii)}) above, provided $\partial \Sigma = \emptyset$. But, as explained in~\cite{Connes1995}, in this algebra there are no interesting idempotents. This can be remedied by adding a ``bit of noncommutativity" to this algebra, that is, by replacing it by the following one: $\mathcal A = C^{\infty}(\Sigma) \otimes M_2 (\mathbb C ).$
Now, recalling the Gauss map, we realize that we can use the normal to the surface in order to obtain an idempotent, namely
$
e=1/2\,(\mathds 1_2 + n\cdot \sigma),
$
where $\sigma = (\sigma_1,\sigma_ 2, \sigma_3)$, a ``vector'' whose components are the Pauli matrices.
Then, if we define -following~\cite{Connes1995}-
\begin{equation}
\tau (f_0\otimes M_0, f_1\otimes M_1, f_2\otimes M_2):= \left(\int_\Sigma f_0 df_1\wedge df_2\right)\mbox{Tr} (M_0 M_1 M_2),\nonumber
\end{equation}
we obtain (an algebraic version of) the Euler characteristic.
\begin{xca} Prove the last statement, \textit{i.e.}, compute $\tau (e,e,e)$ and show that it is (up to a constant factor) precisely $\chi (\Sigma)$.
\end{xca}
What we learn from the previous example is that it seems plausible to encode geometry and topology in algebraic structures related to a manifold $M$.
This is precisely the way noncommutative geometry works. After establishing such algebraic and operator-theoretic characterization of certain geometric or topological structures, one can try to keep the algebraic version, but now dropping the commutativity assumption. Such is the case with differential calculus. Let us then  quickly list the main properties of de Rham cohomology and see whether they can be formulated for more general algebras. Given a compact  manifold $M$, the space of 1-forms  $\Omega^1(M)$ consists of elements of the form $f dg$, where $f$ and $g$ are smooth functions on $M$, and
$dg = \sum_i (\partial_i g) d x_i$ is the exterior derivative of $g$. Differential forms of order $k$,  $\Omega^k(M)$, are generated locally
by expressions of the form $dx_I\equiv dx_{i_1} \wedge \ldots \wedge dx_{i_k}$ so that, in general, the exterior derivative
$d: \Omega^k(M) \rightarrow \Omega^{k+1}(M)$ can be defined through
$d(f dx_I):=df\wedge d x_I$. The exterior derivative satisfies the Leibniz rule
\begin{equation}
d(\omega\wedge \eta) =d\omega\wedge\eta + (-1)^{|\omega|}\omega \wedge \eta \nonumber
\end{equation}
and fulfills the basic identity $d^2=0$, from which we obtain the (de Rham) cohomology ``groups'':
\[
H^k(M)= \mbox{Ker } d^{k+1}/\mbox{Im } d^k,
\]
that contain important information about the topology of $M$. Finally, we also have Stoke's theorem:
\[
\int_{\partial M} \omega = \int_M d\omega.
\]
Now suppose we are given a unital algebra $\mathcal A$. Thinking of it as given by some space of functions on a ``virtual'' (\textit{i.e.} noncommutative) space, we may try to define a differential calculus satisfying properties as close as possible to the ones recalled above. One possible construction goes as follows~\cite{Gracia-Bond'ia2001}:

In order to define the  space of 1-forms, we consider the map $d:\mathcal A\rightarrow \mathcal A\otimes \mathcal A$ defined as
$da:= 1\otimes a - a\otimes 1$. Its usefulness comes from the fact that it acts as a derivation:
\[
d(ab)=(da)b +a (db).
\]
But $\mathcal A\otimes \mathcal A$ is ``too big'' to be regarded as a space of 1-forms. We want to have a closer resemblance to $\Omega^1(M)$. Recalling
that 1-forms can always be put in the form $\theta = f dg$, we consider the restriction to the  subspace of $\mathcal A\otimes \mathcal A$ generated  by elements of the form $a\, db$. Noticing that
\[
a\, db= a(1\otimes b-b\otimes 1)= a\otimes b- ab\otimes 1,
\]
we arrive at the following definition:
\[
\Omega^1 \mathcal A := \mbox{Ker}(m)\cong \mathcal A\otimes \overline{\mathcal A}\;\;\;\;\;(\overline{\mathcal A}\equiv \mathcal A/\mathbb C),
\]
where $m: \mathcal A\otimes \mathcal A \rightarrow \mathcal A$ is the multiplication map $m(a\otimes b) = ab$. Then, starting from
$\Omega^1 \mathcal A=\mathcal A\otimes \overline{\mathcal A}$ we can go on and define $\Omega^k \mathcal A:=\mathcal A\otimes \overline{\mathcal A}^{\otimes  k}$
for $k\ge 1$. The exterior derivative, then,  is simply given by
\[
d(a_0\otimes \bar{a}_1 \otimes \cdots \otimes \bar{a}_k):= 1\otimes \bar{a}_0\otimes \bar{a}_1 \otimes \cdots \otimes \bar{a}_k.
\]
From these definitions we get, as can be easily checked:
\begin{itemize}
    \item[$(i)$] $d^2=0$ and Leibniz rule: $d(\omega_k\,\eta)= d\omega_k\,\eta +(-1)^k \omega_k\,d\eta$,
    \item[$(ii)$] $\Omega^k \mathcal A \cdot \Omega^l \mathcal A \subseteq \Omega^{k+l} \mathcal A$.
\end{itemize}
This means that in any complex algebra we can find an analog of the space of differential forms (\textit{i.e.}, a graded differential algebra).
The importance of this construction lies in the fact that it satisfies a universal property, as any derivation of $\mathcal A$ into a bimodule factors through a bimodule morphism from $\Omega^l \mathcal A$.

We also want to have a notion of integration. In this context this is achieved through the introduction of \emph{cycles}.
An $n$-dimensional cycle is a complex graded differential algebra
($\Omega^\bullet =\bigoplus_{k\leq n} \Omega^k, d$), along with  an ``integral'', \textit{i.e.} a linear map $\int: \Omega^\bullet\rightarrow\mathbb C$, such that:
\[
\int \omega_k =0\;(\mbox{for}\;\, k<n), \;\;\;
\int \omega_k \omega_l = (-1)^{kl} \int \omega_l \omega_k \;\;\;\mbox{and}\;\;\;
\int d \omega_{n-1}=0.
\]
The third requirement is tantamount to Stokes' theorem for the case of spaces without boundary. Now, given an algebra
 $\mathcal A$, a \emph{cycle} over $\mathcal A$ is a cycle $(\Omega^\bullet,d, \int)$ together with a
homomorphism $\rho:\mathcal A\rightarrow \Omega^0$~\cite{Gracia-Bond'ia2001,Connes1995}. A basic example is afforded by the de Rham complex on a compact manifold without boundary.

Given a cycle over $\mathcal A$ we may now define, for $a_0,a_1,\ldots,a_n\in \mathcal A$,
\[
\tau(a_0,\ldots,a_n):= \int \rho(a_0) d(\rho(a_1))\cdots d(\rho(a_n)).
\]
As shown in~\cite{Connes1995}, this map satisfies  (the generalization to $n$ of) the relations (\ref{eq:cond(i)}) and (\ref{eq:cond(ii)}), \textit{i.e.}, it is a \emph{cyclic cocycle}.
\begin{definition}
Given a complex algebra $\mathcal A$, an \emph{odd} Fredholm module over $\mathcal A$ is given by a representation
$\sigma:\mathcal A\rightarrow \mathcal B (\mathcal H)$
on the algebra of bounded operators on a Hilbert space, together with a self-adjoint operator $F:\mathcal H\rightarrow \mathcal H$, such that $F^2=\mathds 1$ and such that $[F,\sigma(a)]$ is a compact operator for each $a$ in $\mathcal A$. An \emph{even} Fredholm module contains, in addition to the structure of an odd Fredholm module, an additional grading of $\mathcal H$, say $\gamma$, such that $\gamma F=-\gamma F$ and $[\gamma ,\sigma(a)]=0, a\in \mathcal A$.
\end{definition}
On a Fredholm module one can introduce the  structure of a graded differential algebra defining the differential as
\[
d(\sigma(a)):= i[F, \sigma(a)].
\]
Letting $\Omega^1$ be spanned by elements of the form $a_0 da_1$, and so on, we arrive at the formula $d\omega= i[F,\omega]_{\pm}$ as a suitable
definition for $d$, where the bracket denotes a supercommutator (depending on the degree of $\omega$).

\begin{example}[Quasi-free Representations]
An important class of examples for cyclic-cocycles is afforded by the so-called Schwinger terms (as defined in the
previous section) in quantum field theory. In the context of charged fermionic fields, we start with a real vector
space $V$ (given as a certain space of solutions to the classical field equation), a symmetric bilinear form $g$ on $V$
and a  compatible (orthogonal) complex structure $J$. Physically, $J$ is determined by the spectrum of the Hamiltonian.
From the data  $(V, J, g)$ one then defines a complex Hilbert space which is the complexification of $V$ using $J$ as
the complex structure. So we get a Hilbert space $\mathcal{H}=(V_J, \langle\cdot|\cdot\rangle_J)$, where the inner
product is given by
\[
\langle u|v\rangle_J:= g(u,v)+ i g(Ju, v).
\]
We also have a representation of the field algebra on Fock space,
\[
\hat{\pi}_J : \mathbb{C} l (V)\rightarrow \mbox{End}(F_J(V)).
\]
Now, the point is that the Schwinger term
\begin{equation}
\label{eq:Schwinger-term}
 \alpha(A,B)=[\dot \mu(A),\dot \mu(B)] -\dot \mu[A, B ]= \frac{i}{8}\Tr \left(J[J,A][J,B]\right),
\end{equation}
defines a cyclic 1-cocycle. A geometric interpretation of this cocycle is then made possible in the context of
non-commutative geometry described above.
\end{example}

When discussing the entanglement properties of spin chains,  we found that the entanglement entropy of a sub-chain of length $L$, at the critical point,  scales logarithmically with $L$. We mentioned that the constant in front of the logarithm was related to \emph{central charge} of the corresponding conformal field theory. These central charges can also be understood as anomalies appearing in the quantization process of the effective field theories describing the spin chains. As an example, let us consider the Ising chain:
\[
H  =  -\frac{1}{2}\sum_{n=1}^N\left( \sigma_n^x \otimes\sigma_{n+1}^x
+\lambda \sigma_n^z\right).
\]
Using the Wigner-Jordan transformation and then defining the (Majorana) operators~\cite{Fradkin2013}
\[
\chi_{1}(n)=(a_n^\dagger+a_n),\;\;\; \chi_{2}(n)=i(a_n^\dagger-a_n),
\]
we obtain
\[
H  =\frac{i}{2} \sum_{n} \big( \chi_2(n) \chi_1(n+1) + \lambda \chi_1(n) \chi_2(n) \big).
\]
We can go to the continuum limit by assuming the sites in the chain are separated by a unit length $l_0$ and defining
$x_n=l_0 n$. Then, to first order in $l_0$, we have:
\[
\chi_\alpha(x_{n\pm1})\approx \chi_\alpha(x_{n}) \pm l_0\partial_x\chi_\alpha(x_n).
\]
Defining the ``chiral'' components $\chi_\pm =(\chi_1\pm\chi_2)/2$
and the corresponding spinor
\[
\psi= \left(\begin{array}{c}\chi_+\\ \chi_-\end{array}\right),
\]
one obtains, using Heisenberg's equations of motion,
\begin{equation}
\label{eq:Dirac}
\left(  i\gamma^\mu \partial_\mu -m  \right) \psi = 0,
\end{equation}
where $m=m(\lambda)$ is such that $m(\lambda_c)=0$.
\begin{xca} Defining light-cone coordinates $x_{\pm}= x^0\pm x^1$, show that, at the critical point ($m=0$), the Dirac equation
(\ref{eq:Dirac}) is equivalent to $\partial_+\chi_- = 0 = \partial_-\chi_+$.
\end{xca}
\begin{xca}
In the massless case, (\ref{eq:Dirac}) can be obtained from the Lagrangian $\mathcal L = i\overline{\psi}\gamma^\mu \partial_\mu\psi$. Considering the spatial coordinate $x_1$ to be defined on $S^1$ (which amounts to imposing periodic boundary conditions on the spin chain), expand the chiral fields in Fourier modes and then express the Hamiltonian density $\mathcal H = \dot\psi\frac{\partial \mathcal L}{\partial \dot \psi }-\mathcal L$, and the energy function
$H=\int_{S^1}d x_1\mathcal H(x_0,x_1)$ in terms of the Fourier modes.
\end{xca}

As the previous calculations show,  when $m=0$ the two chiral fields are completely decoupled from each other. Thus, we may study them separately.
Let us consider $\chi_+$, for which the mode expansion must be of the form
\[
\chi_+(t,x)= \sum_{n\in \mathbb Z} a_n e^{i n(t+x)}.
\]
Imposing canonical anticommutation relations to the field (which is real), we obtain:
\[
\lbrace a_n,a_m\rbrace= \delta_{n+m,0},\;\;\;\; a_n^*=a_{-n}.
\]
From the previous exercise, we see that the  contribution of $\chi_+$ to the (second quantized) Hamiltonian will be
\[
H_+= \sum_{n\in \mathbb Z} n  :a_{-n}a_n: \;,
\]
where normal ordering has been introduced. Now, how do we define normal ordering depends on what we regard as the ``vacuum state''. This has important consequences for the commutation relations of the (Fourier components of the) stress-energy tensor.
\begin{xca}
Decomposing the stress-energy tensor
\[
T_{\mu\nu}= \frac{1}{4}(\overline{\psi}\gamma_\mu \overset{\leftrightarrow}{\partial_\nu} \psi +
\overline{\psi}\gamma_\nu \overset{\leftrightarrow}{\partial_\mu} \psi)
\]
in light-cone components $T_{++},T_{--},T_{-+}$ and $T_{+-}$, show that $T_{+-}=T_{-+}=0$,  and that the Fourier components ($L_n$)
of $T_{++}$ satisfy the following  commutation relations (Virasoro algebra):
\[
 \left[ L_m, L_n\right] = (m-n) L_{m+n} + \frac{c}{12} (m^3-m)\delta_{m+n,0},
\]
with $c=1/2$.
\end{xca}

The correct solution of the previous exercise depends crucially in a careful handling of the normal ordering prescription. This in turn
is related to the ``filling-up of the Dirac sea'' which, in mathematical terms, can be reduced to the problem of quantizing a linear system with an appropriately chosen complex structure. In fact, writing the Hamiltonian density in terms of the spinor $\psi$, we obtain
\[
\mathcal H= -i \psi^\dagger \sigma_3 \partial_x \psi.
\]
This means that the 1-particle Hamiltonian (which is defined on the Hilbert space $L^2(S^1)\otimes \mathbb C^2$) is given by $H^{(1)}= -i\sigma_3 \partial_x$. As we have seen,
the splitting of the 1-particle Hilbert space into positive and negative energy states gives rise to a complex structure, which is the one that should be used
for quantization. When this is done, one obtains the anomalous term of the Virasoro algebra as an anomaly (Schwinger term) given precisely as a cyclic 1-cocyle of the form (\ref{eq:Schwinger-term}). Now, the relevance of this fact in the context of quantum phase transitions is that there are already
some geometric characterizations of the critical point in terms of, \textit{e.g.}, Berry phases~\cite{Carollo2005,Zhu2006}. This has  also been related to the behavior of certain Chern numbers associated to the parameter space of the spin chain~\cite{Contreras2008}. Now, using Araki's self-dual formalism, it should also be possible to compute the cocycle outside the critical point. The behavior of the cocycle as a function of the model's external parameters may provide
a new geometric characterization of the critical point.

\bibliographystyle{alpha} 

\end{document}